\documentclass[11pt]{article}
\usepackage[utf8]{inputenc}
\usepackage[T1]{fontenc}
\usepackage{graphicx}
\usepackage{subcaption}
\usepackage{booktabs}
\usepackage{setspace}
\usepackage{colortbl}
\usepackage{xcolor}
\usepackage[a4paper, inner=3cm,outer=3cm,bottom=3cm]{geometry}
\usepackage{sectsty,textcase}
\sectionfont{\MakeTextUppercase}
\makeatletter
\usepackage{hhline}
\usepackage[
backend=biber,
uniquename=false,
isbn=false,
doi=false,
style=apa,
sorting=nyt
]{biblatex}

\title{A new GARCH model with a deterministic time-varying intercept\thanks{Alexander Back gratefully acknowledges financial support from the OP Research Foundation (grants No. 20220141, 20210112 and 20200213). Material from this paper has been presented at the Graduate School of Finance Summer Workshop in Finance, Helsinki, June 2021; 15th International Conference on Computational and Financial Econometrics (CFE 2021), London, December 2021; 6th International Workshop on Financial Markets and Nonlinear Dynamics (FMND), Paris, June 2022; 4th Quantitative Finance and Financial Econometrics International Conference (QFFE 2022), Marseilles, June 2022; 24th International Conference on Computational Statistics (COMPSTAT 2022), Bologna, August 2022; 29th Nordic Conference in Mathematical Statistics (NORDSTAT 2023), Gothenburg, June 2023; 75th Econometric Society European Meetings (ESEM2023), Barcelona, August 2023. Comments from participants at these conferences are gratefully acknowledged. We also wish to thank Yongmiao Hong, Matthijs Lof  and Esther Ruiz for many thoughtful comments.}}
\author{Niklas Ahlgren$^1$, Alexander Back$^1$ \& Timo Teräsvirta$^{2}$}
\date{%
    $^1$Hanken School of Economics\\
    $^2$Aarhus BSS, Aarhus University
}
\usepackage{amsmath}
\usepackage{amsfonts}
\usepackage{amsthm}
\usepackage{amssymb}
\usepackage{bm}
\newtheorem{thm}{Theorem}
\newtheorem*{thm*}{Theorem}
\newtheorem{rmk}{Remark}
\newtheorem*{asmpt*}{Assumption}
\newtheorem{lemma}{Lemma}
\newcommand\norm[1]{\left\lVert#1\right\rVert}
\providecommand{\keywords}[1]
{
  \small	
  \textbf{\textit{Keywords--}} #1
}
\providecommand{\JEL}[1]
{
  \small	
  \textbf{\textit{JEL Classification Codes--}} #1
}

\begin{document}

\maketitle
\begin{abstract}
\footnotesize
It is common for long financial time series to exhibit gradual change in the unconditional volatility. We propose a new model that captures this type of nonstationarity in a parsimonious way. The model augments the volatility equation of a standard GARCH model by a deterministic time-varying intercept. It captures structural change that slowly affects the amplitude of a time series while keeping the short-run dynamics constant. We parameterize the intercept as a linear combination of logistic transition functions. We show that the model can be derived from a multiplicative decomposition of volatility and preserves the financial motivation of variance decomposition. We use the theory of locally stationary processes to  show that the quasi maximum likelihood estimator (QMLE) of the parameters of the model is consistent and asymptotically normally distributed. We examine the quality of the asymptotic approximation in a small simulation study. An empirical application to Oracle Corporation stock returns demonstrates the usefulness of the model. We find that the persistence implied by the GARCH parameter estimates is reduced by including a time-varying intercept in the volatility equation.\bigskip \\
\noindent \keywords{Locally stationary GARCH, Nonlinear time series, Quasi maximum likelihood estimator, Time-varying GARCH, Smooth transition, Volatility modelling}\\
\JEL{C22, C51, C58}
\end{abstract}
\thispagestyle{empty}
\clearpage

\section{Introduction}
\setcounter{page}{1}
It is common for long financial time series to exhibit gradual change in the unconditional volatility. Failure to model such change can lead to an overstatement of the persistence of volatility (\cite{mikoschstarica2004a}, \cite{mikoschstarica2004b}, \cite{staricagranger2005}). In this paper, we propose a new model that captures this type of nonstationarity in a parsimonious way. The most popular volatility models are the autoregressive conditional heteroskedasticity (ARCH) model by \textcite{Engle82} and its generalization, the generalized ARCH (GARCH) model by \textcite{Bollerslev86} and \textcite[pp. 78--79]{Taylor86}. A major reason for their popularity is that they are easy to fit and likelihood-based inference can be done. This paper extends the GARCH model to accommodate nonstationarity in a way that preserves the parametric likelihood framework. The model, called the additive time-varying (ATV-)GARCH model, augments the volatility equation by a deterministic time-varying intercept. The model combines the flexibility of the time-varying GARCH (tvGARCH) class of models (see \cite{ChenH16}) with the financial motivation of multiplicative decomposition of volatility (see \cite{Engle&Rangel08}). 

In the nonparametric literature several papers consider time-varying GARCH models. \textcite{DahlhausSR06} generalized the ARCH model to a nonstationary ARCH model with time-varying coefficients. The time-varying ARCH (tvARCH) process can be locally approximated by a stationary ARCH process. Such processes are referred to as locally stationary. The locally stationary tvARCH process was further extended to tvGARCH; see \textcite{rao2006}, \textcite{Rohan13}, \textcite{ChenH16}, \textcite{kristensen2019local} and \textcite{KarmakarRichterWu}. \textcite{Truquet17} introduced a semiparametric ARCH model where some, but not all, parameters are time-varying. In statistical tests for detecting non-time-varying parameters in models fitted on real data, he finds that the null hypothesis tends to be rejected for the constant but not the ARCH parameters.  

We show that the entirely parametric ATV-GARCH model admits a representation as a time-varying GARCH process. 
Many asymptotic results are available for time-varying GARCH processes. \textcite{dahlhaus2019towards}, henceforth DRW, developed a general theory for nonlinear locally stationary processes. In this class of models the parameters are Lipschitz continuous parameter curves. DRW proved a global law of large numbers and a global central limit theorem based on the stationary approximation. We apply these results to the ATV-GARCH model.  

Time-varying GARCH models have not been much used to model financial data. Instead, multiplicative decomposition of volatility into a stationary GARCH component and a financially motivated slowly moving component has received a lot of attention in the literature. The first papers examining the multiplicative decomposition were \textcite{Feng04} and \textcite{BellegemS04}, followed by \textcite{Engle&Rangel08}, \textcite{fengmcneil2008}, \textcite{BrownleesG10}, \textcite{MishraSU10}, \textcite{MazurPipien2012}, \textcite{AmadoT13} and \textcite{lintonkoo2015}. For surveys, see \textcite{VBellegem12} and \textcite{Terasvirta12}. For a more recent model with a multiplicative structure, see \textcite{ES2018}. We show that a particular choice of the multiplicative decomposition yields a model that can be asymptotically expressed as an ATV-GARCH model. The ATV-GARCH model preserves the financial motivation of variance decomposition while the framework of tvGARCH processes can be utilized to derive asymptotic results. The ATV specification can be viewed as a parsimoniously parameterized alternative to the very general tvGARCH by \textcite{rao2006}.  

In the ATV-GARCH model, volatility is mean-reverting towards a time-varying mean, but with constant persistence. This way,
the model captures structural change that slowly affects the amplitude of a time series while keeping the short-run dynamics constant over time. The model is particularly well suited for situations in which the volatility of an asset or index is smoothly increasing or decreasing over time. The idea is that short-run fluctuations in volatility are stationary, but long-run structural changes make the assumption of stationarity inappropriate. This view is complementary to the interpretation by \textcite{Diebold86} that a neglected level shift in the intercept of a GARCH model can overstate persistence and erroneously leads one to an integrated GARCH by \textcite{EngleB86}.  \textcite{LamoureuxL90} investigate Diebold’s conjecture in finite samples. \textcite{hillebrand2005} proves that the sum of the estimated autoregressive parameters of the
conditional variance converges to one if a time series contains parameter changes in the conditional volatility. See also \textcite{mikoschstarica2004a}, who argue that assuming a constant unconditional variance can lead to spurious Integrated GARCH (IGARCH) models. A similar comment applies to the fractionally integrated GARCH (FIGARCH) process introduced by \textcite{baillie1996}. The FIGARCH model features a parameter that measures the degree of fractional integration of the volatility process. Neglecting structural breaks can cause overestimation of the parameter and in consequence an overestimation of persistence (see \cite{hillebrand2005} and the references therein). It is therefore possible that allowing for a model that incorporates changes in the intercept can make the evidence of long memory disappear.
By
making the intercept a smooth function of time, it is possible to capture level shifts that occur
gradually, rather than abruptly as in volatility models with structural breaks; see e.g. \textcite{chu1995} and \textcite{smith2008}. 

The ATV-GARCH model is globally nonstationary but locally stationary. The parameters of the model are estimated globally by QMLE. The primary goal of this paper is to establish consistency and asymptotic normality of the QMLE of the parameters of the ATV-GARCH model based on the stationary approximation. More precisely, we make use of the theory of locally stationary processes in DRW. \textcite{BHK2003}, henceforth BHK, and \textcite{FrancqZakoian2004} proved consistency and asymptotic normality of the QMLE of the parameters of strictly stationary GARCH processes under mild conditions. For earlier developments, see \textcite{leehansen1994} and \textcite{lumsdaine1996}. We use the theory in DRW to extend the results to a nonstationary setting. To the best of our knowledge, this paper is the first one to give rigorous proofs of consistency and asymptotic normality in a parametric GARCH model with a deterministic time-varying intercept.

Our model is related to the GARCH-X model of \textcite{HanPark2012} and \textcite{HanK14}, the FIGARCH model with a time-varying intercept of \textcite{baillie2009} and the smoothly time-varying parameter GARCH model of \textcite{ChenGL14}, among others.
In the GARCH-X model there is an additive component in the variance equation, but instead of being smooth and deterministic it is a positive-valued function of an exogenous stochastic variable. It may be viewed as a special case of the functional coefficient GARCH model by \textcite{MedeirosW09}. \textcite{baillie2009} proposed a FIGARCH model with a time-varying intercept given by a trigonometric function. However, the authors do not provide results on the consistency and asymptotic normality of the QMLE of the parameters. \textcite{ChenGL14} proposed a model in which all parameters are time-varying in a way similar to the ATV-GARCH model, but estimation and inference are Bayesian. \textcite{ambrovzevivciute2008tvgarch} study the statistical properties of the tvGARCH$(1,1)$ model with logistic coefficients, but do not consider inference. 

The plan of the paper is as follows. In Section 2 we introduce the ATV-GARCH model and show that it can be locally approximated by a stationary GARCH process. Maximum likelihood estimation of the model is discussed in Section 3. We prove consistency and asymptotic normality of the QMLE of the parameters based on the stationary approximation. In Section 4 we provide a simulation study of the QMLE of the parameters. In Section 5 we demonstrate the usefulness of the model by an empirical application to Oracle Corporation stock returns. Section 6 concludes. Proofs can be found in the Appendix.

Throughout the paper, we use $\norm{X}_p=(\mathbb{E}|X|^p)^{1/p}$ to denote the norm of a random variable $X$, and when applied to a vector or matrix, we use $|\cdot|$ to denote the maximum norm. We use $C_1,C_2,\ldots$ as generic constants, not necessarily the same across contexts.

\section{The model}
In this section, we define the ATV-GARCH model and show that 
a particular choice of multiplicative decomposition yields a model that is asymptotically equivalent to an ATV-GARCH model. We show that the ATV-GARCH model belongs to the class of locally stationary time-varying GARCH processes. We discuss some technical results related to the moment assumptions required in the proofs of consistency and asymptotic normality of the QMLE of the parameters.

\subsection{The additive time-varying GARCH model}
The ATV-GARCH model is defined by augmenting the volatility equation of the GARCH model of Bollerslev (1986) and \textcite[pp. 78--79]{Taylor86} by a deterministic time-varying intercept:
\begin{equation}
	X_{t,T}=\sigma _{t,T}\varepsilon_{t}, \ t=1,\ldots,T,
	\label{eq:GARCH}
\end{equation}
where $\varepsilon_{t}$ is IID$(0,1)$, while the volatility equation is given by
\begin{equation}
\sigma^2_{t,T} = \alpha_0(t/T;\boldsymbol{\theta}) + \alpha_1X^2_{t-1,T} + \beta_1 \sigma^2_{t-1,T}. \label{eq:atvgarch}
\end{equation}
In applications, $X_{t,T}$ will typically be a log-return of a financial asset or an index. To keep the notation simple,  the volatility equation (\ref{eq:atvgarch}) is of order one. The time-varying intercept is defined as $\alpha_0(t/T;\boldsymbol{\theta}):= \alpha_0 + g(t/T, \boldsymbol{\theta})$, where $g(t/T, \boldsymbol{\theta})$ is a Lipschitz continuous function in the parameters $\boldsymbol{\theta}$; the details will follow. The double subscript $(t,T)$ is used to emphasize that we are working with triangular arrays in rescaled time.

Interestingly, the ATV-GARCH model can be derived from a multiplicative decomposition of volatility. Such a decomposition typically contains two parts: a slowly changing component that captures some feature of the data-generation process, and a transient component that generates the daily volatility. A general multiplicative decomposition can be stated by considering (\ref{eq:GARCH}) with 
\begin{align}
    &\sigma^2_{t,T} = w_{t,T}g_{t,T}, \nonumber \\ 
    &w_{t,T} = \alpha_0 + \alpha_1X^2_{t-1,T}/g_{t-1,T} + \beta_1w_{t-1,T}, \label{eq:multvar}
\end{align}
where $g_{t,T}$ is some function capturing the slowly changing component. The function $g_{t,T}$ has been motivated by financial arguments. This is the case in e.g. \textcite{ES2018}, where the authors posit that the slowly changing component stems from a flexible leverage multiplier. An alternative specification is to make $g_{t,T}$ a function of rescaled time $t/T$, $g_{t,T}=g(t/T)$, as in \textcite{AmadoT13}. Here, we take the latter route and further require $g_{t,T}$ to be Lipschitz continuous. To derive the ATV-GARCH model, rewrite (\ref{eq:multvar}) as follows: $$ \sigma^2_{t,T} =\alpha_0g_{t,T} + \alpha_1\frac{X^2_{t-1,T}g_{t,T}}{g_{t-1,T}} + \beta_1\frac{\sigma^2_{t-1,T}g_{t,T}}{g_{t-1,T}}.$$ Note that by Lipschitz continuity, we have $$\left|g(t/T)-g((t-1)/T)\right| \leq C\left|t/T-(t-1)/T\right|=C/T,$$ for some positive constant $C,$ so as $T\to\infty,$ $g_{t,T}/g_{t-1,T}\to 1.$ Asymptotically, the conditional variance in the multiplicative decomposition is equal to\begin{align*} \sigma^2_{t,T} &=\alpha_0g_{t,T} + \alpha_1X^2_{t-1,T} + \beta_1\sigma^2_{t-1,T}.\\
\end{align*} 
Now, by writing $$g_{t,T} = 1 +\frac{1}{\alpha_0}g\left(t/T; \boldsymbol{\theta}\right)$$ for some parametric function $g(\cdot),$ we obtain the specification (\ref{eq:atvgarch}). See \textcite{Truquet17} for a similar development in the ARCH context. Consequently, financial insights originally used to motivate the multiplicative model (\ref{eq:multvar}) apply equally to the ATV-GARCH model. \textcite{AmadoT13} also briefly discuss an \emph{additive decomposition} of volatility, but do not consider it further.

A computational advantage of the ATV-GARCH model is that the parameters can be estimated by QML in a single step. \textcite{AmadoT13} estimated the parameters of the multiplicative decomposition (\ref{eq:multvar}) by an iterative procedure 
called estimation by parts, which is detailed in \textcite{SongFK05}. Further, the asymptotic theory of QML estimation of the ATV-GARCH model relies on assumptions that are straightforward to verify and comparable to the ones in strictly stationary GARCH processes.

The observations are now assumed to come from an ATV-GARCH$(p,q)$ process with volatility equation
\begin{equation}
	\sigma_{t,T}^{2}= \alpha_{0}+g(t/T;\boldsymbol{\theta}_1)
    +\sum_{i=1}^{p}\alpha_{i}X_{t-i,T}^{2}+\sum_{j=1}^{q}\beta_{j}\sigma_{t-j,T}^{2},\label{eq:themodel}
\end{equation}
where $\alpha_0>0$, $\alpha_1, \alpha_2, \ldots, \alpha_p\geq0$, $\beta_1, \beta_2, \ldots, \beta_q\geq0$ and $g(t/T;\boldsymbol{\theta}_1 )>0$. We follow \textcite{AmadoT13} and parameterize $g(\cdot)$ as a linear combination of logistic transition functions
\begin{equation}
	g(t/T;\boldsymbol{\theta}_1 ):=\sum_{l=1}^{L}\alpha_{0l}G_{l}\left(t/T;\gamma_{l},\boldsymbol{c}_{l}\right),\label{eq:tvintercept}
\end{equation}
where
\begin{equation}
	G(t/T;\gamma,\boldsymbol{c})=\left(1+\exp\left\{ -\gamma\prod\limits _{k=1}^{K}\left(t/T-c_{k}\right)\right\} \right)^{-1} \label{eq:lfunct}
\end{equation}
is the general logistic function with $\gamma>0,\ c_{1}<\ldots<c_{K}.$ When $L>1$,  the parameters satisfy $\gamma_{l}>0,l=1,\ldots,L,\ c_{11}<c_{12}<\ldots<c_{1K}<c_{21}<\ldots<c_{LK}$.
We use the partition $\boldsymbol{\theta} = (\boldsymbol{\theta}^{\intercal}_1, \boldsymbol{\theta}^{\intercal}_2)^{\intercal}$, where $\boldsymbol{\theta}_1 = (\alpha_{01}, \ldots, \alpha_{0L}, \gamma_1, \ldots, \gamma_L, \boldsymbol{c}_{1}^{\intercal}, \ldots, \boldsymbol{c}_{L}^{\intercal})^{\intercal}$ and $\boldsymbol{\theta}_2 = (\alpha_0,\alpha_1, \ldots, \alpha_p, \beta_1, \ldots, \beta_q)^{\intercal}$.
The ATV-GARCH model is unidentified if in (\ref{eq:tvintercept}) at least one $\alpha_{0l}=0$, $l=1, \ldots, L$. A Lagrange multiplier (LM) test of the null hypothesis $\alpha_{01}=...=\alpha_{0L}=0$, based on approximating the alternative by a Taylor expansion around the null hypothesis as in \textcite{LST88b}, is presented in \textcite{ABT2023b}. This circumvents the identification problem that invalidates the standard asymptotic inference; see e.g. \textcite{Davies77} who first pointed out this problem and discussed a solution to it.

When $L=1$ in (\ref{eq:tvintercept}) and $K=1$ in (\ref{eq:lfunct}), the change in the unconditional variance is monotonic. Nonmonotonic change can be achieved by choosing $K>1$ or by setting $L>1$ in (\ref{eq:tvintercept}). Typical transition functions are illustrated in Figure \ref{fig:graphlfunct}. Panels (a) and (b) show how the logistic transition function behaves as functions of the shape parameter $\gamma$ and location parameter $c$. Panel (c) depicts a double increase in volatility and panel (d) illustrates a case where volatility first increases and then decreases. As shown by \textcite{hornik1989}, continuous functions on compact subsets of $\mathbb{R}$ can be uniformly approximated as closely as desired by linear combinations of logistic functions. This indicates the flexibility of the function $g(t/T;\boldsymbol{\theta}_1 )$.
\begin{figure}[h]
    \centering
    \caption{The logistic transition function}
    \label{fig:graphlfunct}
    \begin{subfigure}{0.5\textwidth}
    \centering
        \includegraphics[width=0.9\linewidth]{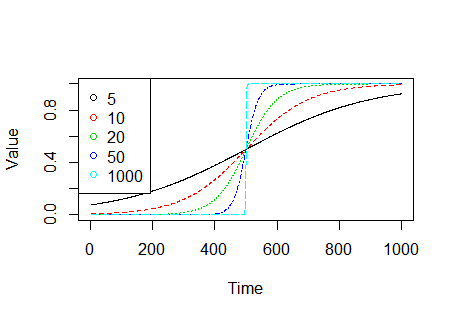}
        \caption{$\gamma_{1}$ varies, $c_{1}$ is fixed at 0.5.}
        \label{fig:changegamma}
    \end{subfigure}%
    \begin{subfigure}{0.5\textwidth}
    \centering
        \includegraphics[width=0.9\linewidth]{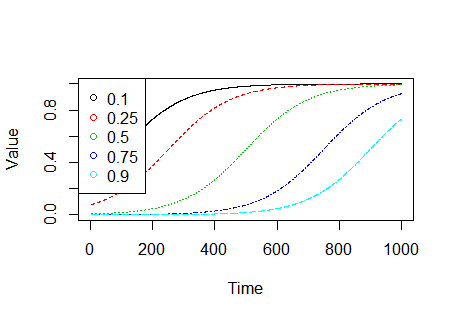}
        \caption{$c_{1}$ varies, $\gamma_{1}$ is fixed at 10.}
 \label{fig:changec}
    \end{subfigure}%
    \\
    \begin{subfigure}{0.5\textwidth}
    \centering
        \includegraphics[width=0.9\linewidth]{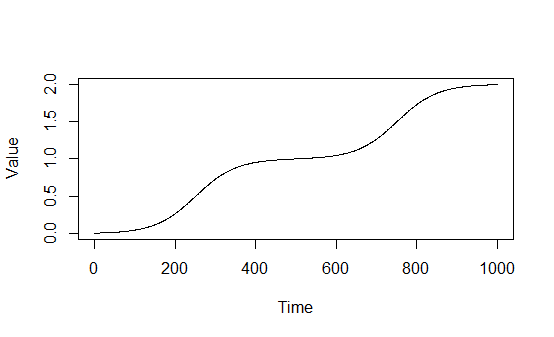}
        \caption{Two transitions with $\gamma_1 = \gamma_2 = 20$, $c_1 = 0.25$, $c_2 = 0.75$, no sign change.}
        \label{fig:2tsamedir}
    \end{subfigure}%
    \begin{subfigure}{0.5\textwidth}
    \centering
        \includegraphics[width=0.9\linewidth]{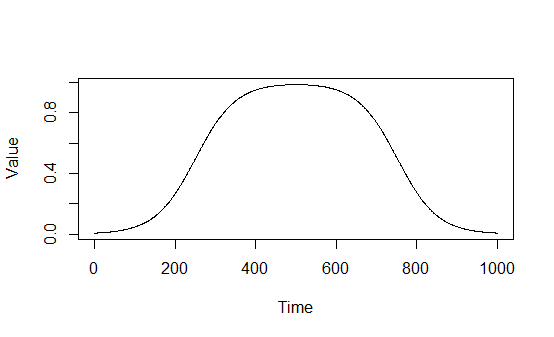}
        \caption{Two transitions with $\gamma_1 = \gamma_2 = 20$, $c_1 = 0.25$, $c_2 = 0.75$, one sign change.}
 \label{fig:2toppdir}
    \end{subfigure}%
   \end{figure}

The GARCH equation (\ref{eq:themodel}) can easily be made asymmetric. The simplest way to do so is to add an indicator variable for negative returns as in the GJR-GARCH model by \textcite{GlostenJR93}. For notational simplicity we retain the form (\ref{eq:themodel}). 

\subsection{Local stationarity}
Since the intercept in (\ref{eq:themodel}) is deterministically time-varying, the process $X_{t,T}$ is nonstationary. Standard asymptotic results for stationary and ergodic processes do not apply to the ATV-GARCH model. However, the rescaling device $t/T$ enables a meaningful asymptotic theory for processes that can be locally approximated by stationary processes.

This rescaling device is used by \textcite{dahlhaus1997, dahlhaus2000} and \textcite{DahlhausSR06} in the definition of local stationarity. In the following we work with $X_{t,T}^{2}$ instead of $X_{t,T}$ (cf. \textcite{DahlhausSR06}). By the triangle inequality, decompose the difference between $X_{t,T}^{2}$ and the stationary approximation $\widetilde{X}_t^{2}(u)$ at $u \in [0,1]$ as \begin{equation}
\left|X_{t,T}^{2}-\widetilde{X}_t^{2}(u)\right|\leq\left|X_{t,T}^{2}-\widetilde{X}_t^{2}(t/T)\right| + \left|\widetilde{X}_t^{2}(t/T)-\widetilde{X}_t^{2}(u)\right|.\label{eq:trianglels}
\end{equation} It is seen from (\ref{eq:trianglels}) that if $t/T$ is close to $u$, then $X_{t,T}^{2}$ and $%
\widetilde{X}_{t}^{2}(u)$ should be close and that the degree of the approximation
depends on the rescaling factor $T$ and the deviation $|t/T-u|$. The process $\{X^2_{t,T}\}$ is said to be locally stationary if (Dahlhaus and Subba Rao 2006) 
\begin{equation}
X_{t,T}^{2}= \widetilde{X}_t^{2}(u) + O_P\left(\left|\frac{t}{T}-u\right| + \frac{1}{T}\right), \label{eq:localstatdef}
\end{equation}
where $u\in[0,1]$ and $\widetilde{X}^2_t(u)$ is the stationary approximation at $u$. 

\textcite{rao2006} considers the following class of time-varying GARCH (tvGARCH) processes:
\begin{align}
X_{t,T} &= \sigma _{t,T} \varepsilon _{t}, \label{eq:tvgarch1} \\
\sigma _{t,T}^{2} &=\alpha _{0}\left( \frac{t}{T}\right)
+\sum_{i=1}^{p}\alpha _{i}\left( \frac{t}{T}\right)
X_{t-i,T}^{2}+\sum_{j=1}^{q}\beta _{j}\left( \frac{t}{T}\right) \sigma
_{t-j,T}^{2}, \label{eq:tvgarch2}
\end{align}
where the parameters $\alpha _{0}(t/T)$, $\alpha _{i}(t/T)$, $i=1,\ldots ,p$%
, and $\beta _{j}(t/T)$, $j=1\ldots ,q$, are smooth functions of time. The stationary approximation $\widetilde{X}_t(u)$ of the tvGARCH process at $u$ is given by 
\begin{align*}
\widetilde{X}_{t}(u) &=\sigma _{t}(u) \varepsilon _{t}, \\
\sigma^{2}_{t}(u) &=\alpha _{0}\left( u\right) +\sum_{i=1}^{p}\alpha
_{i}\left( u\right) \widetilde{X}^{2}_{t-1}(u)+\sum_{j=1}^{q}\beta
_{j}\left(u\right) \sigma^{2}_{t-j}(u).
\end{align*}%
Theorem 2.1 of \textcite{rao2006} states conditions under which a nonstationary, nonlinear process with time-dependent parameters can be locally approximated by a stationary process. Applied to the tvGARCH process, the conditions are as follows (cf. Subba Rao 2006, Section 5.1).
\begin{itemize}
    
\item[(i)] The parameter curves $\{\alpha_0(\cdot)\}$, $\{\alpha_i(\cdot)\}$ and $\{\beta_j(\cdot)\}$ are Lipschitz continuous.

\item[(ii)] \begin{equation*}
\sup_{u}\left\{ \sum_{i=1}^{p}\alpha _{i}\left( u\right)
+\sum_{j=1}^{q}\beta _{j}\left( u\right) \right\} <1-\eta
\end{equation*}
for some $\eta >0$.

\item[(iii)] $\mathbb{E}(\varepsilon_{t}^{2}) = 1$.
\end{itemize}
We now return to the ATV-GARCH$(p,q)$-version of the model in (\ref{eq:GARCH}) and (\ref{eq:atvgarch}) and show that it is locally stationary. The ATV-GARCH model is a parsimonious parameterization of the more general tvGARCH process (\ref{eq:tvgarch1}) and (\ref{eq:tvgarch2}), with $$\alpha
_{0}(t/T; \boldsymbol{\theta}):=\alpha_{0}+g(t/T; \boldsymbol{\theta}_1),$$ where $g(t/T; \boldsymbol{\theta}_1)$ is defined in (\ref{eq:tvintercept}) and (\ref{eq:lfunct}), $%
\alpha _{i}(t/T)=\alpha _{i}$, $i=1,\ldots ,p$, and $\beta _{j}(t/T)=\beta
_{j}$, $j=1,\ldots ,q$. The intercept and unconditional variance are time-varying, but the persistence is constant. The stationary approximation is obtained by fixing the intercept at the value the function $g(\cdot; \boldsymbol{\theta}_1)$ takes at $u$. 
The ATV-GARCH process is locally stationary and we have that \begin{equation}X_{t,T}^{2}=\widetilde{X}^{2}_{t}(u)+\left( \left\vert \frac{t}{T}%
-u\right\vert +\frac{1}{T}\right) R_{t,T},\quad \text{where\quad }\sup_{t,T}%
\mathbb{E}(R_{t,T})<\infty \label{eq:SRAOlocalstat}.\end{equation}
We state the result as a proposition.

\newtheorem{prop}{Proposition}
\begin{prop}
Assume that the parameter space $\varTheta$ is compact, $\sum_{i=1}^{p}\alpha _{i}
+\sum_{j=1}^{q}\beta _{j}<1$ and $\mathbb{E}(\varepsilon _{t}^{2})=1$. Then the ATV-GARCH model is locally stationary. \end{prop}

\begin{proof} Conditions (ii) and (iii) are satisfied by assumption. For (i), see Appendix A.3.
\end{proof}

\begin{rmk} 
The condition on the GARCH coefficients is a necessary condition for weak stationarity of the GARCH$(p,q)$\
process. For estimation by QML, it is a stronger assumption than needed for strictly stationary GARCH processes, where a weaker condition on the coefficients can be obtained in terms of the top Lyapunov exponent; see BHK and \textcite{FrancqZakoian2004}.
\end{rmk}
\subsection{Moment assumptions and the stationary approximation}

In order to prove consistency and asymptotic normality of the QMLE of the parameters of the ATV-GARCH model, we require assumptions on the moment structure of both the process $\{X_{t,T}\}$ and the errors $\{\varepsilon_t\}$. The GARCH$(p,q)$ process has a finite second moment under the condition $\sum_{i=1}^{p}\alpha _{i}
+\sum_{j=1}^{q}\beta _{j}<1$ coupled with the existence of a second moment of the error term. To prove asymptotic normality of the QMLE in strictly stationary GARCH processes, it suffices to assume a finite fourth moment of the errors; see BHK and \textcite{FrancqZakoian2004}. For the ATV-GARCH model, we also have to assume that the process possesses a finite fourth moment. To explain why, we require some additional results from \textcite{rao2006}. Following Subba Rao (2006), the tvGARCH$(p,q)$ process $\{X_{t,T}\}$ admits the state
space representation (assume without loss of generality that $p,q\geq 2$)%
\begin{equation}
\mathcal{X}_{t,T}=\mathbf{b}_{t}\left( \frac{t}{T}\right) +\mathbf{A}%
_{t}\left( \frac{t}{T}\right) \mathcal{X}_{t-1,T}
\label{eq:GARCHSS}
\end{equation}%
with
\begin{equation*}
\mathcal{X}_{t,T}=\left(\sigma _{t,T}^{2}, \ldots,  
\sigma _{t-q+1,T}^{2}, X_{t-1,T}^{2}, \ldots, X_{t-p+1,T}^{2}
\right)^\mathsf{T} \in \mathbb{R}^{p+q-1},
\end{equation*}
\begin{equation*}
\mathbf{b}_{t}(u)=\left(\alpha _{0}(u), 0, \ldots, 0 
\right)^\mathsf{T} \in \mathbb{R}^{p+q-1},\quad 
\end{equation*}
and
\begin{equation*}
\mathbf{A}_{t}(u)=\left( 
\begin{array}{cccc}
\boldsymbol{\tau }_{t}(u) & \beta _{q}(u) & \boldsymbol{\alpha }(u) & \alpha _{p}(u)
\\ 
\mathbf{I}_{q-1} & \mathbf{0} & \mathbf{0} & \mathbf{0} \\ 
\mathbf{z}_{t-1}^{2} & 0 & \mathbf{0} & 0 \\ 
\mathbf{0} & \mathbf{0} & \mathbf{I}_{p-2} & \mathbf{0}%
\end{array}%
\right),
\label{eq:Atu}
\end{equation*}%
a $(p+q-1)\times (p+q-1)$ matrix where $\boldsymbol{\tau }_{t}(u)=(\beta
_{1}(u)+\alpha _{1}(u)\epsilon_{t-1}^{2},\beta _{2}(u),\ldots ,\beta _{q-1}(u))$, $%
\boldsymbol{\alpha }(u)=(\alpha _{2}(u),\ldots ,\alpha _{p-1}(u))$ and $\mathbf{z}_{t-1}^{2}=(\varepsilon _{t-1}^{2},0,\ldots ,0)\in \mathbb{R}^{q-1}$.
The stationary process $\{\widetilde{\mathcal{X}}_t(u)\}$ at $u$ is given by
\begin{equation}
\widetilde{\mathcal{X}}_t(u)=\mathbf{b}_{t}(u) 
+\mathbf{A}_t(u) \widetilde{\mathcal{X}}_{t-1}(u).
\label{eq:ssstat}
\end{equation}

We now show that the ATV-GARCH process admits the representation (\ref{eq:GARCHSS}), and that the results in Subba Rao (2006) apply to it. The vectors $\mathbf{b}_{t}(u)$ are defined as in (\ref{eq:GARCHSS}) with $\alpha_0 + g(u)$ in place of $\alpha_0(u)$. The matrices $\mathbf{A}_t(u)$ are matrices with constant parameters. By Lipschitz continuity of $g(u)$, it follows from Subba Rao (2006, Theorem 2.1) that
\begin{equation}
	\left|\mathcal{X}_t-\widetilde{\mathcal{X}}_t(u)\right|\leq\left|\frac{t}{T}-u\right|W_t + \frac{1}{T}V_{t,T}, \label{eq:lipstochprocesses}
\end{equation}
where $W_t$ and $V_{t,T}$ are stochastic processes.
For asymptotic results for locally stationary processes, conditions of the type
\begin{equation}
	\norm{\mathcal{X}_t-\widetilde{\mathcal{X}}_t(u)}_n \leq \left|\frac{t}{T}-u\right|C_1 + \frac{1}{T}C_2
	\label{eq:statespaceS1}
\end{equation}
are needed (see Appendices A.1 and A.2). For consistency, we need (\ref{eq:statespaceS1}) to be satisfied with $n=1$ and for asymptotic normality with $n=2$. It is obvious that (\ref{eq:statespaceS1}) requires the existence of moments of $W_t$ and $V_{t,T}$, to which we turn next.  

Define the stationary sequence
$$
\widetilde{\mathbf{b}}_{t}=\left(	\sup_{u\in[0,1]}\alpha_0(u), 0, \ldots, 0 \\
\right)^\mathsf{T} 
\in \mathbb{R}^{p+q-1}.
$$
The quantity $\widetilde{\mathbf{b}}_{t}$ is deterministic and bounded. The matrices included in Assumption 2.1 of Subba Rao that bound $\mathbf{A}_t(u)$ can be taken to be $\mathbf{A}_t(u)$ in (\ref{eq:GARCHSS}) with constant parameters, denoted here $\mathbf{A}_t$. Define the matrix  $(\mathbf{A})_n = \left(\mathbb{E}\left|A_{t, ij}\right|^n\right)^{1/n}$ and let $\lambda_{\text{spec}}(\mathbf{A})$ denote the largest absolute eigenvalue of $\mathbf{A}$. By Proposition 2.1 of Subba Rao (2006), if

(i) $\norm{\widetilde{\mathbf{b}_t}}_n^n<\infty$

(ii) $\lambda_{\text{spec}}((\mathbf{A})_n)<1-\delta$ for some $\delta>0$, then
\begin{equation}
	\sup_{t,T}\norm{W_t}_n^n <\infty.
\end{equation}
and
\begin{equation}
	\sup_{t,T}\norm{V_{t,T}}_n^n <\infty.
\end{equation}
By boundedness of the intercept, (i) is fulfilled. For asymptotic normality, we have to assume that (ii) holds with $n=2$. Because the matrices $\mathbf{A}_{t}$ are matrices with constant parameters, (ii) with $n=2$ is equivalent to the necessary and sufficient condition for the existence of a fourth moment of the GARCH process in \textcite{he1999}, and \textcite{LingMcaleer}. To see this, consider the ATV-GARCH$(1,1)$ model with one transition function. The state space representation is given by (\ref{eq:GARCHSS}) with 
\begin{align*}
\mathbf{b}_{t}(u) &=\left(\alpha _{0}+\alpha _{01}G(u), 0, 0
\right)^\mathsf{T},  \\
\mathbf{A}_{t} &=\left( 
\begin{array}{ccc}
\beta _{1}+\alpha _{1}\varepsilon_{t-1}^{2} & 0 & 0 \\ 
1 & 0 & 0 \\ 
\varepsilon_{t-1}^{2} & 0 & 0%
\end{array}%
\right).
\end{align*}%
We have
\begin{align*}
\widetilde{\mathbf{b}}_{t}=\sup_{u}|\mathbf{b}_{t}(u)| &\leq \left(\alpha _{0}+\alpha _{01}, 0, 0 
\right)^\mathsf{T} , \\
(\mathbf{A})_{2} &=\left( 
\begin{array}{ccc}
\{\mathbb{E}(\beta _{1}+\alpha _{1}\varepsilon_{t-1}^{2})^{2}\}^{1/2} & 0 & 0 \\ 
1 & 0 & 0 \\ 
\{\mathbb{E}(\varepsilon_{t-1}^{2})^{2}\}^{1/2} & 0 & 0%
\end{array}%
\right) \\
&=\left( 
\begin{array}{ccc}
\{\beta_{1}^{2}+2\alpha _{1}\beta _{1}\mathbb{E}(\varepsilon_{t}^{2})+\alpha _{1}^{2}%
\mathbb{E}(\varepsilon_{t}^{4})\}^{1/2} & 0 & 0 \\ 
1 & 0 & 0 \\ 
\{\mathbb{E}(\varepsilon_{t}^{4})\}^{1/2} & 0 & 0%
\end{array}%
\right) .
\end{align*}%
Since the trace of a matrix is equal to the sum of its eigenvalues, the condition $\lambda_{\text{spec}}((\mathbf{A})_{2})<1-\delta $ translates
into%
\begin{equation}
\beta _{1}^{2}+2\alpha _{1}\beta _{1}\mathbb{E}(\varepsilon_{t}^{2})+\alpha _{1}^{2}%
\mathbb{E}(\varepsilon_{t}^{4})<1-\delta \label{eq:4thmoment},
\end{equation} 
which, less the term $\delta$, is the fourth moment condition in \textcite{he1999}. Figure 3 depicts the region of existence of the fourth moment for the ATV-GARCH$(1,1)$ model with Gaussian errors: $\mathbb{E}(\varepsilon_{t}^{2})=1$ and $\mathbb{E}(\varepsilon_{t}^{4})=3$ in (\ref{eq:4thmoment}). 
\begin{figure}[h]
    \centering
    \includegraphics[width=0.98\linewidth]{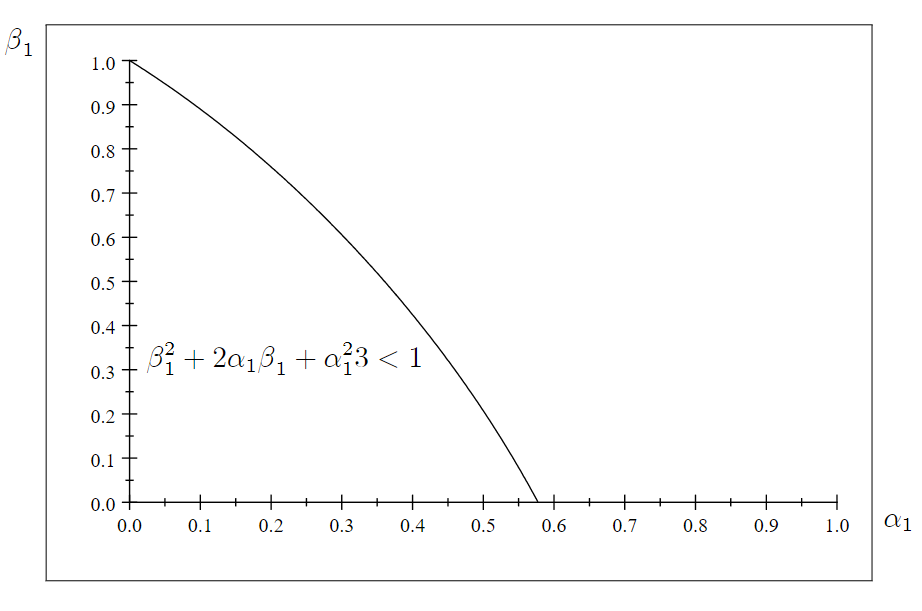}
    \caption{Existence of the fourth moment for the ATV-GARCH$(1,1)$ model with Gaussian errors.}
    \label{fig:4thmoment}
\end{figure}

\section{Maximum likelihood estimation}

In this section we consider estimation of the ATV-GARCH model by QML. Estimation is based on a representation of $\sigma^2_{t,T}$ in terms of past observations $X_{t-i,T}$ and the time-varying intercept $g_{t-i+1,T}$, $1 \leq i < \infty$. We use results from DRW for processes that can be locally approximated by stationary processes to derive the asymptotic properties of the QMLE.

\subsection{Representation for the conditional variance}

BHK derived a representation $h_{t,T}(\boldsymbol{\theta})$ for the conditional variance of the GARCH$(p,q)$ process in terms of past observations $X_{t-i,T}$. Our aim is to obtain a similar representation for the ATV-GARCH process in terms of $X_{t-i,T}$ and the time-varying intercept $g_{t-i+1,T}(\boldsymbol{\theta}_1)$, $1 \leq i < \infty$. The BHK representation was also used by \textcite{ChenH16} in their non-parametric tvGARCH model. First, we define a parameter space similar to the one in BHK. Let $\boldsymbol{\theta} = (\boldsymbol{\theta}_1^{\intercal},\boldsymbol{\theta}_2^{\intercal})^{\intercal}$ be the partitioning of the parameters in Section 2.1, where $\boldsymbol{\theta}_1 \subset \mathbb{R}^{\text{dim}(\boldsymbol{\theta}_1)}$ contains the parameters in the parameterisation of the time-varying intercept and $\boldsymbol{\theta}_2 \subset \mathbb{R}^{p+q+1}$ the remaining GARCH parameters. Let $0<\underline{\vartheta}<\overline{\vartheta}$, $0<\rho_0<1$ and $q\underline{\vartheta}<\rho_0$. Define \begin{equation}
\Theta = \{\boldsymbol{\theta}:\beta_1 + \beta_2 + \ldots + \beta_q\leq\rho_0\} \label{eq:parameterspace1}
\end{equation}
and
\begin{align}
&\{\underline{\vartheta}\leq \min(\alpha_0 + \sum_{l = 1}^L\alpha_{0l}, \alpha_1, \alpha_2, \ldots, \alpha_p, \beta_1, \beta_2, \ldots, \beta_q) \nonumber \\ &\leq \max(\alpha_0 + \sum_{l = 1}^L\alpha_{0l}, \alpha_1, \alpha_2, \ldots, \alpha_p, \beta_1, \beta_2, \ldots, \beta_q)\leq\overline{\vartheta}\}. \label{eq:parameterspace2}
\end{align}
The sequence $h_{t,T}(\boldsymbol{\theta})$ is computed from
\begin{align}
	h_{t,T}(\boldsymbol{\theta})=c_{0}(\boldsymbol{\theta})
	+ \sum^\infty_{i=1}d_i(\boldsymbol{\theta})g_{t-i+1,T}(\boldsymbol{\theta}_1) + \sum^{\infty}_{i=1}c_i(\boldsymbol{\theta})X^2_{t-i,T},
	\label{eq:decrecht}
\end{align}
where
\begin{align}
	c_{0}(\boldsymbol{\theta}) = 
	\frac{\alpha_0}{1-\sum^q_{j=1} \beta_j}
\end{align}
and $g_{t,T}(\boldsymbol{\theta}_1)$ is a short-hand notation for $g(t/T;\boldsymbol{\theta}_1)$.

The formulas for the coefficients $c_i(\boldsymbol{\theta})$ in (\ref{eq:decrecht}) are given in BHK. By Lemma 3.1 of BHK, formula (3.4), the coefficients $c_{i}(\boldsymbol{\theta})$ satisfy
\begin{align}
    c_{i}(\boldsymbol{\theta}) \leq C_2 \rho_0^{i/q}, \quad 0 \leq i < \infty, \label{eq:cineq}
\end{align}
for some constant $C_2$. The coefficients  
$d_{i}(\boldsymbol{\theta})$ obey a similar formula and (see Appendix A.3)
\begin{align}
    d_{i}(\boldsymbol{\theta}) \leq C_1 \rho_0^{i/q}, \quad 0 \leq i < \infty. \label{eq:dineq}
\end{align}

BHK proved that under the assumption $\mathbb{E}\ln\sigma^2_0<\infty$, the representation (\ref{eq:decrecht}) with $d_{i}(\boldsymbol{\theta})=0$, $1 \leq i < \infty$, yields $\sigma^2_{t,T}$ almost surely. In the ATV-GARCH model, the result follows similarly by considering quantities that bound $\sigma^2_{t,T}.$ Assume for notational simplicity the ATV-GARCH$(1,1)$ model. Then from (\ref{eq:themodel}),
\begin{align}
    \alpha_0 + g_{t,T}(\boldsymbol{\theta}_1) + \alpha_1X^2_{t-1,T} = \sigma^2_{t,T} - \beta_1\sigma^2_{t-1,T}.
    \label{repsigma2t}
\end{align}
Define the nonstationary sequence $\Phi_{t,T} = \alpha_0 + g_{t,T}(\boldsymbol{\theta}_1) + \alpha_1X^2_{t-1,T}$. Iterating (\ref{repsigma2t}), we obtain
\begin{align}
\Phi_{t,T} + \beta_1 \Phi_{t-1,T} + \cdots + \beta_1^i \Phi_{t-i,T} = \sigma_t^2 - \beta_1^{i+1} \sigma^2_{t-i-1,T}.
\label{itrepsigma2t}
\end{align}
Lemma 2.2 of BHK yields that the left-hand side of (\ref{itrepsigma2t}) with $g_{t,T}(\boldsymbol{\theta}_1) = 0$ in (\ref{repsigma2t}) converges almost surely as $i\to\infty$. Define $\Phi_{t}^* = \alpha_0 + \sup_{[0,1]}g_{t,T}(\boldsymbol{\theta}_1) + \alpha_1X^{*2}_{t-1}$, and similarly $\sigma^{*2}_{t}$. By weak stationarity of the bounding process, $\mathbb{E} \sigma^{*2}_0<\infty$, and therefore also $\mathbb{E}\ln \sigma^{*2}_0<\infty$. Lemma 2.2 of BHK now also yields almost sure convergence of the bounding process $\Phi_t^*$, which in turn implies that the original process also converges. Next we use the following zero-one law, which is a consequence of the first Borel-Cantelli lemma: if $\{Z_t\}$ is a sequence and $\sum_{t=1}^{\infty} P(|Z_t| > \delta ) < \infty$ for all $\delta>0$, then $Z_T \to 0$ a.s., as $T\to\infty$; see \textcite[p. 98]{gut2005}. Applying this theorem to the right-hand side of (\ref{itrepsigma2t}) yields
\begin{align*}
    \sum_{i=1}^{\infty} P \{|\beta_1^{i+1} \sigma_{t-i-1,T}^2| > \delta \} < \infty,
\end{align*}
so
\begin{equation*}
    \beta_1^{i+1} \sigma_{t-i-1,T}^2 \to 0
\end{equation*}
almost surely as $i\to\infty$ since $\beta_1^i$ decays exponentially and $\sigma^2_{t,T}$ is bounded by $\sigma^{*2}_t,$ which by weak stationarity is bounded in expectation. This proves the representation (\ref{eq:decrecht}) for the ATV-GARCH$(1,1)$ model. Let $\boldsymbol{\theta}_0$ denote the "true" parameter vector and note that $\sigma^2_{t,T}=h_{t,T}(\boldsymbol{\theta}_0)$. The representation for the ATV-GARCH$(p,q)$ model follows along similar lines. This scheme generates the nonstationary sequence $h_{t,T}(\boldsymbol{\theta})$ without the need to initialize it with arbitrary starting values but, as noted by \textcite{FrancqZakoian2004}, the computational cost of the procedure is of order $O(T^2)$ instead of $O(T)$.

\subsection{Likelihood functions}

We now turn to QML estimation of the ATV-GARCH model (\ref{eq:GARCH}) and (\ref{eq:themodel})--(\ref{eq:lfunct}). The Gaussian log-likelihood function of $(X_{1,T} \ldots, X_{T,T})$, is given by
\begin{equation}
	L_{T}(\boldsymbol{\theta}) = \frac{1}{T}\sum^T_{t=1}l_{t,T}(\boldsymbol{\theta}),
	\label{loglik} 
\end{equation}
where
\begin{equation*}
	l_{t,T}(\boldsymbol{\theta})=-\frac{1}{2}\left[\log h_{t,T}(\boldsymbol{\theta})+\frac{X^2_{t,T}}{h_{t,T}(\boldsymbol{\theta})}\right]
\end{equation*}
is the log-likelihood for observation $t$. The QMLE is defined as   
\begin{equation}
	\widehat{\boldsymbol{\theta}}_T=\arg \underset{\boldsymbol{\theta}\in\boldsymbol{\Theta}}{\max}\frac{1}{T}\sum^T_{t=1}l_{t,T}(\boldsymbol{\theta}). \label{infestimator} 
\end{equation}

Denote the score for observation $t$ by $\boldsymbol{s}_{t,T}(\boldsymbol{\theta})= \partial l_{t,T}(\boldsymbol{\theta})/\partial \boldsymbol{\theta}$
and the Hessian by
$\boldsymbol{H}_{t,T}(\boldsymbol{\theta})= \partial^2 l_{t,T}/(\boldsymbol{\theta})\partial \boldsymbol{\theta}\partial\boldsymbol{\theta}^{\intercal}.$
We have
\begin{align}
	\boldsymbol{s}_{t,T}(\boldsymbol{\theta})&=-\frac{1}{2}\frac{\partial}{\partial\boldsymbol{\theta}}\left\{ \ln h_{t,T}(\boldsymbol{\theta})+\frac{X_{t,T}^{2}}{h_{t,T}(\boldsymbol{\theta})}\right\}\nonumber \\ 
	&=  -\frac{1}{2}\left(1-\frac{X_{t,T}^{2}}{h_{t,T}(\boldsymbol{\theta})}\right)\frac{1}{h_{t,T}(\boldsymbol{\theta})}\frac{\partial h_{t,T}(\boldsymbol{\theta})}{\partial\boldsymbol{\theta}} \label{eq:firstderivs}
\end{align}
and \begin{align}
	&\boldsymbol{H}_{t,T}(\boldsymbol{\theta})=-\frac{1}{2}\frac{\partial}{\partial\boldsymbol{\theta}^{\intercal}}\left\{ \left(1-\frac{X_{t,T}^{2}(\boldsymbol{\theta})}{h_{t,T}(\boldsymbol{\theta})}\right)\frac{\partial\ln h_{t,T}(\boldsymbol{\theta})}{\partial\boldsymbol{\theta}}\right\} \nonumber \\ &=-\frac{1}{2}\left(1-\frac{X_t^{2}}{h_{t,T}(\boldsymbol{\theta})}\right)\frac{1}{h_{t,T}(\boldsymbol{\theta})}\frac{\partial^2 h_{t,T}(\boldsymbol{\theta})}{\partial\boldsymbol{\theta}\partial\boldsymbol{\theta}^{\intercal}}-\frac{1}{2}\left(2\frac{X_{t,T}^{2}}{h_{t,T}(\boldsymbol{\theta})}-1\right)\frac{1}{h_{t,T}(\boldsymbol{\theta})}\frac{\partial h_{t,T}(\boldsymbol{\theta})}{\partial\boldsymbol{\theta}}\frac{1}{h_{t,T}(\boldsymbol{\theta})}\frac{\partial h_{t,T}(\boldsymbol{\theta})}{\partial\boldsymbol{\theta}^{\intercal}}. \label{eq:secondderivs}
\end{align}

In practice, we only observe $(X_{1,T}, \ldots, X_{T,T})$, so the log-likelihood (\ref{loglik}) cannot be computed. Hence, as in BHK, we replace $l_{t,T}(\boldsymbol{\theta})$ with
\begin{equation*}
	\bar{l}_{t,T}(\boldsymbol{\theta})=-\frac{1}{2}\left[\log \bar{h}_{t,T}(\boldsymbol{\theta})+\frac{X^2_{t,T}}{\bar{h}_{t,T}(\boldsymbol{\theta})}\right],
\end{equation*}
where
\begin{equation}
	\bar{h}_{t,T}(\boldsymbol{\theta})=c_{0}(\boldsymbol{\theta}) + 
 \sum^{t}_{i=1}d_i(\boldsymbol{\theta})g_{t-i+1,T}(\boldsymbol{\theta}_1)
   + \sum^{t-1}_{i=1}c_i(\boldsymbol{\theta})X^2_{t-i,T}. \label{eq:frecht}
\end{equation}
The truncated log-likelihood function is given by
\begin{equation}
	\bar{L}_{T}(\boldsymbol{\theta}) = \frac{1}{T}\sum^{T}_{t=1}\bar{l}_{t,T}(\boldsymbol{\theta}).
\end{equation}
Analogously to (\ref{infestimator}), the truncated QMLE is defined as
\begin{equation}
	\bar{\boldsymbol{\theta}}=\arg \underset{\boldsymbol{\theta}\in\varTheta}{\max}\frac{1}{T}\sum^T_{t=1}\bar{l}_{t,T}(\boldsymbol{\theta}) \label{truncestimator}
\end{equation}
and, analogously to (\ref{eq:firstderivs}) and (\ref{eq:secondderivs}) the expressions for the score and Hessian for observation $t$ are  $\bar{\boldsymbol{s}}_{t,T}(\boldsymbol{\theta})$ and $\bar{\boldsymbol{H}}_{t,T}(\boldsymbol{\theta})$.

Our proofs of consistency and asymptotic normality of the QMLE are based on the local approximation of the log-likelihood function $l_{t,T}(\boldsymbol{\theta})$. Thus we also need the log-likelihood function of the stationary process $\widetilde{X}_t(u)$, defined by
\begin{equation}
	\widetilde{L}_{T}(u, \boldsymbol{\theta}) = \frac{1}{T}\sum^T_{t=1}\widetilde{l}_t(u, \boldsymbol{\theta}),
\end{equation}
where 
\begin{equation*}
	\widetilde{l}_t(u,\boldsymbol{\theta})=-\frac{1}{2}\left[\log \widetilde{h}_t(u,\boldsymbol{\theta})+\frac{\widetilde{X}^2_t(u)}{\widetilde{h}_t(u,\boldsymbol{\theta})}\right].
\end{equation*}
Furthermore,
\begin{equation*}
	\widetilde{h}_t(u,\boldsymbol{\theta})=c_0(\boldsymbol{\theta}) 
 + \sum^{\infty}_{i=1}d_i(\boldsymbol{\theta})g_{t-i+1}(u) +\sum^{\infty}_{i=1}c_i(\boldsymbol{\theta})\widetilde{X}^2_{t-i}(u), \label{eq:infatv}
\end{equation*}
for $u\in[0,1]$, where $g_{t,T}(\boldsymbol{\theta}_1)$ is approximated by $g_{t}(u, \boldsymbol{\theta}_1)$.
Similarly to (\ref{eq:firstderivs}) and (\ref{eq:secondderivs}), the expressions for the score and Hessian for observation $t$ are $\widetilde{\boldsymbol{s}}_t(u,\boldsymbol{\theta})$ and $\widetilde{\boldsymbol{H}}_t(u,\boldsymbol{\theta})$.

\clearpage
\subsection{Main results}

We are now ready to present the main results on consistency and asymptotic normality of the QMLE of the parameters of the ATV-GARCH model. 

To show consistency of the QMLE, we make the following assumptions.
\begin{itemize}
	\item[(A1)]  The random variables $\varepsilon _{t}$ are IID with $\mathbb{E}(\varepsilon_0) = 0$ and $\mathbb{E}(\varepsilon_0^2) = 1$.  $\mathbb{E}\left|\varepsilon _{0}^{2}\right|^{1+d}<\infty$, for some $d>0$. The random variable $\varepsilon^2_0$ is non-degenerate and $\underset{t\to0}{\lim}\ t^{-\mu}\mathbb{P}\{\varepsilon^2_0\leq t\}=0$, for some $\mu>0.$
	\item[(A2)] The parameter space $\Theta$ is compact and the vector $\boldsymbol{\theta}_{0}\in \mbox{int}(\Theta).$
	\item[(A3)] The equation $$\boldsymbol{\lambda}^{\intercal}\frac{\partial \alpha_0(u,\boldsymbol{\theta}_0)}{\partial \boldsymbol{\theta}} =\boldsymbol{0}$$
 for a vector of constants $\boldsymbol{\lambda}$  implies $\boldsymbol{\lambda} = \boldsymbol{0}.$
	\item[(A4)]  The polynomials $\mathcal{A}(z)=\alpha_1z+\alpha_2z^2+\ldots+\alpha_pz^p$ and $\mathcal{B}(z)=1 - \beta_1z-\beta_2z^2-\ldots-\beta_qz^q$ are coprime on the set of polynomials with real coefficients.
	\item[(A5)]  $\sum_{i=1}^{p}\alpha _{i}+\sum_{j=1}^{q}\beta_j<1.$
\end{itemize}

\begin{rmk}
    The assumption that the random variables $\varepsilon _{t}$ are IID$(0,1)$ was made in Proposition 1 to show that the ATV-GARCH model can be locally approximated by a stationary process. The remaining assumptions in (A1) are inherited from BHK. The assumption (A2) is a standard assumption used for proving consistency and asymptotic normality. As BHK noted, (A2) precludes zero coefficients for $\boldsymbol{\alpha}=(\alpha_1, \ldots, \alpha_p)^\intercal$ and $\boldsymbol{\beta}=(\beta_1, \ldots \beta_q)^\intercal$. (A3) is an identification condition on the time-varying intercept. In the appendix, we show that the identification condition holds for an intercept containing one logistic transition function with one $c$ parameter. It is likely to hold more generally, but to avoid difficulties in proving it, we leave it as an assumption. (A4) is an identification condition. (A5) together with the first part of (A1) is a sufficient condition for weak stationarity of the standard GARCH process.
\end{rmk}

We are now ready to present the following theorem.

\begin{thm}
	Under (A1)--(A5),
	\begin{equation*}	 \widehat{\boldsymbol{\theta}}_T\overset{P}{\to}{\boldsymbol{\theta}_0}, \text{ as } T\to\infty. \label{eq:consistency}
		\end{equation*}
\end{thm}
\begin{proof}
    See Appendix A.4.
\end{proof}

\begin{rmk}
Under (A1)–(A5), the ATV-GARCH process is locally stationary. In the proof of the theorem, we show that the sequence of log-likelihood functions $l_{t,T}(\boldsymbol{\theta})$ can be locally approximated by the
stationary process $\widetilde{l}_t(u,\boldsymbol{\theta})$.
Using the global law of large numbers (Theorem 2.7(i) of DRW), our proof also requires 
$\sup_{u\in[0,1]}\norm{\widetilde{l}_t(u,\boldsymbol{\theta})}_1 <\infty$. 
\end{rmk}
\begin{rmk} A main difference between the ATV-GARCH model and the nonparametric maximum likelihood estimation
of parameter curves in DRW is that in a nonparametric framework, smoothness conditions on the log-likelihood function are imposed, whereas in the parametric ATV-GARCH model we need to show that Lipschitz- 
or H\"{o}lder-type conditions hold for the log-likelihood function and its derivatives (the score and the Hessian). In the proof we show that a condition of the type (\ref{eq:statespaceS1}) with $n=1$ holds for the log-likelihood function.
\end{rmk}
\begin{rmk} Theorem 2.7 (i) of DRW gives convergence in $L^{1}$. Therefore, we are not able to show strong (almost sure) consistency of the QMLE of the parameters. 
\end{rmk}
\begin{rmk} In practice, we have to use the truncated estimator
$\bar{\boldsymbol{\theta}}_{T}$ instead of the estimator $\widehat{\boldsymbol{\theta}}_{T}$. To prove consistency of $\bar{\boldsymbol{\theta }}_{T}$ from consistency of $\widehat{\boldsymbol{\theta }}_{T}$, it is sufficient to show that the truncated log-likelihood function $\bar{L}_{T}(\boldsymbol{\theta })$ converges uniformly to the log-likelihood function $L_{T}(\boldsymbol{\theta})$ (cf. BHK). 
\end{rmk}
To show asymptotic normality of the QMLE, we make the following additional assumptions. 	\begin{itemize}
\item[(A6)] $\mathbb{E}|\varepsilon_t|^{4+d}<\infty$, for some $d>0$.
\item[(A7)] $\lambda_{\text{spec}}((\mathbf{A})_2)<1-\delta$, for some $\delta>0$, where $\mathbf{A}_t(u)$ is given in (\ref{eq:GARCHSS}) and the parameters are constant. 
\end{itemize}

\begin{rmk}
    (A7) is equivalent to the necessary and sufficient condition for the existence of a fourth moment of the standard GARCH$(p,q)$ process in He and Teräsvirta (1999) and \textcite{LingMcaleer}.
\end{rmk}

\begin{thm}
Under (A1)--(A7), 
\begin{equation}
		\sqrt{T}\left(\widehat{\boldsymbol{\theta}}_T-\boldsymbol{\theta}_0\right) \overset{D}{\to} N\left(\mathbf{0},\mathbf{B}^{-1}\mathbf{A}\mathbf{B}^{-1}\right), \text{ as } T\to\infty. \label{eq:asynormality}
	\end{equation}
The expressions for $\mathbf{A}$ and $\mathbf{B}$ are integrals of the stationary approximation of the expected score and Hessian at $u$, $u \in [0,1]$, respectively:

$$
\mathbf{A} = \int_0^1 \mathbb{E}(\mathbf{A}(u))  
\ \text{d}u,
$$
where
$$
\mathbf{A}(u) = \text{Var}(\widetilde{\mathbf{s}}_{0} (u,\boldsymbol{\theta}_0)),
$$
and
$$
\mathbf{B} = \int_0^1 \mathbb{E}(\widetilde{\mathbf{H}}(u,\boldsymbol{\theta}_0))\ \text{d}u.
$$
\end{thm}
\begin{proof}
    See Appendix A.5.
\end{proof}

\begin{rmk} To obtain asymptotic normality, we show that a condition of the type (\ref{eq:statespaceS1}) with $n=2$ holds for the score. Further, the proof requires a moment condition and a mixing condition on the stationary approximation $\widetilde{\mathbf{s}}_{t}(u,\boldsymbol{\theta})$ of the score. 
\end{rmk}
\begin{rmk} To show that Theorem 2 remains true for the truncated estimator $\bar{\boldsymbol{\theta}}_{T}$, it suffices to show that $|\widehat{\boldsymbol{\theta}}_{T}-\bar{\boldsymbol{\theta}}
_{T}|=O(T^{-1})$ almost surely. Asymptotic normality of $\sqrt{T}(\bar{\boldsymbol{\theta}}_{T}-\boldsymbol{\theta}_{0})$
is then an immediate consequence of Theorem 2 (cf. BHK, proof of Theorem 4.4, p. 226).
\end{rmk}

\section{Simulation study}

In order to examine the approximations in Theorems 1 and 2, we carry out a simulation study. We use time series lengths $T = 3000$ and  $6000$. The number of Monte Carlo replications is $10000$. The GARCH$(1,1)$ parameters in the data-generating processes (DGPs) are $\alpha _{0}=0.05$, $\alpha _{1}=0.1$ and $\beta _{1}=0.8$. We consider three different DGPs for $g(t/T;\boldsymbol{\theta }_{1})=\alpha_{01} G(t/T; \gamma, c)$, with $\alpha _{01}=0.15$. This implies that the intercept increases by a factor of four. The value of the linear parameter $\alpha_{01}$ is empirically relevant. In the empirical example in the next section, the intercept is reduced by a factor of five. We experimented with other values for $\alpha_{01}$ and found that the approximations work better in finite samples when the transition is a pronounced feature of the data. We choose $\gamma $ such that $G(a;\gamma ,c)=0.01$ and $G(b;\gamma ,c)=0.99$. The values $a=0.1$ and $b=0.9 $ in DGP 1 define a slow transition with $80\%$ of the observations affected by the transition. The values $a=0.25$ and $b=0.75$ in DGP 2 define a moderate transition with $50\%$ of the observations affected by the transition. Finally, the values $a=0.4 $ and $b = 0.6$ in DGP 3 define a rapid transition with $20\%$ of the observations affected by the transition. The value for $c$ is $c=0.5$. The parameter values are contained in Table \ref{tab:simstudyspecs}. Figure 3 plots the transition functions of DGPs 1--3. The errors $\{\varepsilon_t\}$ are NID$(0,1)$. To reduce the impact of the starting values, we use a burn-in period of $500$ observations. 

\begin{figure}[h]
    \centering
    \caption{Transition functions of DGPs 1--3.}
    \label{fig:gfunctions}
    \includegraphics[width=0.95\linewidth]{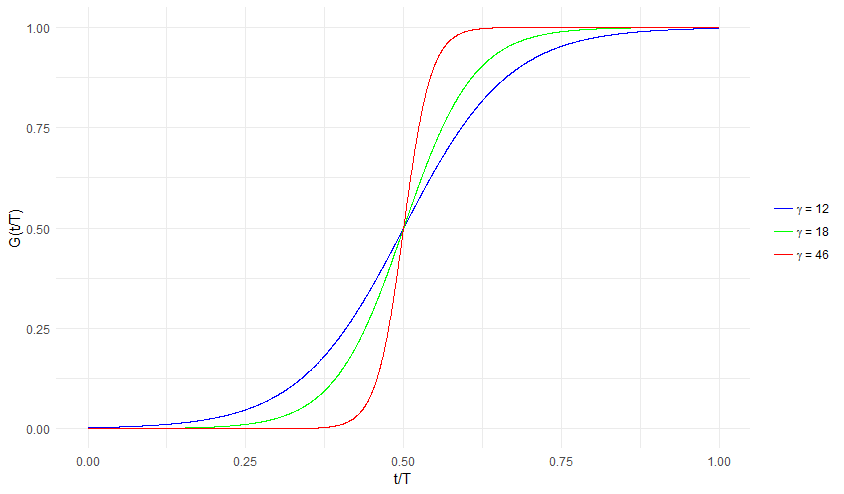}
    \centering
  
\end{figure}

ML estimation is carried out using the solver \texttt{solnp} by \textcite{Ye1987}, implemented in the \texttt{R} package \texttt{Rsolnp} by \textcite{Rsolnp}.  The recursion for the conditional variance (\ref{eq:frecht}) is truncated at $200$ observations. This has a negligible impact on the conditional variance, but it drastically speeds up estimation. In the estimation we impose $\alpha _{1}+\beta _{1}<1$. For large values of $\gamma $, its impact on the shape of the transition function becomes small and the log-likelihood function is flat in the direction of this parameter. To improve the numerical accuracy of the estimate of $\gamma$, we apply the transformation $\gamma =\eta (1-\eta )^{-1}$ proposed by \textcite{ekner2013parameter}. We set the maximum value for $\eta$ equal to $0.999$, which was reached in a small number of replications. They were discarded and new draws obtained until the number of replications reached $10000$. For each DGP, we report in Table \ref{tab:simstudyspecs} the frequency with which this happens. The starting values for the parameters $(\alpha _{0},\alpha _{1},\beta _{1},\gamma,c,\alpha
_{01} )$ are set to the true values in the DGPs.

\begin{table}[tbp] \centering
\caption{DGPs in the simulations. Disc denotes the percentage of replications discarded.}\vspace{0.5cm}%
\begin{tabular}[b]{llllllll}
\hline
DGP & Transition & $a$ & $b$ & $\gamma $ & $c$ & Disc & \\ 
\hline
&  &  &  &  & & $T=3000$ & $T=6000$   \\ 
1 & Slow & $0.1$ & $0.9$ & $12$ & $0.5$  & 0.575\% & 0.033\%\\ 
2 & Medium & $0.25$ & $0.75$ & $18$ & $0.5$  & 1.025\% & 0.042\%   \\ 
3 & Rapid & $0.4$ & $0.6$ & $46$ & $0.5$  & 6.558\% & 1.425\%  \\ \hline
\end{tabular}
\label{tab:simstudyspecs}
\end{table}

Table 2 reports the means and standard deviations of the parameter estimates. The estimates converge to the true parameter values when $T$ increases. The results for $\widehat{\alpha }_{1}$ show that the ARCH parameter $\alpha _{1}$ is accurately estimated in moderate samples. The estimates $\widehat{\alpha }_{0}$ and $\widehat{\alpha }_{01}$ converge
from above, whereas the estimate $\widehat{\beta }_{1}$ converges from below. In small samples, the effect of the overestimated intercept parameters $\alpha _{0}$ and $\alpha _{01}$ on the unconditional variance is offset by an underestimated $\beta_1.$ Consequently, the time-varying intercept is more accurately reproduced than suggested by the individual parameter estimates. The location parameter $c$ and the slope parameter $\gamma $ are well estimated in moderate samples. For large values of $\gamma $, the estimate $\hat{\gamma}$ is less precise. As $\gamma $ increases, the first derivative of the logistic transition function with respect to $\gamma $ increases around $c$, making the Lipschitz constant larger. This means that the number of observations affected by the transition gets smaller.

Figures 4--9 show the distributions of the standardised parameter estimates. 
It is seen that the empirical distributions of the standardised parameter estimates of the GARCH parameters $(\alpha _{0},\alpha _{1},\beta _{1})$
are close to normal for moderate sample sizes. For the estimates of the nonlinear parameters $\gamma $ and $c$, longer time series are required for the normal approximation to be good. As the increasing Lipschitz constant suggests, the quality of the normal approximation deteriorates with an increasing $\gamma $. In some replications the algorithm gets stuck at the initial value for $\gamma $, which results in a spike in the simulated distribution of the standardised parameter estimates. 

The ATV-GARCH model is a model for long financial time series which exhibit gradual change in the unconditional volatility. A consequence of this is that long time series are required to provide enough observations affected by the transition, so that the normal approximation to $\widehat{\gamma }$ is a good approximation. From a practical point of view, the logistic transition function is an approximation of a non-linear structure. In empirical applications of the ATV-GARCH model the shape of the transition function implied by the parameter estimates of $\gamma $ and $c$ is more important than the exact estimates and their standard errors.
\begin{table}
\caption{Simulated means and standard deviations of the parameter estimates.}
\centering
\begin{tabular}[t]{lrrrrrr}
\toprule
 &\multicolumn{3}{c}{GARCH} & \multicolumn{3}{c}{$G$} \\
\cmidrule(l{3pt}r{3pt}){2-4} \cmidrule(l{3pt}r{3pt}){5-7}
  & $\alpha_0$ & $\alpha_1$ & $\beta_1$ & $\eta$ & $c$ & $\alpha_{01}$\\
\midrule
$\gamma = 12$ & & & & & & \\
\cline{1-1}
3000 &$\underset{(0.016)}{0.056}$ &$\underset{(0.017)}{0.101}$ &$\underset{(0.038)}{0.786}$ &$\underset{(0.032)}{0.923}$ &$\underset{(0.058)}{0.507}$ &$\underset{(0.059)}{0.176}$\\
6000 &$\underset{(0.010)}{0.053}$ &$\underset{(0.012)}{0.100}$ &$\underset{(0.026)}{0.793}$ &$\underset{(0.020)}{0.923}$ &$\underset{(0.033)}{0.502}$ &$\underset{(0.034)}{0.161}$\\
\midrule
$\gamma = 18$ & & & & & & \\
\cline{1-1}
3000 &$\underset{(0.015)}{0.056}$ &$\underset{(0.017)}{0.101}$ &$\underset{(0.038)}{0.786}$ &$\underset{(0.021)}{0.948}$ &$\underset{(0.034)}{0.502}$ &$\underset{(0.048)}{0.172}$\\
6000 &$\underset{(0.010)}{0.053}$ &$\underset{(0.012)}{0.100}$ &$\underset{(0.026)}{0.793}$ &$\underset{(0.014)}{0.948}$ &$\underset{(0.022)}{0.500}$ &$\underset{(0.031)}{0.160}$\\
\midrule
$\gamma = 46$ & & & & & & \\
\cline{1-1}
3000 &$\underset{(0.014)}{0.056}$ &$\underset{(0.017)}{0.100}$ &$\underset{(0.038)}{0.786}$ &$\underset{(0.012)}{0.978}$ &$\underset{(0.018)}{0.501}$ &$\underset{(0.045)}{0.171}$\\
6000 &$\underset{(0.010)}{0.053}$ &$\underset{(0.012)}{0.100}$ &$\underset{(0.026)}{0.793}$ &$\underset{(0.008)}{0.979}$ &$\underset{(0.012)}{0.500}$ &$\underset{(0.030)}{0.160}$\\
\bottomrule
\multicolumn{7}{l}{\rule{0pt}{1em}\textit{Note: }}\\
\multicolumn{7}{l}{\rule{0pt}{1em}\small{See Table \ref{tab:simstudyspecs} for details on the simulation design. }}\\
\multicolumn{7}{l}{\rule{0pt}{1em}\small{We report the estimate of the transformed parameter $\eta$. }}\\
\multicolumn{7}{l}{\rule{0pt}{1em}\small{The value $\gamma=12$ corresponds to $\eta=0.923$. }}\\
\multicolumn{7}{l}{\rule{0pt}{1em}\small{The value $\gamma=18$ corresponds to $\eta=0.947$. }}\\
\multicolumn{7}{l}{\rule{0pt}{1em}\small{The value $\gamma=46$ corresponds to $\eta=0.979$. }}\\
\multicolumn{7}{l}{\rule{0pt}{1em}\small{Sample standard deviations in parentheses.}}\\
\end{tabular}
\end{table}

\newpage 

\begin{figure}[h]
    \centering
    \caption{Simulated distributions of the standardised parameter estimates. DGP 1: $\gamma = 12$, $T=3000$.}
    \label{fig:2.2}
    \includegraphics[width=0.95\linewidth]{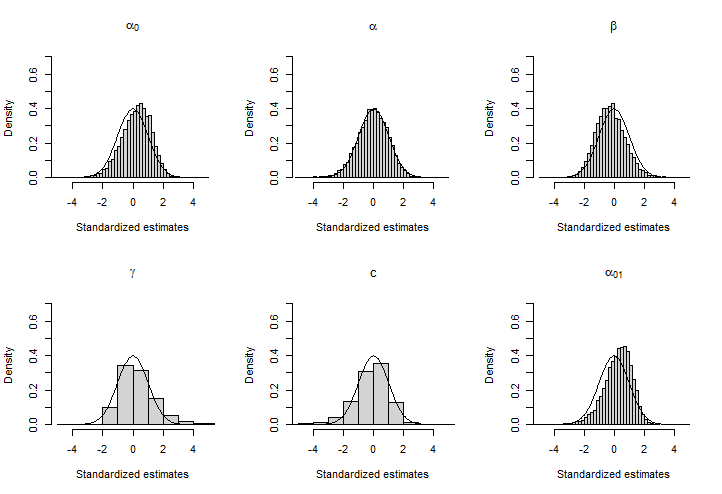}
    \centering
\end{figure}

\begin{figure}[h]
    \centering
    \caption{Simulated distributions of the standardised parameter estimates. DGP 1: $\gamma = 12$, $T=6000$.}
    \label{fig:3.2}
    \includegraphics[width=0.95\linewidth]{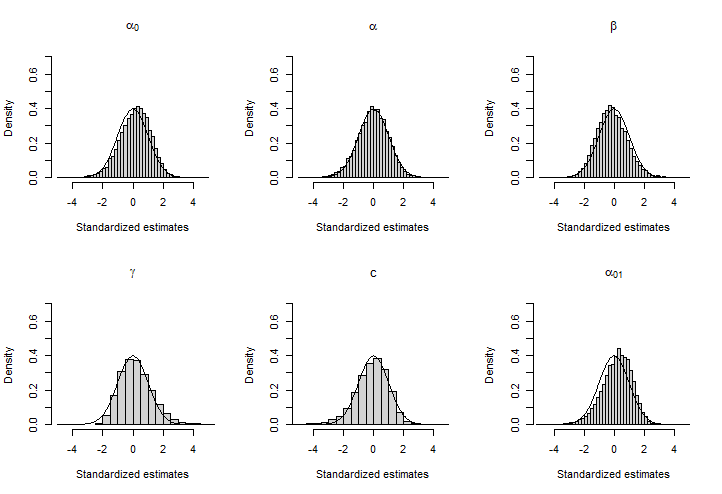}
    \centering
\end{figure}

\begin{figure}[h]
    \centering
    \caption{Simulated distributions of the standardised parameter estimates. DGP 2: $\gamma = 18$, $T=3000$.}
    \label{fig:4.2}
    \includegraphics[width=0.95\linewidth]{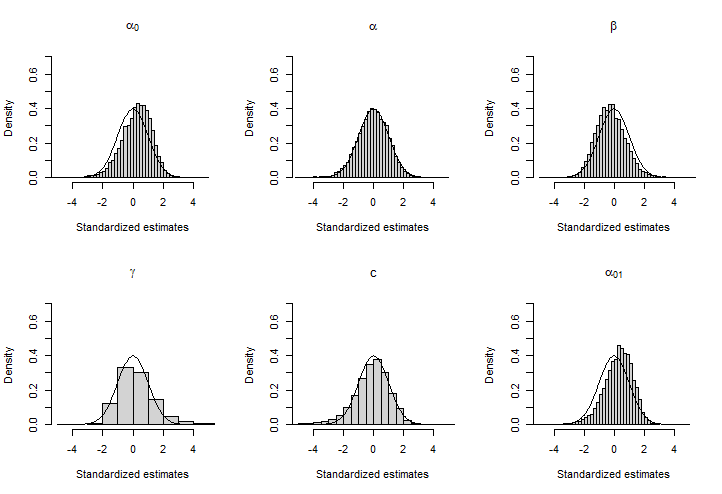}
    \centering
\end{figure}

\begin{figure}[h]
    \centering
    \caption{Simulated distributions of the standardised parameter estimates. DGP 2: $\gamma = 18$, $T=6000$.}
    \label{fig:5.2}
    \includegraphics[width=0.95\linewidth]{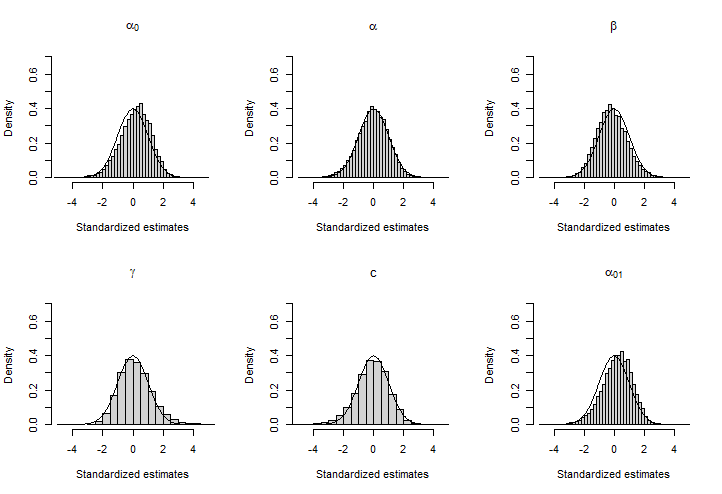}
    \centering
\end{figure}

\begin{figure}[h]
    \centering
    \caption{Simulated distributions of the standardised parameter estimates. DGP 3: $\gamma = 46$, $T=3000$.}
    \label{fig:6.2}
    \includegraphics[width=0.95\linewidth]{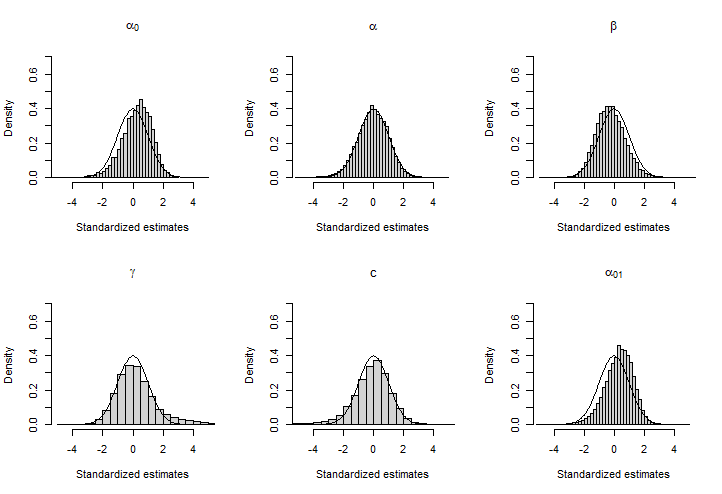}
    \centering
\end{figure}

\begin{figure}[h]
    \centering
    \caption{Simulated distributions of the standardised parameter estimates. DGP 3: $\gamma = 46$, $T=6000$.}
    \label{fig:7.2}
    \includegraphics[width=0.95\linewidth]{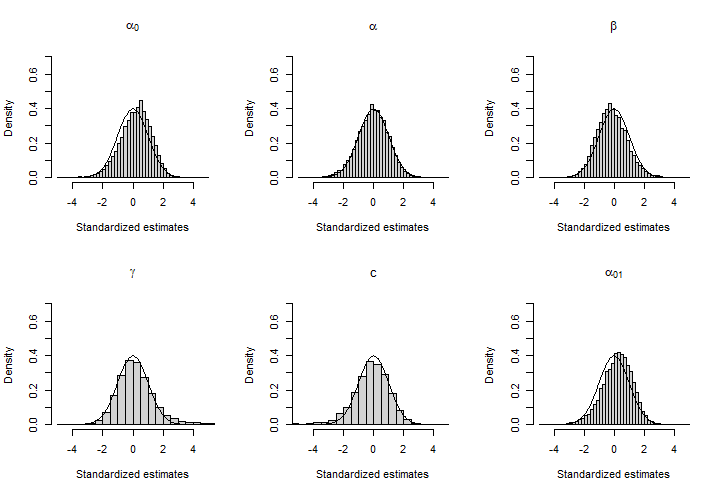}
    \centering
\end{figure}
\newpage

\pagebreak

\section{Empirical example}
As an example of the use of the ATV-GARCH model we consider Oracle Corporation (ticker: ORCL) daily closing prices on the New York Stock Exchange from the beginning of its listing on 12 March 1986 until 13 February 2023, 9306 observations in all. The data are downloaded from Yahoo Finance. The series and the corresponding logarithmic returns are depicted in Figure 10. As can be seen, up until the early 2000s the amplitude of the volatility clusters is large, and the descent that follows corresponds to the time when the dot-com bubble bursts. Summary statistics for the log returns appear in Table 3, Panel A. They show that the robust skewness measure based on quartiles does not indicate any skewness. This suggests that the observed nonrobust skewness is mainly caused by a small number of large negative returns being greater than their positive counterparts. As may be expected, both measures of kurtosis indicate leptokurtosis in the returns.

The results of fitting a standard GARCH$(1,1)$ model to the logarithmic returns can be found in Table 3, Panel B. The estimated persistence is $\widehat{\alpha }_{1}+\widehat{\beta }_{1} > 0.999$, which indicates that the unconditional variance may not be constant over time. This tentative conclusion is supported by the LM-type tests discussed in Section 2.1. The table contains results from both the nonrobust and robust (robustified against departures from distributional assumptions as in  \cite{wooldridge1990}) versions of the tests. The statistics are asymptotically $\chi^{2}$-distributed with three degrees of freedom under the null hypothesis of constant unconditional variance. Both tests strongly reject the null hypothesis.

The parameter estimates and standard errors of the fitted ATV-GARCH(1,1) model can be found in Table \ref{tab:empapp}, Panel C. The results show that the persistence, compared to the standard GARCH, decreases from above $0.999$ to $0.94$. This lends credence to the hypothesis in Diebold (1986). 
 The fitted variance and transition function for the period $1998-2005$ are shown in Figure \ref{fig:oraclmodel}. We note that the fitted model agrees with the informal discussion in the beginning of this section. The smoothly time-varying intercept starts out high and exhibits a precipitous but discernibly gradual decline during the years $2001-2004$, which coincides with the bursting of the tech bubble.
To put the value $\widehat{\eta} =0.994$ of the slope parameter into perspective, the reverse transformation yields $\widehat{\gamma } =166$. Very small changes in $\eta$, when it is large, translate into substantial changes in $\gamma $. From Table \ref{tab:empapp} it is seen that the standard deviation of $\eta $ equals 0.003. The range $\widehat{\eta } \pm 0.003$ corresponds to the interval $(110, 332)$ for $\widehat{\gamma }$.

As can be seen from Table 3, Panel B, testing one transition against two  does not lead to a rejection of the null hypothesis. Our conclusion is that the ATV-GARCH$(1,1)$ model with a single transition is an adequate first-order GARCH-type description of the volatility of the Oracle return series.

\begin{figure}[h]
     \centering
     \caption{Oracle Corporation stock price (ORCL) on the New York Stock Exchange, 12 March 1986 to 13 February 2023.}
     \begin{subfigure}[b]{0.7\textwidth}
         \centering
         \includegraphics[width=\textwidth]{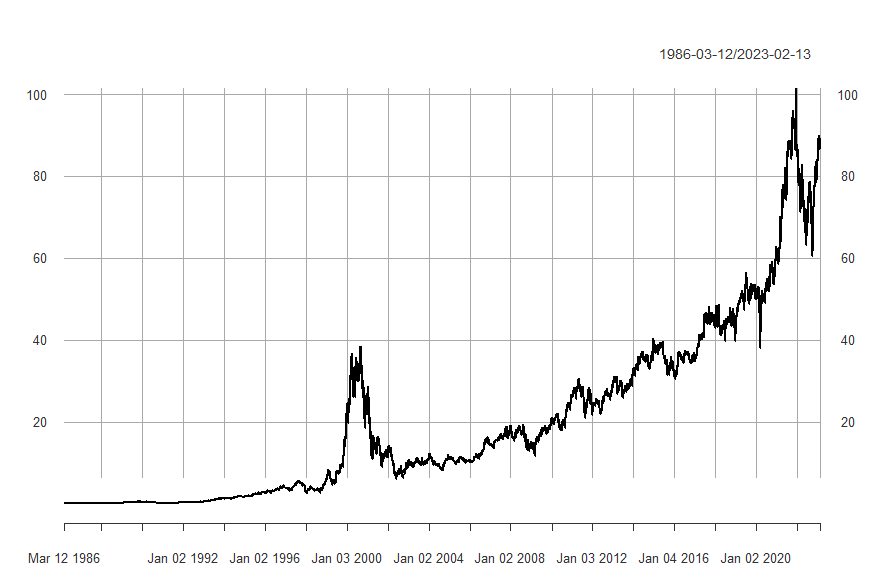}
         \caption{Oracle Corporation stock price.}
         \label{fig:oraclprices}
     \end{subfigure}
     \\
     \begin{subfigure}[b]{0.7\textwidth}
         \centering
         \includegraphics[width=\textwidth]{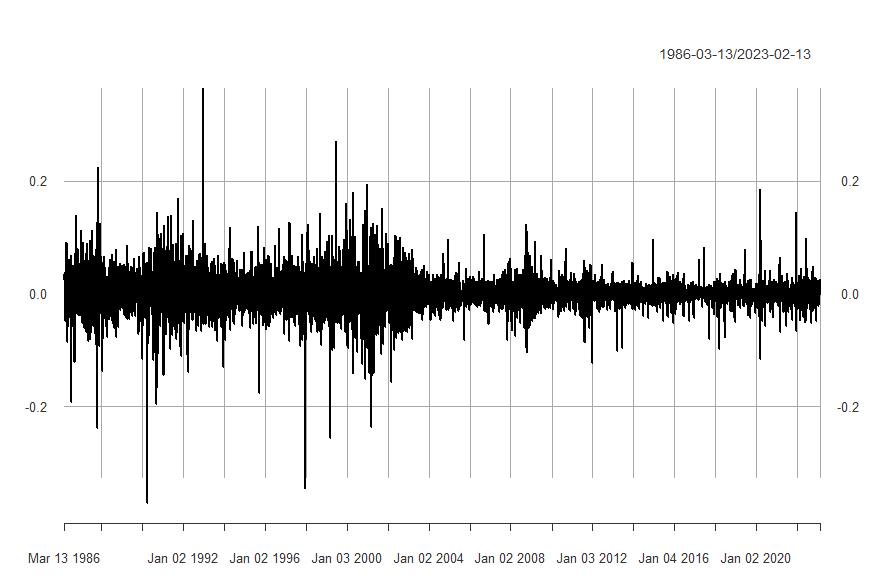}
         \caption{Oracle Corporation log returns.}
         \label{fig:oracllogrets}
     \end{subfigure}
     \label{fig:oraclestock}
\end{figure}

\begin{figure}[h]
     \centering
     \caption{ATV-GARCH$(1,1)$ model of Oracle log returns}
     \begin{subfigure}[b]{0.7\textwidth}
         \centering
         \includegraphics[width=\textwidth]{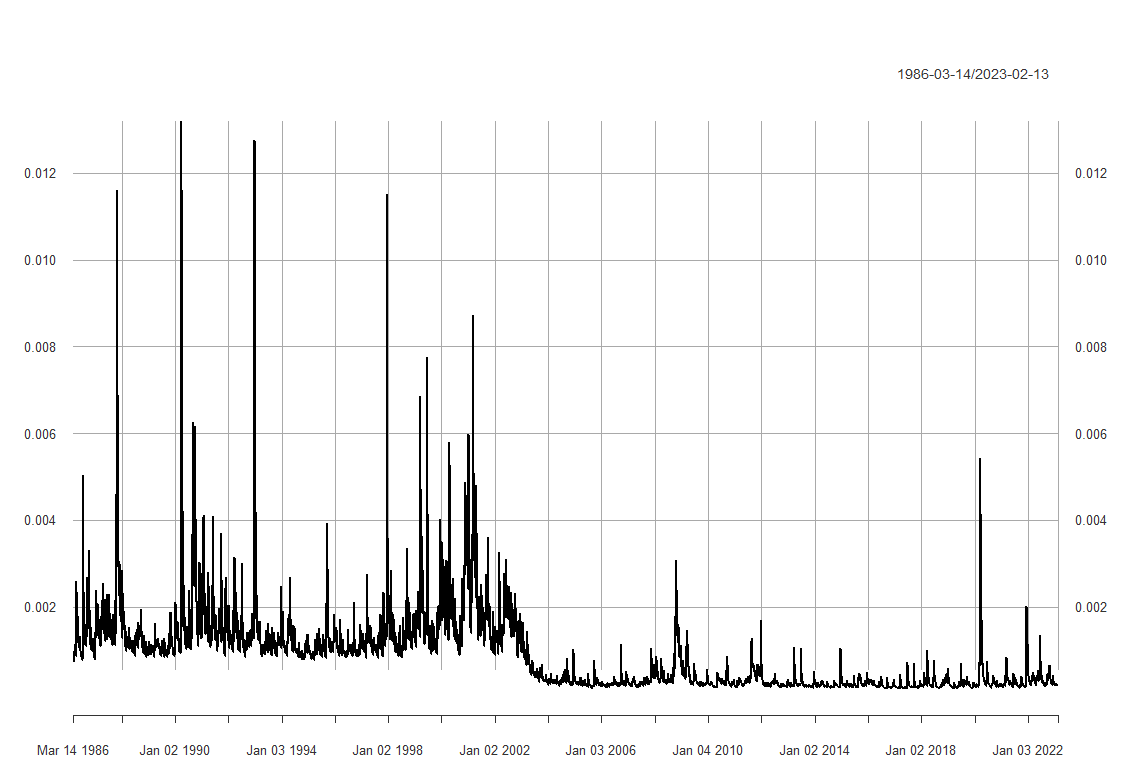}
         \caption{The fitted conditional variance.}
         \label{fig:oraclcondvar}
     \end{subfigure}
     \\
     \begin{subfigure}[b]{0.7\textwidth}
         \centering
         \includegraphics[width=\textwidth]{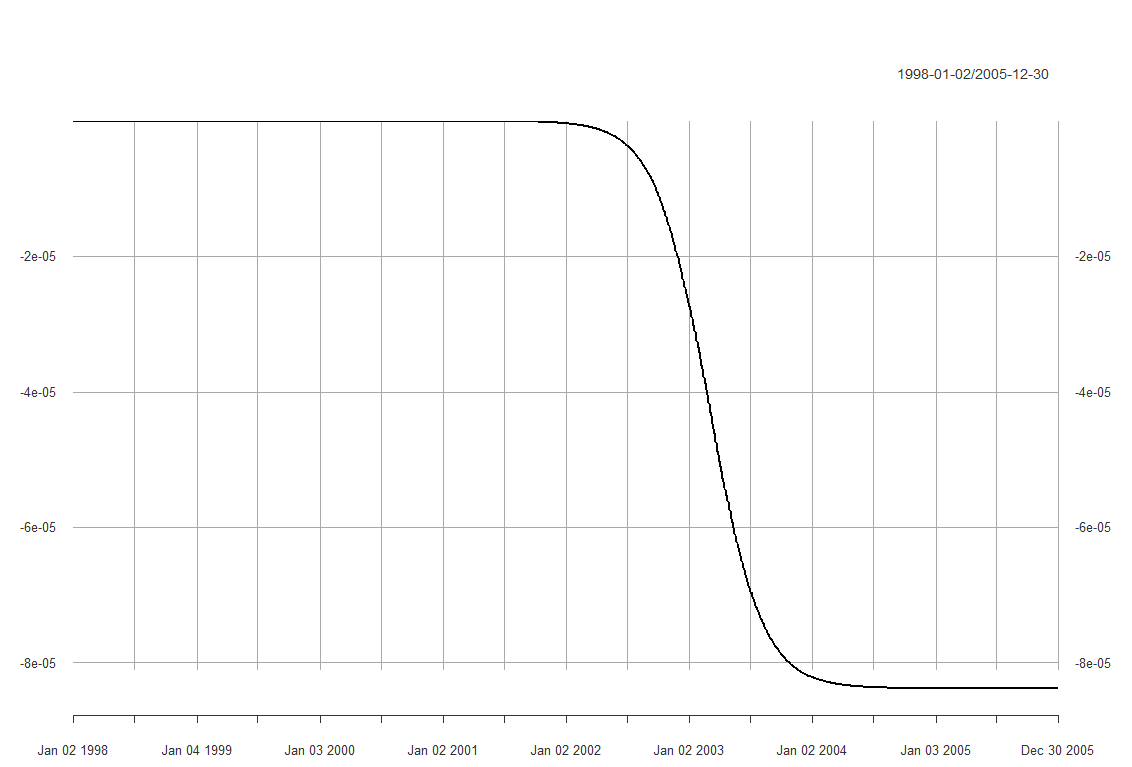}
         \caption{The fitted logistic transition function 1998--2005.}
         \label{fig:oraclgt}
     \end{subfigure}
     \label{fig:oraclmodel}
\end{figure}

\begin{table}[h]
\centering
\caption{Results of the empirical application to Oracle logarithmic returns 1986--2023.}
\resizebox{\columnwidth}{!}{
\begin{tabular}{llllllllll} 
\toprule
\textbf{Panel A. Summary statistics}   & ~                          & ~                         & ~                       & ~                          & ~                         & ~                          & ~              & ~                         & ~             \\ 

Series            & Mean   & Sd    & Med & Min    & Max   & Skew   & R Skew & Kurt & R Kurt   \\ 

ORCL             & 0.001 & 0.029 & 0.000   & $-0.372$ & 0.364 & $-0.168$ & 0.031   & 16.781  &  0.332\\ 
\bottomrule
\textbf{Panel B. Misspecification tests} &                            &                           &                         &                            &                           &                            &          & &                  \\ 

Null              & $LM(3)$                       & R $LM(3)$                     & $p$-val                  & R $p$-val                    & $\widehat{\alpha}_1$                      &  $\widehat{\beta}_1$                            & $\widehat{\alpha}_1 + \widehat{\beta}_1$       & &                       \\ 

$L = 0$               & 34.582                          & 27.281                          &  $<0.000$                       & $<0.000$                           & 0.081                          & 0.917                           &  $> 0.999$            & &              \\

$L=1$                & 4.641                          & 3.841                          &  $0.200$                       & $0.275$                           & 0.087                          & 0.854                           &  0.940     & &                   \\
   
\bottomrule
\textbf{Panel C. ATV-GARCH$(1,1)$ model} &                            &                           &                         &                            &                           &                            &                  & &          \\ 

Parameter              & Estimate                       & Se                    & R se                 &                    &                      &                            &   & &                         \\ 

$100 \times \alpha_{0}$               &  0.010                         &   0.000                        &   0.000                     &                         &                    &                            &                & &            \\

$\alpha_{1}$               &  0.087                          &   0.002                        &   0.002                      &                           &                         &                            &       & &    \\

$\beta_{1}$              &  0.854                          & 0.003                          &    0.005                   &                         &                        &                            &       & &    \\   

$\eta$               &  0.994                          &  0.003                        &  0.004                      &                         &                       &                            &       & &    \\   

$c_1$               &  0.460                        & 0.003                         &   0.004                                                &                       &                            &      & &     \\   

$100 \times \alpha_{01}$               &  $-0.008$                          & 0.000                         &  0.000                      &                            &                     &                            &     & &      \\   
\bottomrule
\end{tabular}
}
\raggedright{{\\ \footnotesize \emph{Note:} \\ Mean is the average value of the series over the sample period, Sd is the standard deviation and Med is the median. Min and Max are the minimum and maximum values, respectively. Skew is the skewness and Kurt is the kurtosis. R Skew and R Kurt are robust and centered measures of skewness and kurtosis (Bowley's skewness and Moors's Kurtosis). Null is the null hypothesis of the misspecification test. $LM(3)$ is the LM statistic and R $LM(3)$ is the robust LM statistic. The table reports the transformed parameter $\eta$. Se denotes standard error. The standard errors are computed using numerical approximations of the score and Hessian. 
 }} 
\label{tab:empapp}
\end{table}

\section{Conclusion}

We have introduced a new GARCH model with a deterministic time-varying intercept, called the ATV-GARCH model. In this model, volatility is mean-reverting towards a time-varying mean. The model captures structural change that slowly affects the amplitude of a time series while keeping the short-run dynamics constant over time. The idea is that short-run fluctuations in volatility are stationary, but long-run change makes the assumption of stationarity inappropriate. The model is particularly well suited for situations in which the volatility of an asset or index is smoothly increasing or decreasing over time. 

The ATV-GARCH model is globally nonstationary but can be locally approximated by a stationary GARCH process. The stationary approximation enables the derivation of asymptotic results. We show that the QMLE of the parameters of the ATV-GARCH model is consistent and asymptotically normal. The asymptotic theory of QML estimation of the ATV-GARCH model relies on assumptions that are straightforward and comparable to the ones in strictly stationary GARCH processes.

We apply our model to Oracle Corporation daily stock returns. As LM tests find strong evidence of a smooth transition in the intercept, we fit an ATV-GARCH$(1,1)$ model to the series. The persistence implied by the standard GARCH parameter estimates is substantially reduced. We conjecture that, by using an ATV-GARCH model, the estimated persistence in volatility is reduced in many situations. This reduction would likely impact the results of empirical studies in which GARCH models are fitted to long financial time series with the purpose of forecasting volatility. It is possible that the slowly changing intercept may generate spurious evidence of long memory in volatility, leading to an IGARCH or a FIGARCH model. This is an empirical problem which is left for future research.

\clearpage

\section*{Appendix: Proofs}
This appendix contains the proofs of consistency and asymptotic normality of the QMLE of the parameters of the ATV-GARCH model. We use results from DRW for processes that can be locally approximated by stationary processes.
Following \textcite{dahlhaus1997} and \textcite{DahlhausSR06}, such processes are called locally stationary. As shown by Subba Rao (2006), the time-varying GARCH process is included in the class of locally stationary processes. By Proposition 1 in the main text, the ATV-GARCH model is locally stationary. To show that the QMLE of the parameters of the ATV-GARCH model is consistent, we mainly need the global law of large numbers of DRW (Theorem 2.7(i)). Application of the theorem to the log-likelihood function requires showing Assumption 2.1(S1) in DRW is satisfied with $n=1$ for the log likelihood function and its stationary approximation. To prove asymptotic normality, we apply the global central limit theorem of DRW (Theorem 2.9) to the elements of the score. In order to do this, we have to show that Assumptions 2.1(S1), (S2) and 2.3(MI) in DRW hold with $n=2$ for the elements of the score and their stationary approximations.

The notation from the main text is used. Section A.1 contains the assumptions from DRW that we have to verify. In Section A.2 we present the global law of large numbers and the central limit theorem from DRW that will be used in the derivations. Section A.3 contains some intermediate results. The proofs of Theorems 1 and 2 can be found in Sections A.4 and A.5, respectively. Finally, Section A.6 contains the derivations for the truncated QMLE. 

\subsection*{A.1 Assumptions from DRW}
\label{section:ADRW}
We begin by stating Assumptions 2.1(S1), (S2) and 2.3(M1) in DRW. The assumptions are referred to as (S1), (S2) and (M1), respectively. Assumptions (S1) and (S2) are concerned with the existence of a stationary approximation. Let $n>0$ and $\norm{X}_n = \left(\mathbb{E}\left|X\right|^n \right )^{1/n}$. Let $X_{t,T}$, $t=0,1, \ldots T$, be a triangular array of stochastic processes. For each $u\in[0,1]$, let $\widetilde{X}_t(u)$ be a stationary and ergodic process such that the following assumptions hold.
\begin{asmpt*}[S1]
(i) $\sup_{u\in[0,1]}\norm{\widetilde{X}_{t}(u)}_n<\infty.$ (ii) There exist $1\geq\alpha>0$, $C>0$ such that uniformly in $t=1,\ldots,T$ and $u,v\in[0,1]$,
$$\norm{\widetilde{X}_{t}(u)-\widetilde{X}_{t}(v)}_n\leq C\left|u-v\right|^{\alpha}, \quad \norm{X_{t,T}-\widetilde{X}_t(t/T)}_n\leq CT^{-\alpha}.$$
\end{asmpt*}
Throughout, we set $\alpha=1$ in (S1).
\begin{asmpt*}[S2]
$u\mapsto \widetilde{X}_t(u)$ is almost surely continuous for all $t\in\mathbb{Z}$ and $\norm{\sup_{u\in [0,1]}\left|\widetilde{X}_{t}(u)\right|}_n<\infty$.
\end{asmpt*}

Finally, Assumption (M1) is a mixing condition on the stationary approximation $\widetilde{X}_t(u)$, required for the proof of asymptotic normality. Let $\varepsilon_t$, $t\in\mathbb{Z}$, be a sequence of independent and identically distributed random variables. For $t\geq0$, define the $\sigma $-fields $\mathcal{F}_{t}=\sigma (\varepsilon _{t},\varepsilon _{t-1}, \ldots)$ and $\mathcal{F}_{t}^{e 0}=\sigma (\varepsilon _{t},\varepsilon _{t-1}, \ldots, \varepsilon_{1}, 
\varepsilon_{0}^{e }, \varepsilon_{-1}, \ldots)$,
where $\varepsilon _{0}^{e }$ has the same distribution as $\varepsilon
_{0}$ and is independent of all $\varepsilon _{t}$, $t\in \mathbb{Z}$. The functional dependence measure for the stationary process $\widetilde{X}_{t}(u)$ is
defined as (see Wu 2011, p. 2) 
\begin{equation*}
\delta^{\widetilde{X}(u)}_{n}(t)=\left\Vert \widetilde{X}_{t}(u)-\widetilde{X}_{t}^{e}(u)\right\Vert _{n},
\end{equation*}
where $\widetilde{X}_{t}^{e}(u)$ is a coupled version of $\widetilde{X}_{t}(u)$
with $\varepsilon _{0}$ in the latter being replaced by $\varepsilon
_{0}^{e}$. The process is said to be $n$-stable if $\sum_{t=0}^{\infty }\delta^{\widetilde{X}(u)}_{n}(t)<\infty$ (see Wu 2011, p. 3).
\begin{asmpt*}[M1, dependence measure of the stationary approximation]
For some $n>0$ and each $u\in[0,1]$, there exists a measurable function $H(u,\cdot)$ such that $\widetilde{X}_{t}(u)=H(u,\mathcal{F}_{t})$ and $\delta_{n}^{\widetilde{X}}(t):= \sup_{u\in[0,1]} \delta_{n}^{\widetilde{X}(u)}(t)$ fulfils $\sum_{t=0}^{\infty }\delta_{n}^{\widetilde{X}}(t)<\infty.$
\end{asmpt*}
\begin{rmk}
The stationary approximation $\widetilde{X}_{t}(u)$ of the ATV-GARCH model is a standard GARCH process with an intercept given by $\alpha_0 + g_t(u, \boldsymbol{\theta}_1)$. \textcite{BougerolPicard} proved that under a Lyapunov exponent condition (sufficient conditions are (A5) and $\mathbb{E}(\varepsilon_0^2) = 1$), the standard GARCH process admits a stationary and causal solution; see also Subba Rao (2006). It follows that  $\widetilde{X}_{t}(u)=H(u, \mathcal{F}_t)$ in (M1). The general theory in DRW is for a nonlinear process defined by a recursion and then the existence of a measurable function $H(u,\mathcal{F}_{t})$ needs to be explicitly included in (M1).  
\end{rmk}

\subsection*{A.2 Global law of large numbers and central limit theorem}
\label{section:ThmDRW}

We will use the following global law of large numbers (Theorem A.1) and central limit theorem (Theorem A.2) from DRW in our derivations. DRW include an assumption about a kernel in the statements of Theorems A.1 and A.2 below. Since our model is parametric, we will not be requiring this assumption. 
\newtheorem*{thmappendix}{Theorem A.1}
\begin{thmappendix}[DRW, Theorem 2.7(i)] 
Suppose that Assumption (S1) holds with $n=1$ for the process $X_{t,T}$. Then
\begin{equation*}
\frac{1}{T} \sum_{t=1}^{T} X_{t,T} \overset{P}{\to} \int_0^1 \mathbb{E} \widetilde{X}_{0}(u)\ \text{d}u.
\end{equation*}
\end{thmappendix}
DRW proved Theorem A.1 in $L^{1}$. Convergence in $L^{1}$ implies convergence in probability. The theorem therefore holds in probability. 
\newtheorem*{thmappendix2}{Theorem A.2}
\begin{thmappendix2}[DRW, Theorem 2.9] 
Suppose that Assumptions (S1), (S2) and (M1) hold with $n=2$ for the process $X_{t,T}$. Define $S_T = \sum_{t=1}^{T} (X_{t,T} - \mathbb{E}X_{t,T})$. Then
\begin{equation*}
\left\{S_{[Tu]}/ \sqrt{T}, \ 0 \leq u \leq 1 \right\}
\overset{D}{\to} 
\left\{\int_0^u \sigma(v) \ \text{d}W(v), \ 0 \leq u \leq 1 \right\},
\end{equation*}
where $W(v)$ is a standard Brownian motion and the long-run variance $\sigma^2 (v)$ is given by
\begin{equation*}
\sigma^2 (v) = \sum_{k \in \mathbb{Z}} \text{Cov}
(\widetilde{X}_0(v),\widetilde{X}_k(v)).
\end{equation*}
\end{thmappendix2}
In the ATV-GARCH model, we apply Theorem A.2 with $u=1$, in which case
\begin{equation}
    S_T/\sqrt{T} \overset{D}{\to} \int_0^1 \sigma(v) \ \text{d}W(v).
\end{equation}

To prove consistency of the QMLE of the parameters of the ATV-GARCH model, we apply Theorem A.1 to the log-likelihood function $l_{t,T}(\boldsymbol{\theta})$. For this we need (S1) to hold with $n=1$ for the log likelihood function $l_{t,T}(\boldsymbol{\theta})$ and its stationary approximation $\widetilde{l}_t(u, \boldsymbol{\theta})$. In the proof of asymptotic normality, we apply Theorem A.2 to a linear combination of the elements of the score $\mathbf{s}_{t,T}(\boldsymbol{\theta})$. For this, (S1), (S2) and (M1) have to hold with $n=2$ for the elements of the score $\mathbf{s}_{t,T}(\boldsymbol{\theta})$ and for the elements of its stationary approximation 
$\widetilde{\mathbf{s}}_t(u,\boldsymbol{\theta})$. Additionally, (S1) has to hold with $n=1$ for the Hessian $\mathbf{H}_{t,T}(\boldsymbol{\theta})$ and its stationary approximation $\widetilde{\mathbf{H}}_t(u,\boldsymbol{\theta})$. 

\subsection*{A.3 Lemmata}
Lemmas 1--5 contain some intermediate results. 
\begin{lemma}
	The logistic function $G=G(t/T; \gamma, c)$ in (\ref{eq:lfunct}) is Lipschitz continuous in $t/T\in[0,1]$, and its first three partial derivatives with respect to the parameters $c$ and $\gamma$, $\gamma<\infty$, are bounded and Lipschitz continuous in t/T. 
\end{lemma}
\begin{proof}

The partial derivative of $G(t/T; \gamma, c)$ with respect to $t/T$ is given by $\partial G/\partial (t/T)=\gamma G(1-G)$, which achieves its maximum $\gamma/4$ when $G=1/2.$ Therefore, by the mean value theorem, for any points $a,b\in[0,1]$, $G(a)-G(b)\leq \gamma/4(a-b).$

	The first partial derivatives with respect to the parameters $\gamma$ and $c$ are
	\begin{align*}
		\frac{\partial G}{\partial \gamma} &=(t/T-c)G(1-G), \\
		\frac{\partial G}{\partial c} &=-\gamma G(1-G).
	\end{align*}%
	The second partial derivatives are
	\begin{align*}
		\frac{\partial ^{2}G}{\partial \gamma^{2}} &=\frac{\partial }{\partial \gamma}\left\{
		(t/T-c)G(1-G)\right\} =(t/T-c)^{2}G(1-G)(1-2G), \\
		\frac{\partial ^{2}G}{\partial \gamma\partial c} &=\frac{\partial }{\partial c}%
		\left\{ (t/T-c)G(1-G)\right\} =-\gamma(t/T-c)G(1-G)(1-2G)-G(1-G), \\
		\frac{\partial ^{2}G}{\partial c^{2}} &=\frac{\partial }{\partial c}\left\{
		-\gamma G(1-G)\right\} =(-\gamma)^{2}G(1-G)(1-2G).
	\end{align*}%
	The third partial derivatives are%
	\begin{align*}
		\frac{\partial ^{3}G}{\partial \gamma^{3}} &=\frac{\partial }{\partial \gamma}\left\{
		(t/T-c)^{2}G(1-G)(1-2G)\right\} =(t/T-c)^{3}G(1-G)(1-6G+6G^{2}), \\
		\frac{\partial ^{3}G}{\partial \gamma^{2}\partial c} &=\frac{\partial }{\partial
			c}\left\{ (t/T-c)^{2}G(1-G)(1-2G)\right\} =
        -\gamma(t/T-c)^{2}G(1-G)(1-6G+6G^{2}) \\
		&-2(t/T-c)G(1-G)(1-2G), \\
		\frac{\partial ^{3}G}{\partial \gamma\partial c^{2}} &=\frac{\partial }{\partial
			\gamma}\left\{ (-\gamma)^{2}G(1-G)(1-2G)\right\} =(-\gamma)^{2}(t/T-c)G(1-G)(1-6G+6G^{2}) \\
		&-2\gamma G(1-G)(1-2G), \\
		\frac{\partial ^{3}G}{\partial c^{3}} &=\frac{\partial }{\partial c}\left\{
		(-\gamma)^{2}G(1-G)(1-2G)\right\} =(-\gamma)^{3}G(1-G)(1-6G+6G^{2}).
	\end{align*}	
	These functions have bounded derivatives, so they are Lipschitz continuous.
\end{proof}

\begin{lemma}
 Under (A1), Lipschitz continuity of $g(\cdot)$ and (A5), for $n=1$, and additionally under (A6) and (A7), for $n=2$, uniformly in $t=1,\ldots ,T$ and $u,v\in[0,1]$,
 \begin{equation}
		\norm{X^2_{t,T}-\widetilde{X}^2_t(t/T)}_n=\frac{C_1}{T}, \label{eq:tuineq}
	\end{equation}
	\begin{equation}
		\norm{X^2_{t,T}-\widetilde{X}^2_t(u)}_n\leq\frac{C_1}{T} + C_2\left|t/T-u\right| \label{eq:uvineq}
	\end{equation}
	and
	\begin{equation}
		\norm{\widetilde{X}^2_{t}(u)-\widetilde{X}^2_t(v)}_n\leq C_3\left|u-v\right|.\label{eq:uvineq2}
	\end{equation}
	Furthermore,
	\begin{equation}
		\norm{h_{t,T}(\boldsymbol{\theta})-\widetilde{{h}}_{t}(t/T,\boldsymbol{\theta})}_n\leq \frac{C_4}{T},\label{eq:hthu}
	\end{equation}
	\begin{equation}
		\norm{h_{t,T}(\boldsymbol{\theta})-\widetilde{{h}}_{t}(u,\boldsymbol{\theta})}_n\leq \frac{C_4}{T} + C_5\left|t/T-u\right|, \label{eq:htapprox}
	\end{equation}
	and 
	\begin{equation}
		\norm{\widetilde{h}_{t}(u,\boldsymbol{\theta})-\widetilde{h}_{t}(v,\boldsymbol{\theta})}_n\leq C_6\left|u-v\right|. \label{eq:huhv}
	\end{equation}
\end{lemma}

\begin{proof}
	It follows from Subba Rao (2006, Section 5.2) that if $\mathbb{E}(\varepsilon_{t}^{2}) = 1$, and the parameter curves $\{\alpha_{i}(u)\}, i=0,\ldots,p$, and $\{\beta_{j}(u)\}, j=1,\ldots,q$, $u\in[0,1]$, are Lipschitz continuous and satisfy 
	\begin{equation}
		\underset{u}{\sup}\left(\sum\limits _{i=1}^{p}\alpha_{i}(u)+\sum\limits _{j=1}^{q}\beta_{j}(u)\right)<1-\eta, \label{eq:tvstatcond}
	\end{equation} 
	for some $\eta>0$, then the tvGARCH process is locally stationary. In the ATV-GARCH model only $\alpha_0$ is time-varying, and (\ref{eq:tvstatcond}) holds by (A5). It remains to show that the time-varying intercept is Lipschitz continuous, i.e. 
 \begin{equation}
		\left|\alpha_0(u) - \alpha_0(v)\right|\leq C\left|u-v\right|
	\end{equation}
	for some constant $C$. It follows from the mean value theorem that it is sufficient to show that the curve $\{\alpha_0(u)\}$ has bounded first derivatives. This is the case as the logistic function $G(t/T; \gamma, c)$ in (\ref{eq:lfunct}) is Lipschitz continuous by Lemma 1, and the parameter space is compact. The results (\ref{eq:tuineq}), (\ref{eq:uvineq}) and (\ref{eq:uvineq2}) now follow from Subba Rao (2006, Theorem 2.1 and Proposition 2.1).	
Next consider (\ref{eq:htapprox}). From (\ref{eq:decrecht}), we have
\begin{equation*}
	h_{t,T}(\boldsymbol{\theta})=c_{0}(\boldsymbol{\theta})
	+ \sum^\infty_{i=1}d_i(\boldsymbol{\theta})g_{t-i+1,T}(\boldsymbol{\theta}_1) + \sum^{\infty}_{i=1}c_i(\boldsymbol{\theta})X^2_{t-i,T},
\end{equation*}
and
\begin{equation*}
\widetilde{h}_t(u,\boldsymbol{\theta})=c_{0}(\boldsymbol{\theta})
	+ \sum^\infty_{i=1}d_i(\boldsymbol{\theta})g_{t-i+1}(u,\boldsymbol{\theta}_1) + \sum^{\infty}_{i=1}c_i(\boldsymbol{\theta})\widetilde{X}^{2}_{t-i}(u),
\end{equation*}
where $\boldsymbol{\theta}=(\boldsymbol{\theta}_{1}^{\intercal},\boldsymbol{\theta}_{2}^{\intercal})^{\intercal}$ is the partitioning of the parameter space defined in Section 3.1. BHK, Lemma 3.1, formula (3.4) gives $	c_i(\boldsymbol{\theta})\leq C_{1}\rho_0^{i/q}$
for a constant $C_{1} < \infty$ and $0<\rho_0<1$. Mimicking the proof of this lemma, it is straightforward to show that the coefficients $d_{i}(\boldsymbol{\theta })$ satisfy a similar formula: $d_i(\boldsymbol{\theta})\leq C_{2}\rho_0^{i/q}$  for $C_{2} < \infty$ and $0<\rho_0<1$.
Since the coefficients $c_i(\boldsymbol{\theta})$ and $d_i(\boldsymbol{\theta})$ are absolutely summable, \begin{align}
		\left|h_{t,T}(\boldsymbol{\theta})-\widetilde{h}_{t}(u,\boldsymbol{\theta})\right|  &\leq
 C_2
\sum_{i=1}^{\infty }\rho _{0}^{i/q}\left\vert g_{t-i+1,T}(\boldsymbol{\theta }_{1})-g_{t-i+1}(u,\boldsymbol{\theta }_{1})\right\vert \nonumber \\
		&+ C_{1}\sum^\infty_{i=1}\rho^{i/q}_0\left| X^2_{t-i,T}-\widetilde{X}^2_{t-i}(u) \right|. 
\label{eq:htdiffcicoeff}
 \end{align}
By Lemma 1, the function $g_{t,T}(\boldsymbol{\theta}_{1})$ is bounded and Lipschitz continuous. We therefore obtain for the first term on the right-hand side of (\ref{eq:htdiffcicoeff})
\begin{equation*}
\left| g_{t-i+1,T}(\boldsymbol{\theta }_{1})-g_{t-i+1}(u,\boldsymbol{\theta }_{1})\right| 
\leq C_{5}\left|\frac{t}{T}-u \right|.
\end{equation*}
For the second term, similarly to an argument in the proof of \textcite{kristensen2019local}, Theorem
6,  adding and subtracting $\widetilde{X}^2_{t-i}\left((t-i)/T\right)$ and $\widetilde{X}^2_{t-i}\left(t/T\right)$, and applying the triangle inequality yields for each index $t-i$, 
	\begin{align*}
		\left|X^2_{t-i,T}-\widetilde{X}^2_{t-i}(u)\right| &\leq \left|X^2_{t-i,T}-\widetilde{X}^2_{t-i}\left(\frac{t-i}{T}\right)\right| + \left|\widetilde{X}^2_{t-i}\left(\frac{t}{T}\right)-\widetilde{X}^2_{t-i}\left(\frac{t-i}{T}\right)\right| \\
		&+ \left|\widetilde{X}^2_{t-i}\left(\frac{t}{T}\right)-\widetilde{X}^2_{t-i}(u)\right|.\\
	\end{align*}
	Taking norms and using (\ref{eq:tuineq})--(\ref{eq:uvineq2}) gives
		\begin{equation*}
	\norm{X^2_{t-i,T}-\widetilde{X}^2_{t-i}(u)}_n\leq \frac{C_6}{T} + C_7\frac{i}{T} + C_8\left|\frac{t}{T}-u\right|.
	\end{equation*}
	Collecting results, we have
	\begin{align*}
		\norm{h_{t,T}(\boldsymbol{\theta})-\widetilde{h}_{t}(u,\boldsymbol{\theta})}_n
		& \leq C_{2}\sum_{i=1}^{\infty }\rho _{0}^{i/q}
        \left(\frac{C_3}{T} + C_4\frac{i}{T} + C_5\left|\frac{t}{T}-u\right|\right)\\
        & +  C_{1}\sum^\infty_{i=1}\rho^{i/q}_0 \left(\frac{C_6}{T} + C_7\frac{i}{T} + C_8\left|\frac{t}{T}-u\right|\right)\\
		&\leq \frac{C_9}{T} + C_{10}\left|\frac{t}{T}-u\right|.
	\end{align*}
We have used monotone convergence and the fact that for $0<\rho_0<1$, the arithmetic-geometric series $\sum^\infty_{i=1}\rho^{i/q}_0i<\infty$. 
This proves (\ref{eq:htapprox}). The result (\ref{eq:hthu}) follows from (\ref{eq:htapprox}) and the triangle inequality. The same arguments can be used to establish (\ref{eq:huhv}). 
\end{proof}

The following lemma generalises Lemma 5.1 in BHK to the ATV-GARCH$(p,q)$ model.
\begin{lemma}
Assume that (A1)--(A3) hold and $\mathbb{E}|\varepsilon_0^2|^d<\infty$ for some $d>0$. Then for any $0<\upsilon<d$, 
\begin{equation}
    \mathbb{E}\left\{\sup_{\boldsymbol{\theta}\in\Theta}\frac{h_{t,T}(\boldsymbol{\theta}_0)}{h_{t,T}(\boldsymbol{\theta})}\right\}^{\upsilon}<\infty. \label{eq:htratio} 
\end{equation}
\end{lemma}

\begin{proof}
Let $\{X_{t,T}^{\ast}\}$ be a stationary process that bounds the locally stationary process $\{X_{t,T}\}$, i.e. 
\begin{align}
    X_{t,T}^* &= \sigma_{t,T}^{*} \varepsilon_t, \\
	\sigma_{t,T}^{*2} &= g^* + \sum_{i=1}^{p}\alpha_{i}X_{t-i,T}^{*2}+\sum_{j=1}^{q}\beta_{j}\sigma_{t-j,T}^{*2},\label{eq:sigmastar}
\end{align}
where $g^*:=\sup_{u\in[0,1]}g(u)$. Since $G(u, \gamma, \boldsymbol{c})$ is bounded from above by $1$ and the parameter space is compact, the process $\{X^{*}_{t,T}\}$ is well-defined. Further, by (A5), the process $\{X^{*}_{t,T}\}$ is stationary. We have
\begin{align}
\mathbb{E}\left\{\sup_{\boldsymbol{\theta}\in\Theta}\frac{h_{t,T}(\boldsymbol{\theta}_0)}{h_{t,T}(\boldsymbol{\theta})}\right\}^{\upsilon}
&\leq \mathbb{E}\left\{\sup_{\boldsymbol{\theta}\in\Theta}\frac{h_{t,T}^{\ast}(\boldsymbol{\theta}_0)}{h_{t,T}(\boldsymbol{\theta})}\right\}^{\upsilon} \nonumber \\ 
&\leq  \mathbb{E}\left\{\sup_{\boldsymbol{\theta}\in\Theta}\frac{h_{t,T}^*(\boldsymbol{\theta}_0)}{\widetilde{h}_t(u,\boldsymbol{\theta})}\right\}^{\upsilon}<\infty. \label{eq:boundedlsratio} 
\end{align}
The first inequality in (\ref{eq:boundedlsratio}) is due to the fact that $h_{t,T}^{\ast}(\boldsymbol{\theta})$ is the conditional variance of a stationary GARCH($p,q$) process with an intercept that bounds the time-varying intercept in (\ref{eq:tvintercept}). The second inequality follows, because by (\ref{eq:parameterspace1}), $\alpha_{0}+\sum\limits _{l=1}^{L}\alpha{}_{0l}G_l\left(u,\gamma_l,\boldsymbol{c}_l\right)\geq\inf_{\boldsymbol{\theta}\in\Theta}\alpha_0$. The third member of (\ref{eq:boundedlsratio}) is a ratio of conditional variances of two stationary GARCH$(p,q)$ processes. The result now follows from Lemma 5.1 of BHK. 
\end{proof}

Next we prove similar results for the derivatives of $h_{t,T}(\boldsymbol{\theta})$.

\begin{lemma}
	Assume that (A1)--(A5) hold. Then
	\begin{equation}     
		\mathbb{E}\sup_{\boldsymbol{\theta}\in\Theta}\left|\frac{1}{h_{t,T}(\boldsymbol{\theta})}\frac{\partial h_{t,T}(\boldsymbol{\theta})}{\partial\boldsymbol{\theta}}\right|^{\upsilon}<\infty, \label{eq:boundedh1deriv}
	\end{equation}

	\begin{equation}
		\mathbb{E}\sup_{\boldsymbol{\theta}\in\Theta}\left|\frac{1}{h_{t,T}(\boldsymbol{\theta})}\frac{\partial^2 h_{t,T}(\boldsymbol{\theta})}{\partial\boldsymbol{\theta}\partial\boldsymbol{\theta}^{T}}\right|^{\upsilon}<\infty \label{eq:boundedh2deriv}
	\end{equation}

and
	\begin{equation}
	\mathbb{E}\sup_{\boldsymbol{\theta}\in\Theta}\left|\frac{1}{h_{t,T}(\boldsymbol{\theta})}\frac{\partial^3 h_{t,T}(\boldsymbol{\theta})}{\partial \theta_i\partial \theta_j\partial \theta}_k\right|^{\upsilon}<\infty \label{eq:boundedh3deriv}
\end{equation}

for any $\upsilon>0.$
\end{lemma}
\begin{proof}
We begin by proving (\ref{eq:boundedh1deriv}). By (19), we have
\begin{equation}
\frac{\partial h_{t,T}(\boldsymbol{\theta })}{\partial \boldsymbol{\theta }}=\frac{%
\partial c_{0}(\boldsymbol{\theta })}{\partial \boldsymbol{\theta }}
+\sum_{i=1}^{\infty }\frac{\partial d_{i}(\boldsymbol{\theta })}{\partial 
\boldsymbol{\theta }}g_{t-i+1,T}(\boldsymbol{\theta }_{1})+\sum_{i=1}^{\infty }d_{i}(%
\boldsymbol{\theta })\frac{\partial g_{t-i+1,T}(\boldsymbol{\theta }_{1})}{\partial 
\boldsymbol{\theta }}+\sum_{i=1}^{\infty }\frac{\partial c_{i}(\boldsymbol{
\theta })}{\partial \boldsymbol{\theta }}X_{t-i,T}^{2}. \label{eq:lemma31}
\end{equation}
The proof consists of bounding the terms on the right-hand side of (\ref{eq:lemma31}). 
For the first term, because $c_{0}(\boldsymbol{\theta })=\alpha _{0}/(1-\sum_{i=1}^{q}\beta _{i})$, it
holds (cf. BHK, formula 3.6-3.8)
\begin{equation}
\sup_{\boldsymbol{\theta} \in \Theta }\left\vert \frac{\partial c_{0}(\boldsymbol{\theta })}{%
\partial \boldsymbol{\theta }}\right\vert \leq \frac{C_{1}}{(1-\rho _{0})^{2}}, \label{eq:lemma32}
\end{equation}
where $\sum_{i=1}^{q}\beta _{i}\leq \rho _{0}<1$. 

The derivatives $\partial c_{i}(\boldsymbol{\theta })/\partial \boldsymbol{\theta }_{1}=\boldsymbol{0}$, and by Lemma 3.2 of BHK, formula (3.9),
\begin{equation}
\frac{1}{c_{i}(\boldsymbol{\theta})} 
\left\vert \frac{\partial c_{i}(\boldsymbol{\theta })}{\partial \boldsymbol{\theta }%
_{2}}\right\vert \leq C_{2}i,\quad 1\leq i<\infty, \label{eq:lemma33}
\end{equation}
for some constant $C_2$. The derivatives $\partial d_{i}(\boldsymbol{\theta })/\partial \boldsymbol{\theta }_{1}=\boldsymbol{0}$. Mimicking the proof of this lemma, it is straightforward to show that the derivatives $\partial d_{i}(\boldsymbol{\theta })/\partial 
\boldsymbol{\theta }_{2}$ satisfy a similar formula:
\begin{equation}
\frac{1}{d_{i}(\boldsymbol{\theta })} \left\vert \frac{\partial d_{i}(\boldsymbol{\theta })}{\partial \boldsymbol{\theta }%
_{2}}\right\vert \leq C_{3}i,\quad 1\leq i<\infty, \label{eq:lemma34}
\end{equation}
for some constant $C_3$. The function $g_{t,T}(\boldsymbol{\theta }_{1})$ is bounded, and by Lemma 1 its derivatives are bounded. Now, applying (\ref{eq:lemma34}) to the second term on the right-hand side of (\ref{eq:lemma31}) yields
\begin{align}
\sup_{\boldsymbol{\theta} \in \Theta }\sum_{i=1}^{\infty }\left\vert \frac{\partial d_{i}(%
\boldsymbol{\theta })}{\partial \boldsymbol{\theta }}\right\vert g_{t-i+1,T}(\boldsymbol{\theta }_1) &\leq
C_{3}\sum_{i=1}^{\infty }id_{i}(\boldsymbol{\theta })g_{t-i+1,T}(\boldsymbol{\theta }_1) \label{eq:lemma35} \nonumber \\
&\leq C_{4}\sum_{i=1}^{\infty }i\rho _{0}^{i/q}g_{t-i+1,T}(\boldsymbol{\theta }_1) \nonumber \\
&<C_{5}.   
\end{align}
Similarly, applying (\ref{eq:dineq}) to the third term gives
\begin{align}
\sup_{\theta \in \Theta }\sum_{i=1}^{\infty }d_{i}(\boldsymbol{\theta }%
)\left\vert \frac{\partial g_{t-i+1,T}(\boldsymbol{\theta }_{1})}{\partial \boldsymbol{%
\theta }_{1}}\right\vert  &\leq C_{6}\sum_{i=1}^{\infty }\rho
_{0}^{i/q}\left\vert \frac{\partial g_{t-i+1,T}(\boldsymbol{\theta }_{1})}{%
\partial \boldsymbol{\theta }_{1}}\right\vert \label{eq:lemma36} \\
&< C_{7}. \nonumber 
\end{align}
Applying (\ref{eq:lemma32}), (\ref{eq:lemma33}), (\ref{eq:lemma35}) and (\ref{eq:lemma36}) to (\ref{eq:lemma31}), we obtain
\begin{equation}
\left\vert \frac{1}{h_{t,T}(\boldsymbol{\theta )}}\frac{\partial h_{t,T}(\boldsymbol{%
\theta )}}{\partial \boldsymbol{\theta }}\right\vert \leq C_{8}+ C_{9} \left\vert \frac{%
\sum_{i=1}^{\infty }ic_{i}(\boldsymbol{\theta })X_{t-i,T}^{2}}{c_0(\boldsymbol{\theta }) + d_0(\boldsymbol{\theta }) +\sum_{i=1}^{%
\infty }c_{i}(\boldsymbol{\theta })X_{t-i,T}^{2}}\right\vert. 
\end{equation}
Exploiting the inequality $x/(1+x)\leq x^{s}$ for all $x\geq 0$ and $s\in
(0,1)$ as in \textcite{boussama2000} and \textcite{FrancqZakoian2004}, and Jensen's inequality $(a+b)^{s}\leq a^{s}+b^{s}$, yields
 \begin{align*}
		\mathbb{E}\sup_{\boldsymbol{\theta}\in\Theta}\left|\frac{1}{h_{t,T}(\boldsymbol{\theta})}\frac{\partial h_{t,T}(\boldsymbol{\theta})}{\partial\boldsymbol{\theta}}\right|
		&\leq C_8 + C_{9}\mathbb{E}\left(\sup_{\boldsymbol{\theta}\in\Theta} \frac{\sum_{i=1}^\infty i c_{i}(\boldsymbol{\theta}) X_{t-i,T}^{2}}{c_0(\boldsymbol{\theta }) + d_0(\boldsymbol{\theta })
  +\sum_{i=1}^\infty c_{i}(\boldsymbol{\theta}) X_{t-i,T}^{2}}\right)\\
		&\leq C_8 + C_{9}\mathbb{E}\left(\sup_{\boldsymbol{\theta}\in\Theta} \sum_{i=1}^\infty i \left( \frac{c_i(\boldsymbol{\theta})X_{t-i,T}^{2}}{c_0(\boldsymbol{\theta }) + d_0(\boldsymbol{\theta })} \right) ^{s} \right) \\
		&\leq C_8 + C_{10}\mathbb{E}\left(\sup_{\boldsymbol{\theta}\in\Theta} \sum_{i=1}^\infty i\rho_0^{is/q}X_{t-i,T}^{2s}\right)\\
		&\leq C_8 + C_{10}\mathbb{E}\left(\sup_{\boldsymbol{\theta}\in\Theta} \sum_{i=1}^\infty i\rho_0^{is/q}X_{t-i,T}^{*2s}\right)\\
		& \leq C_8 + C_{10}\sup_{\boldsymbol{\theta}\in\Theta} \sum_{i=1}^\infty i\rho_0^{is/q}\mathbb{E}\left(\sup_{\boldsymbol{\theta}\in\Theta}X_{t-i,T}^{*2s}\right)\\
		&<\infty,
	\end{align*}
	for all $s\in (0,1]$ by weak stationarity of the bounding process $\{X_{t,T}^{\ast}\}$. Note that for large but finite values of $\upsilon$, we can choose $s$ close to zero. This proves (\ref{eq:boundedh1deriv}). The proofs of (\ref{eq:boundedh2deriv}) and (\ref{eq:boundedh3deriv}) follow using similar arguments, boundedness of the function $g_{t,T}(%
\boldsymbol{\theta }_{1})$, Lemma 1 and Lemmas 3.1--3.3 of BHK.
\end{proof}
\begin{lemma}
\label{lemma:gpd}
The partial derivatives corresponding to a time-varying intercept $\alpha_0(u,\boldsymbol{\theta}_1) = \alpha_0 + \alpha_{01}G(u,\boldsymbol{\theta}_1)$ consisting of one logistic transition function with one $c$ parameter are linearly independent functions in $u$.
\end{lemma}
\begin{proof}
    
Suppose for contradiction that there exists $\boldsymbol{\lambda} = (\lambda_0, \lambda_1, \lambda_2, \lambda_3)^\intercal \neq \boldsymbol{0}$ such that
\begin{equation*}
\boldsymbol{\lambda}^\intercal\frac{\partial \alpha_0(u,\boldsymbol{\theta}_1)}{\partial \boldsymbol{\theta}_1} = 0.
\end{equation*}
Suppress for notational convenience the dependence of $G$ on $u$.
Then it holds 
\begin{equation}
-\lambda_0 -\lambda_1 G = \lambda_2\alpha_{01}(u-c)G(1-G) + \lambda_3 \alpha_{01}(-\gamma)G(1-G). \label{eq:monotone}
\end{equation}
Then the function \begin{align*}
f &= \lambda_2\alpha_{01}(u-c)G(1-G) + \lambda_3 \alpha_{01}(-\gamma)G(1-G)
\end{align*} is either monotone or constant in $u$ as the LHS of (\ref{eq:monotone}) is a monotone function if $\lambda_1 \neq 0$ and constant if $\lambda_1 = 0$. We begin by showing that $\lambda_1 \neq 0$ leads to a contradiction. Let $b_1 = \lambda_2\alpha_{01}$ and $b_2 = \lambda_3\alpha_{01}(-\gamma).$ We have that \begin{equation}
\frac{\partial f}{\partial u} = b_1G(1-G) + \gamma b_1(u-c)G(1-G)(1-2G) + \gamma b_2G(1-G)(1-2G). \label{eq:fderiv}
\end{equation}
Note that, evaluated at $c$, we have $\partial f / \partial u =0.25b_1.$ In order for the function to be monotone, the derivative has to have the same sign everywhere. Therefore, $b_1$ completetely determines whether the function is monotonically increasing or decreasing. Suppose without loss of generality that $b_1>0.$ Then $f$ is monotonically increasing. Note that $G(1-G)$ is increasing for $u<c.$ Consider again \begin{equation}
f(u) = b_1(u-c)G(1-G) + b_2G(1-G).
\end{equation}
This function is increasing for $u<c$ if $b_1(u-c) + b_2>0$. Letting $u\to c$ then shows $b_2>0,$ as if $b_2<0$ we can pick $u$ so that $b_1(u-c) + b_2 < 0.$ Note now that $G(1-G)$ is decreasing for $u>c.$ 
But then $b_1(u-c)>0$ and $b_2>0$ implies that $f$ is decreasing for $u>c,$ which contradicts $f$ monotone. 

Now suppose $\lambda_1 = 0$ but $\lambda_0 >0.$ Then $f$ is constant. Then $\partial f / \partial u$ is zero everywhere. Evaluating again this derivative at $u=c$, we see that $b_1 =0$ in (\ref{eq:fderiv}), and consequently we have to have $b_2 = 0$ as well. Then the RHS of (\ref{eq:monotone}) is zero, which in turn implies $\lambda_0 =0.$ Now, it follows easily from (\ref{eq:monotone}) that $\lambda_2$ and $\lambda_3$ are $0$ as well. This concludes the proof. 

\end{proof}

\subsection*{A.4 Proof of Theorem 1}

We prove consistency of the estimator $\widehat{\boldsymbol{\theta}}_{T}$. Consistency of the truncated estimator $\bar{\boldsymbol{\theta}}_{T}$ will be proved in Appendix A.6. We make use of the following theorem.  
\newtheorem*{thmappendix3}{Theorem A.3}
\begin{thmappendix3}[\cite{Amemiya1985}, Theorem 4.1.1]

	Suppose that \begin{itemize}
		\item[(C1)] The parameter space $\Theta$ is a compact subset of $\mathbb{R}^k$.
		\item[(C2)] The objective function $Q_T(\boldsymbol{\theta})$ is a measurable function of the data for all $\boldsymbol{\theta}\in\Theta$, and $Q_T(\boldsymbol{\theta})$ is continuous in $\boldsymbol{\theta}\in\Theta$.
		\item[(C3)] $Q_T(\boldsymbol{\theta})$ converges uniformly in probability to a non-stochastic function $Q(\boldsymbol{\theta})$, and $Q(\boldsymbol{\theta})$ attains a unique global maximum at $\boldsymbol{\theta}_0$.
	\end{itemize}
Then $$\widehat{\boldsymbol{\theta}}_T\overset{P}{\to}\boldsymbol{\theta}_0.$$
	\end{thmappendix3}
\noindent \textbf{Proof of Theorem 1.} In order to prove Theorem 1, we verify conditions (C1)--(C3) in Theorem A.3. In our case, $Q_T(\boldsymbol{\theta})=L_T(\boldsymbol{\theta})$ 
defined in (\ref{loglik}). Condition (C1) holds by (A2). Clearly, $L_T(\boldsymbol{\theta})$ is a measurable function of the data. Moreover, it is Lipschitz continuous in $\boldsymbol{\theta}$, which implies that it is continuous in $\boldsymbol{\theta}$, so (C2) is satisfied. To verify (C3), we first use Theorem A.1 to show that the log-likelihood function converges in probability to its expectation. Compactness of the parameter space and stochastic equicontinuity are then sufficient conditions for uniform convergence. Finally, uniqueness of the global maximum at $\boldsymbol{\theta}_0$ follows from a widely used argument for GARCH-type log-likelihood functions.

To employ Theorem A.1, we have to verify Assumption (S1) with $n=1$ for the log-likelihood function $l_{t,T}(\boldsymbol{\theta})$ and its stationary approximation $\widetilde{l}_t(u, \boldsymbol{\theta})$. We begin by proving two lemmas that will be used for the purpose.
\begin{lemma} \label{lemma:lls1}
Uniformly in $t=1,\ldots,T$ and $u\in[0,1]$,
	\begin{equation}
	\sup_{\boldsymbol{\theta}\in\Theta}\norm{l_{t,T}(\boldsymbol{\theta})-\widetilde{l}_t(u, \boldsymbol{\theta})}_1\leq \frac{C_1}{T} + C_2\left|\frac{t}{T}-u\right|. \label{eq:datlip}
	\end{equation}
\end{lemma}

\begin{rmk}
DRW and Kristensen and Lee (2019) proved invariance principles for locally stationary processes.  Theorem 6 in Kristensen and Lee (2019) contains such an invariance principle for time-varying GARCH-type models. We verify (S1) and (S2) for the log-likelihood function directly in order to avoid introducing notation and terminology that would otherwise not be needed in this paper.    
\end{rmk}

\begin{proof}
By the triangle inequality,
	\begin{equation}
		\norm{l_t(\boldsymbol{\theta})-l^{\prime}_t(\boldsymbol{\theta})}_1 \leq\norm{\ln h_t-\ln h^{\prime}_{t}}_1 + \norm{\frac{{X}^2_t}{h_t}-\frac{X^{\prime2}_t}{h_t^{\prime}}}_1. \label{eq:logliktriangle}
	\end{equation}
First examine $\norm{\ln h_t-\ln h^{\prime}_{t}}_1$. Consider $\left| \ln x - \ln y \right|$ and assume that $y$ is bounded from below by some constant $C>0$. Then, as $\ln (1+x)\leq x$ for all $x>-1$,
	\begin{align*}
		\left| \ln x - \ln y \right|
		&=\left|\ln \left( 1+\left( \frac{%
			x}{y}-1\right) \right)\right|\\ &\leq \left|\frac{x}{y}-1\right| \leq \frac{1}{C}\left|x-y\right|.
	\end{align*}
Consequently, for $h_t$ bounded from below by $C>0$,
	\begin{align}
		\norm{\ln h_t-\ln h_t^{\prime}}_1&\leq \frac{1}{C}\norm{h_t-h_t^{\prime}}_1.
  \label{eq:loglikh}
	\end{align}
Next, take the second term on the right-hand side of (\ref{eq:logliktriangle}). By adding and subtracting $X_t^{\prime2}/ h_t$ and using the triangle inequality, we obtain \begin{align}
		\norm{\frac{{X}^2_t}{{h}_t}-\frac{X^{\prime2}_t}{{h}_t^{\prime}}}_1
		&\leq \norm{\frac{1}{h_t}\left(X_t^2 - X_t^{\prime 2}\right)}_1 +  \norm{\frac{X_t^{\prime2}}{h_t h_t^{\prime}}\left(h_t - h_t^{\prime}\right)}_1.
  \label{eq:loglikx}
	\end{align}
Now, set $h_{t} = h_{t,T}(\boldsymbol{\theta})$ and $h^{\prime}_t = \widetilde{h}_t(u,\boldsymbol{\theta})$ in (\ref{eq:loglikh}), and take the sup norm of the resulting function. Further, as $\widetilde{h}_t(u,\boldsymbol{\theta})>0$ and bounded from below by a constant by (A3),
\begin{equation}
	 \frac{1}{\widetilde{h}_t(u,\boldsymbol{\theta})}<\infty, \label{eq:boundedbelow}
\end{equation}
so it follows that
\begin{align}
		\sup_{\boldsymbol{\theta}\in\Theta} \norm{\ln h_{t,T}(\boldsymbol{\theta})-\ln \widetilde{h}_t(u,\boldsymbol{\theta})}_1&\leq \frac{1}{C}\sup_{\boldsymbol{\theta}\in\Theta}\norm{h_{t,T}(\boldsymbol{\theta})-\widetilde{h}_t(u,\boldsymbol{\theta})}_1. \label{eq:A1}
	\end{align}
Similarly, setting $X^{2}_t=X_{t,T}^{2}$ and 
$X_t^{\prime 2} = \widetilde{X}_t^2 (u)$ in (\ref{eq:loglikx}) and taking the sup norm leads to
\begin{align}
		\sup_{\boldsymbol{\theta}\in\Theta}\norm{\frac{{X}^2_{t,T}}{{h}_{t,T}(\boldsymbol{\theta})}-\frac{\widetilde{X}^2_t(u)}{\widetilde{h}_t(u,\boldsymbol{\theta})}}_1\leq\sup_{\boldsymbol{\theta}\in\Theta}\norm{\frac{1}{h_{t,T}(\boldsymbol{\theta})}\left({X}^2_{t,T}-\widetilde{X}^2_t(u)\right)}_1 +  \sup_{\boldsymbol{\theta}\in\Theta}\norm{\frac{\widetilde{X}_t^{2}(u)}{\widetilde{h}_t(u,\boldsymbol{\theta})}\frac{\left(h_{t,T}(\boldsymbol{\theta}) - \widetilde{h}_t(u,\boldsymbol{\theta})\right)}{h_{t,T}(\boldsymbol{\theta})}}_1. \label{eq:A2}
	\end{align}
Using (\ref{eq:boundedbelow}), we have for the first term on the right-hand side of (\ref{eq:A2}),
\begin{align}
\sup_{\boldsymbol{\theta}\in\Theta}\norm{\frac{1}{h_{t,T}(\boldsymbol{\theta})}\left({X}^2_{t,T}-\widetilde{X}^2_t(u)\right)}_1
\leq
\frac{1}{C} \norm{X_{t,T}^2 - \widetilde{X}_t^2 (u)}_1.
\label{eq:A22}
\end{align}
Applying Hölder's inequality to the second term yields 
\begin{align}
\sup_{\boldsymbol{\theta}\in\Theta}\norm{\frac{\widetilde{X}_t^{2}(u)}{\widetilde{h}_t(u,\boldsymbol{\theta})}\frac{\left(h_{t,T}(\boldsymbol{\theta}) - \widetilde{h}_t(u,\boldsymbol{\theta})\right)}{h_{t,T}(\boldsymbol{\theta})}}_1
\leq
\sup_{\boldsymbol{\theta}\in\Theta}\norm{\frac{\widetilde{X}_t^{2}(u)}{\widetilde{h}_t(u,\boldsymbol{\theta})}}_p\sup_{\boldsymbol{\theta}\in\Theta}\norm{\frac{\left(h_{t,T}(\boldsymbol{\theta}) - \widetilde{h}_t(u,\boldsymbol{\theta})\right)}{h_{t,T}(\boldsymbol{\theta})}}_q. \label{eq:A3}
\end{align}
Let $p=1+\delta/2$ with $\delta>0$ in (\ref{eq:A3}). By Lemma 5.1 of BHK, (A1) and independence of $\varepsilon_t$ and $\widetilde{h}_t(u,\boldsymbol{\theta})$,
\begin{align*}
\mathbb{E}\left(
\sup_{\boldsymbol{\theta}\in\Theta}\frac{\widetilde{X}_t^{2}(u)}{\widetilde{h}_t(u,\boldsymbol{\theta})}
\right)^{1+\delta/2}
&=
\mathbb{E}\left(
\varepsilon_t^2\sup_{\boldsymbol{\theta}\in\Theta}\frac{\widetilde{h}_t(u, \boldsymbol{\theta}_0)}{\widetilde{h}_t(u,\boldsymbol{\theta})}
\right)^{1+\delta/2}\\
&=
\mathbb{E}(\varepsilon_t^{2+\delta})
\mathbb{E}\left(
\sup_{\boldsymbol{\theta}\in\Theta}\frac{\widetilde{h}_t(u, \boldsymbol{\theta}_0)}{\widetilde{h}_t(u,\boldsymbol{\theta})}
\right)^{1+\delta/2}
< \infty.
\end{align*}
Since the arguments in (\ref{eq:loglikx}) remain valid if we instead of $X_t^{\prime2}/ h_t$ add and subtract $X_t^{2}/ h_t^{\prime}$ in (\ref{eq:logliktriangle}), we may without loss of generality assume that $h_{t,T}(\boldsymbol{\theta})\geq \widetilde{h}_t(u,\boldsymbol{\theta}).$ Then, since $\left(h_{t,T}(\boldsymbol{\theta}) - \widetilde{h}_t(u,\boldsymbol{\theta})\right)/h_{t,T}(\boldsymbol{\theta})\in[0,1)$, we have
\begin{equation}
\left(
\frac{h_{t,T}(\boldsymbol{\theta}) - \widetilde{h}_t(u,\boldsymbol{\theta})}{h_{t,T}(\boldsymbol{\theta})}
\right)^q
\leq
\frac{h_{t,T}(\boldsymbol{\theta}) - \widetilde{h}_t(u,\boldsymbol{\theta})}{h_{t,T}(\boldsymbol{\theta})} 
\label{eq:holderq}
\end{equation}
for all $q \geq 1$.
From (\ref{eq:boundedbelow})  and (\ref{eq:holderq})  it follows that
\begin{equation}
\sup_{\boldsymbol{\theta}\in\Theta}\norm{\frac{\left(h_{t,T}(\boldsymbol{\theta}) - \widetilde{h}_t(u,\boldsymbol{\theta})\right)}{h_{t,T}(\boldsymbol{\theta})}}_q    
\leq
\sup_{\boldsymbol{\theta}\in\Theta}\norm{\frac{\left(h_{t,T}(\boldsymbol{\theta}) - \widetilde{h}_t(u,\boldsymbol{\theta})\right)}{h_{t,T}(\boldsymbol{\theta})}}_1
\leq
\frac{1}{C} \sup_{\boldsymbol{\theta}\in\Theta}\norm{
h_{t,T}(\boldsymbol{\theta}) - \widetilde{h}_t(u,\boldsymbol{\theta})}_1.
\label{eq:A222}
\end{equation}
In view of (\ref{eq:A1}), (\ref{eq:A22}) and (\ref{eq:A222}), we have
\begin{align}
\sup_{\boldsymbol{\theta}\in\Theta}\norm{l_{t,T}(\boldsymbol{\theta})-\widetilde{l}_t(u,\boldsymbol{\theta})}_1 
&\leq \frac{1}{C_{1}}\sup_{\boldsymbol{\theta}\in\Theta}\norm{h_{t,T}(\boldsymbol{\theta})-\widetilde{h}_t(u,\boldsymbol{\theta})}_1 \nonumber \\ 
& + 
\frac{1}{C_{2}}\sup_{\boldsymbol{\theta}\in\Theta}\norm{{X}^2_{t,T}-\widetilde{X}^2_t(u)}_1. \label{eq:loglikfinal}
\end{align}
Applying (\ref{eq:uvineq}) and (\ref{eq:htapprox}) in Lemma 2 to the terms on the right-hand side of (\ref{eq:loglikfinal}) completes the proof.
\end{proof}

\begin{lemma} \label{lemma:ll5}
	The inequality $\sup_{u\in[0,1]}\norm{\widetilde{l}_t(u,\boldsymbol{\theta})}_1<\infty$ holds.
\end{lemma}
\begin{proof}
Lemma 5.3 in BHK implies that $\norm{\widetilde{l}_t(u,\boldsymbol{\theta})}_1<\infty$ for all $u \in [0,1]$, so the result also holds for the supremum.
\end{proof}
Lemmas \ref{lemma:lls1} and \ref{lemma:ll5} together imply that (S1) is satisfied with $n=1$ for the log likelihood function.
The following lemma strengthens pointwise convergence to uniform convergence.

\begin{lemma} \label{lemma:ll6}
	The log-likelihood function $L_T(\boldsymbol{\theta}):= (1/T)\sum^T_{t=1}l_{t,T}(\boldsymbol{\theta})$ is stochastically equicontinuous (SE). 
\end{lemma}
\begin{proof}
By an application of the mean value theorem, we get 	$$\left|l_{t,T}(\boldsymbol{\theta}^{'})-l_{t,T}(\boldsymbol{\theta})\right|
\leq \sup_{\boldsymbol{\theta}\in\Theta} K_{t,T}(\boldsymbol{\theta}_m)|\boldsymbol{\theta}^{'}-\boldsymbol{\theta}|,$$
where 
\begin{equation}
    K_{t,T}(\boldsymbol{\theta}_m)=\left|\frac{1}{h_{t,T}(\boldsymbol{\theta}_m)}\frac{\partial h_{t,T}(\boldsymbol{\theta}_m)}{\partial \boldsymbol{\theta}}\left(1+\frac{X^2_{t,T}}{h_{t,T}(\boldsymbol{\theta}_m)}\right)\right| \label{eq:Kn}
	\end{equation}
and $\boldsymbol{\theta}_m$ lies on the line segment connecting $\boldsymbol{\theta}^{'}$ and $\boldsymbol{\theta}$. Let $K_{t,T} = \sup_{\boldsymbol{\theta}\in\Theta} K_{t,T}(\boldsymbol{\theta})$. 
We show that the following sufficient condition for SE in Andrews (1992, Lemma 2):
\begin{equation}
	    \sup_{T\geq 1}\frac{1}{T}\sum_{t=1}^{T}\mathbb{E}(K_{t,T})<\infty
     \label{eq:SE}
\end{equation} 
is satisfied. Hölder's inequality applied to $K_{t,T}$ yields 
\begin{align}
\mathbb{E}(K_{t,T}(\boldsymbol{\theta}))
&= 
\mathbb{E}\sup_{\boldsymbol{\theta}\in\Theta}\left|\frac{1}{h_{t,T}(\boldsymbol{\theta})}\frac{\partial h_{t,T}(\boldsymbol{\theta})}{\partial \boldsymbol{\theta}}\left(1+\frac{X^2_{t,T}}{h_{t,T}(\boldsymbol{\theta})}\right)\right| \nonumber \\
&\leq
\mathbb{E}\sup_{\boldsymbol{\theta}\in\Theta}
\left|\frac{1}{h_{t,T}(\boldsymbol{\theta})}\frac{\partial h_{t,T}(\boldsymbol{\theta})}{\partial \boldsymbol{\theta}}\right|^p
\mathbb{E}\sup_{\boldsymbol{\theta}\in\Theta}\left(1 + \frac{X^2_{t,T}}{h_{t,T}(\boldsymbol{\theta})}\right)^q. \label{eq:holderk}
\end{align}
Set $q = 1 + \delta/2$ in (\ref{eq:holderk}). Then by Lemma 3, (A1) and independence of $\varepsilon_t$ and $h_{t,T}(\boldsymbol{\theta})$,
\begin{equation}
\mathbb{E}\sup_{\boldsymbol{\theta}\in\Theta}\left(\frac{X^2_{t,T}}{h_{t,T}(\boldsymbol{\theta})}\right)^{1+\delta/2}
=
\mathbb{E}\left(\varepsilon_t^2 \sup_{\boldsymbol{\theta}\in\Theta}\frac{h_{t,T}(\boldsymbol{\theta}_0)}{h_{t,T}(\boldsymbol{\theta})}\right)^{1+\delta/2}
=
\mathbb{E}(\varepsilon_t^{2+\delta}) \mathbb{E}\left(\sup_{\boldsymbol{\theta}\in\Theta}\frac{h_{t,T}(\boldsymbol{\theta}_0)}{h_{t,T}(\boldsymbol{\theta})}\right)^{1+\delta/2}
< \infty. \label{eq:holderk2} 
\end{equation}
From Lemma 4 we have
\begin{equation*}
\mathbb{E}\left|\sup_{\boldsymbol{\theta}\in\Theta}\frac{1}{h_{t,T}(\boldsymbol{\theta})}\frac{\partial h_{t,T}(\boldsymbol{\theta})}{\partial \boldsymbol{\theta}}\right|^{v} < \infty
\end{equation*}
for any $v$, which, together with (\ref{eq:holderk2}), verifies (\ref{eq:SE}). This completes the proof of Lemma 8.
\end{proof}

Applying Theorem A.1 to $L_T(\boldsymbol{\theta})$ yields
\begin{equation}
	L_T(\boldsymbol{\theta})\overset{P}{\to} \int_0^1 \mathbb{E}(L(u,\boldsymbol{\theta}))\ \text{d}u=L(\boldsymbol{\theta})
\end{equation}
pointwise. Stochastic equicontinuity (Lemma 8) and compact parameter space (A2) are sufficient conditions for the convergence to hold uniformly. 

It remains to be shown that $L(\boldsymbol{\theta})$ has a unique solution at the maximum $\boldsymbol{\theta}_0$. This is done as follows. We note that the log-likelihood function has the same form as in the stationary case, which enables us to use the arguments in BHK and Francq and Zakoïan (2004) to show that a global maximum exists. Since $$L(u,\boldsymbol{\theta}_0) - L(u,\boldsymbol{\theta}) = -\frac{1}{2}+\frac{1}{2}\mathbb{E}\left(\frac{\widetilde{h}_t(\boldsymbol{\theta}_0)}{\widetilde{h}_t(\boldsymbol{\theta})}-\ln \frac{\widetilde{h}_t(\boldsymbol{\theta}_0)}{\widetilde{h}_t(\boldsymbol{\theta})}\right), $$ 
we can use the fact that $x-\ln(x)>0$ for $x>0$ and reaches its minimum when $x=1$ to deduce that for each $u$, $L(u,\boldsymbol{\theta})$ has a global maximum at $\boldsymbol{\theta}_0.$ BHK, Theorem 2.3 and Lemma 5.5, shows that for each $u,$ if $\widetilde{h}_t(\boldsymbol{\theta}_a) = \widetilde{h}_t(\boldsymbol{\theta}_b),$ then $\boldsymbol{\theta}_a = \boldsymbol{\theta}_b.$ Therefore, the log likelihood function is uniquely maximized at $\boldsymbol{\theta}_0$. This completes the verification of (C3) and with it the proof of Theorem 1. $\qedsymbol{}$

\begin{rmk}
In this article we focus on weak consistency, i.e. convergence in probability of the estimator $\widehat{\boldsymbol{\theta}}_T$. Theorem A.1 yields $L_1$ convergence, which implies convergence in probability. The condition 
(\ref{eq:SE}) is a condition for (weak) SE. BHK proved strong consistency of the QMLE in strictly stationary GARCH models. It requires the condition $(1/T)\sum_{t=1}^{T}K_{t,T} =  O(1)$, almost surely, which in \textcite{andrews1992generic} is a condition for strong SE.
\end{rmk}

\subsection*{A.5 Proof of Theorem 2}
The proof of asymptotic normality of the estimator $\widehat{\boldsymbol{\theta}}_{T}$ is based on a standard Taylor series expansion argument. Asymptotic normality of the truncated estimator $\bar{\boldsymbol{\theta}}$ will be considered in Appendix A.6.  

Expanding the score vector around $\boldsymbol{\theta} _{0}$ (ignoring the factor $-1/2$ in what follows), we have
\begin{align}
	\bold{0} &= \frac{1}{\sqrt{T}} \sum_{t=1}^{T} \frac{\partial l_{t,T}(\widehat{\boldsymbol{\theta}}_{T})}{\partial \boldsymbol{\theta}} \nonumber \\ 
	&= \frac{1}{\sqrt{T}} \sum_{t=1}^{T} \frac{\partial l_{t,T}\left(\boldsymbol{\theta}_{0}\right)}{\partial \boldsymbol{\theta}}
	+ \left(
	\frac{1}{T} \sum_{t=1}^{T} \frac{\partial^{2} l_{t,T}\left(\boldsymbol{\theta}^{\ast}\right)}{\partial \boldsymbol{\theta} \partial \boldsymbol{\theta}^{\intercal}} \right)
	\sqrt{T} (\widehat{\boldsymbol{\theta}}_T - \boldsymbol{\theta}_0), \label{eq:TaylorScore}
\end{align}
where $\boldsymbol{\theta}^{\ast}$ lies on the line segment connecting $\widehat{\boldsymbol{\theta}}_T$ and $\boldsymbol{\theta}_0$. 
To avoid direct verification of uniform convergence of the Hessian, we use a Taylor series expansion of the second derivatives of the log-likelihood function about $\boldsymbol{\theta}_{0}$. For all $i, j$, we have
\[
\frac{1}{T} \frac{\partial ^{2}l_{t,T}(\boldsymbol{\theta} _{ij}^{\ast })}{\partial \theta _{i}\partial
\theta _{j}}= \frac{1}{T}\frac{\partial ^{2}l_{t,T}(\boldsymbol{\theta} _{0})}{\partial \theta
_{i}\partial \theta _{j}}
+ \frac{1}{T} \frac{\partial}{\partial \boldsymbol{\theta}^T }
\left\{
\frac{\partial^2 l_{t,T}(\boldsymbol{\theta}_{ij}^{\ast \ast})}{\partial \theta_i \partial \theta_j}
\right\}
\left(
\boldsymbol{\theta}_{ij}^{\ast} 
- \boldsymbol{\theta}_0
\right),
\]
where $\boldsymbol{\theta} _{ij}^{\ast \ast }$ lies on the line segment connecting $\boldsymbol{\theta} _{ij}^{\ast }$ and $\boldsymbol{\theta} _{0}$. Since $\boldsymbol{\theta} _{ij}^{\ast }$ converges in probability to $\boldsymbol{\theta} _{0}$, convergence in probability of the Hessian evaluated at $\boldsymbol{\theta}_{0}$ to $\mathbf{B}$ and bounded third derivatives of the log-likelihood function in a neighbourhood of $\boldsymbol{\theta}_{0}$ are sufficient conditions for the first factor of the second term on the right-hand side of  (\ref{eq:TaylorScore}) to converge in probability to $\mathbf{B}$. 
To prove Theorem 2, we show that the following three conditions hold. First,

(N1)$$ \frac{1}{\sqrt{T}} \sum_{t=1}^{T} \frac{\partial l_{t,T}\left(\boldsymbol{\theta}_{0}\right)}{\partial \boldsymbol{\theta}} \overset{D}{\to} N\left(\mathbf{0},\mathbf{A}
\right)$$ as $T \to\infty$, for a nonrandom $\mathbf{A}$. Second,

(N2) $$-\frac{1}{T} \sum_{t=1}^{T} \frac{\partial^{2} l_{t,T}\left(\boldsymbol{\theta}_{0}\right)}{\partial \boldsymbol{\theta} \partial \boldsymbol{\theta}^{\intercal}} \overset{P}{\to} 
\mathbf{B}$$
as $T \to\infty$, for a nonrandom positive-definite
matrix $\mathbf{B}$, and finally,

(N3)$$\mathbb{E}\left(\sup _{\left\|\boldsymbol{\theta}-\boldsymbol{\theta}_{0}\right\| \leq \delta}\left|\frac{\partial^{3} l_{t,T}(\boldsymbol{\theta})}{\partial \theta_{i} \partial \theta_{j} \partial \theta_{k}}\right|\right) < \infty $$ for all $i, j, k$ and all $\delta>0$. These conditions were employed by \textcite{ComteLieberman} for strictly stationary GARCH models. See also \textcite{BasawaFeiginHeyde}.

We begin by showing that (N1) holds. The score for observation $t$ evaluated at $\boldsymbol{\theta}_0$ is 
\begin{equation}
    \mathbf{s}_{t,T}(\boldsymbol{\theta}_0) = (1-\varepsilon^2_t)\left(\frac{1}{h_{t,T}(\boldsymbol{\theta}_0)}\frac{\partial h_{t,T}(\boldsymbol{\theta}_0)}{\partial \boldsymbol{\theta}}\right). \label{eq:scoretheta0}
\end{equation}
Our proof strategy will be to show that the assumptions of Theorem A.2 hold for an arbitrary linear combination of the elements of the score. This will allow us to appeal to the Cramér-Wold device, which will yield the desired multivariate convergence in distribution. To this end, let $\boldsymbol{\lambda} \in \mathbb{R}^{\text{dim}(\boldsymbol{\theta})}\setminus \{0\}.$ We need to verify Assumptions (S1), (S2) and (M1) with $n=2$ for  $\boldsymbol{\lambda}^\intercal\mathbf{s}_{t,T}(\boldsymbol{\theta}_0)$ and its stationary approximation 
$\boldsymbol{\lambda}^\intercal\widetilde{\mathbf{s}}_t(u,\boldsymbol{\theta}_0)$. We begin with (S2) since it implies part of (S1). By Corollary 3.2 of \textcite{rao2006}, $\widetilde{X}_{t}(u)$ is almost surely Lipschitz continuous in $u$, which implies that it is almost surely continuous in $u$ for all $t \in \mathbb{Z}$. Because the derivatives of the log-likelihood function are continuous functions, they are almost surely continuous for all $t \in \mathbb{Z}$. In (\ref{eq:scoretheta0}), $\varepsilon_t$ is independent and  $\mathbb{E}(1-\varepsilon^2_t)^2 < \infty$ by (A6). By Lemma 5.2 of BHK, \begin{equation}
		\mathbb{E}\left|\sup_{u\in[0,1]}\frac{1}{\widetilde{h}_t(u,\boldsymbol{\theta}_0)}\frac{\partial \widetilde{h}_t(u,\boldsymbol{\theta}_0)}{\partial\boldsymbol{\theta}}\right|^{2}<\infty.
\end{equation}
It follows that 
\begin{equation}
\norm{\sup_{u\in[0,1]}\widetilde{\bold{s}}_t(u,\boldsymbol{\theta}_0)}_2 = (\mathbb{E}(1-\varepsilon^2_t)^2)^{1/2}\norm{\sup_{u\in[0,1]}\frac{1}{\widetilde{h}_t(u,\boldsymbol{\theta}_0)}\frac{\partial \widetilde{h}_t(u,\boldsymbol{\theta}_0)}{\partial \boldsymbol{\theta}}}_2<\infty.    
\end{equation}
This completes the verfication of (S2). 

To verify (S1) for an arbitrary linear combination of the elements of the score, we have to show that
\begin{equation}	
\sup_{u \in [0,1]} \norm{\boldsymbol{\lambda}^{\intercal}\widetilde{\mathbf{s}}_t(u,\boldsymbol{\theta}_0)}_2 
<
\infty \label{eq:supscore}
\end{equation}
and
\begin{equation}	
\norm{\boldsymbol{\lambda}^{\intercal}\mathbf{s}_{t,T}(\boldsymbol{\theta}_0) -\boldsymbol{\lambda}^{\intercal}\widetilde{\mathbf{s}}_t(u,\boldsymbol{\theta}_0)}_2 
\leq
C\left(\frac{1}{T} + |t/T-u| \right). \label{eq:scoreS1}
\end{equation}
The condition (\ref{eq:supscore}) is implied by (S2), since \begin{equation*}
    \norm{\sup_{u \in [0,1]} \boldsymbol{\lambda}^{\intercal}\widetilde{\mathbf{s}}_t(u,\boldsymbol{\theta})}_{2} \geq \sup_{u \in [0,1]} \norm{\boldsymbol{\lambda}^{\intercal}\widetilde{\mathbf{s}}_t(u,\boldsymbol{\theta})}_2.
\end{equation*}
To prove the second part of (S1), it suffices to consider differences of the type
\begin{equation*}	
\left|\frac{1}{h_t}\boldsymbol{\lambda}^{\intercal}\frac{\partial h_t}{\partial \boldsymbol{\theta}} - \frac{1}{h_t^{\prime}}\boldsymbol{\lambda}^{\intercal}\frac{\partial h_t^{\prime}}{\partial \boldsymbol{\theta}}\right|.
\end{equation*}
Adding and subtracting $(1/h_t)\boldsymbol{\lambda}^{\intercal}\partial h_t^{\prime}/\partial \boldsymbol{\theta}$, rearranging terms and using the triangle inequality, we can write
\begin{align}
	\left|\frac{1}{h_t}\boldsymbol{\lambda}^{\intercal}\frac{\partial h_t}{\partial \boldsymbol{\theta}} - \frac{1}{h_t^{\prime}}\boldsymbol{\lambda}^{\intercal}\frac{\partial h_t^{\prime}}{\partial \boldsymbol{\theta}}\right| 
 &\leq  \left|\frac{1}{h_t}\right|\left|\boldsymbol{\lambda}^{\intercal}\frac{\partial h_t}{\partial \boldsymbol{\theta}} - \boldsymbol{\lambda}^{\intercal}\frac{\partial h_t^{\prime}}{\partial \boldsymbol{\theta}}\right|
	+ \left|\frac{1}{h_t^{\prime}}\boldsymbol{\lambda}^{\intercal}\frac{\partial h_t^{\prime}}{\partial \boldsymbol{\theta}}\right|\left|\frac{h_t - h_t^{\prime}}{h_t}\right|. \label{eq:score_decomposition}
\end{align}
Consider the first term on the right-hand side of (\ref{eq:score_decomposition}). Set $h_t = h_{t,T}(\boldsymbol{\theta}_0)$, $\partial h_{t}/\partial \boldsymbol{\theta} = \partial h_{t,T}(\boldsymbol{\theta}_0)/\partial \boldsymbol{\theta}$, $\partial h_t^{\prime}/\partial \boldsymbol{\theta} = \partial \widetilde{h}_t (u,\boldsymbol{\theta}_0)/\partial \boldsymbol{\theta}$. Since $h_{t,T}(\boldsymbol{\theta}_0)>0$ and bounded from below, we have $1/h_{t,T}(\boldsymbol{\theta}_0) < \infty$ and, consequently,
\begin{align*} 
	\norm{\frac{1}{h_{t,T}(\boldsymbol{\theta}_0)} \left( 
 \boldsymbol{\lambda}^{\intercal}\frac{\partial h_{t,T}(\boldsymbol{\theta}_0)}{\partial \boldsymbol{\theta}} - \boldsymbol{\lambda}^{\intercal}\frac{\partial \widetilde{h}_t (u,\boldsymbol{\theta}_0)}{\partial \boldsymbol{\theta}}\right)}_2
 \leq
\frac{1}{C_1}
\norm{\boldsymbol{\lambda}^{\intercal}\left(\frac{\partial h_{t,T}(\boldsymbol{\theta}_0)}{\partial \boldsymbol{\theta}} - \frac{\partial \widetilde{h}_t (u,\boldsymbol{\theta}_0)}{\partial \boldsymbol{\theta}}\right)}_2.
\end{align*}
By (\ref{eq:lemma31}) and the Minkowski inequality, we obtain\begin{align}
&\norm{\boldsymbol{\lambda}^{\intercal}\left(\frac{\partial h_{t,T}(\boldsymbol{\theta})_0}{\partial \boldsymbol{\theta}}- \frac{\partial \widetilde{h}_t (u,\boldsymbol{\theta}_0)}{\partial \boldsymbol{\theta}}\right)}_2 \leq \left|\boldsymbol{\lambda}^{\intercal}\left(\frac{\partial c_0(\boldsymbol{\theta}_0)}{\partial \boldsymbol{\theta}} - \frac{\partial c_0(u,\boldsymbol{\theta}_0)}{\partial \boldsymbol{\theta}}\right)\right| \nonumber \\
&+ \left|\boldsymbol{\lambda}^{\intercal}\left(\sum_{i=1}^{\infty}\frac{\partial d_i(\boldsymbol{\theta}_0)}{\partial \boldsymbol{\theta}}g_{t-i+1,T}(\boldsymbol{\theta}_0) - \sum_{i=1}^{\infty}\frac{\partial d_i(\boldsymbol{\theta}_0)}{\partial \boldsymbol{\theta}}g_{t-i+1,T}(u,\boldsymbol{\theta}_0)\right)\right| \nonumber \\
&+ \left|\boldsymbol{\lambda}^{\intercal}\left(\sum_{i=1}^{\infty}d_i(\boldsymbol{\theta}_0)\frac{\partial g_{t-i+1,T}(\boldsymbol{\theta_0})}{\partial \boldsymbol{\theta}} - \sum_{i=1}^{\infty}d_i(\boldsymbol{\theta}_0)\frac{\partial g_{t-i+1,T}(u,\boldsymbol{\theta_0}) }{\partial \boldsymbol{\theta}})\right)\right| \nonumber\\
&+ \norm{\boldsymbol{\lambda}^{\intercal}\left(\sum_{i=1}^{\infty}\frac{\partial{c}_i(\boldsymbol{\theta}_0)}{\partial \boldsymbol{\theta}}X^2_{t-i,T} - \sum_{i=1}^{\infty}\frac{\partial{c}_i(\boldsymbol{\theta}_0)}{\partial \boldsymbol{\theta}}\widetilde{X}^2_{t-i}(u)\right)}_2. \label{eq:fullderivs}
\end{align}
The first term in (\ref{eq:fullderivs}) is zero. 
By using (\ref{eq:cineq}), (\ref{eq:dineq}), (\ref{eq:lemma33}) and monotone convergence, we have 
\begin{align}	
  &\left|\sum_{i=1}^{\infty}\frac{\partial d_i(\boldsymbol{\theta}_0)}{\partial \boldsymbol{\theta}}\left(g_{t-i+1,T}(\boldsymbol{\theta}_0) - g_{t-i+1,T}(u,\boldsymbol{\theta}_0)\right)\right| \leq
\left|C_1\sum_{i=1}^{\infty}i\rho_0^{i/q}\left(g_{t-i+1,T}(\boldsymbol{\theta}_0) - g_{t-i+1}(u,\boldsymbol{\theta}_0)\right)\right|, \nonumber \\ 
&\left|\sum_{i=1}^{\infty}d_i(\boldsymbol{\theta}_0)\left(\frac{\partial g_{t-i+1,T}(\boldsymbol{\theta}_0) }{\partial \boldsymbol{\theta}} - \frac{\partial g_{t-i+1,T}(u,\boldsymbol{\theta}_0) }{\partial \boldsymbol{\theta}} \right)\right| \leq
\left|C_2\sum_{i=1}^{\infty}\rho_0^{i/q}\left(\frac{\partial g_{t-i+1,T}(\boldsymbol{\theta}_0)}{\partial \boldsymbol{\theta}} - \frac{\partial g_{t-i+1}(u,\boldsymbol{\theta}_0)}{\partial \boldsymbol{\theta}}\right)\right|\nonumber, \\
&\text{and} \nonumber \\
&\norm{\sum_{i=1}^{\infty}\frac{\partial{c}_i(\boldsymbol{\theta}_0)}{\partial \boldsymbol{\theta}}X^2_{t-i,T} - \sum_{i=1}^{\infty}\frac{\partial{c}_i(\boldsymbol{\theta}_0)}{\partial \boldsymbol{\theta}}\widetilde{X}^2_{t-i,T}(u)}_2 \leq
C_3\sum_{i=1}^{\infty}i\rho_0^{i/q}\norm{X^2_{t-i,T} - \widetilde{X}_{t-i}^{2}(u)}_2. \nonumber
\end{align}
By Lemma 1, $g_{t,T}(\boldsymbol{\theta}_1)$ and its derivatives are Lipschitz continuous. As already mentioned, $\sum_{i=1}^{\infty}i\rho_0^{i/q}$ converges. We therefore obtain for the second and third terms in (\ref{eq:fullderivs})
\begin{align}
\left\vert C_1\sum_{i=1}^{\infty }i\rho _{0}^{i/q}( g_{t-i+1,T}(\boldsymbol{\theta }_{0})-g_{t-i+1}(u,\boldsymbol{\theta }_{0})) \right\vert
 \leq C_4\left( \frac{1}{T}+\left\vert \frac{t}{T}-u\right\vert \right), \label{eq:gderivdiff1}
\end{align}
and
\begin{align}
  \left\vert C_2\sum_{i=1}^{\infty }\rho _{0}^{i/q} \left( \frac{\partial
g_{t-i+1,T}(\boldsymbol{\theta }_{0})}{\partial \boldsymbol{\theta }}-\frac{\partial g_{t-i+1}(u,\boldsymbol{\theta }_{0})}{\partial \boldsymbol{\theta }} \right) \right\vert \leq C_5\left( \frac{1}{T}+\left\vert
\frac{t}{T}-u\right \vert \right), \label{eq:gderivdiff2}
\end{align}
respectively.
Finally, applying Lemma 2, formula (\ref{eq:uvineq}), to the fourth term gives
\begin{equation}
C_3\sum_{i=1}^{\infty}i\rho_0^{i/q}\norm{X^2_{t-i} - \widetilde{X}^{2}_{t-i}(u)
}_2
\leq
C_6 \left( \frac{1}{T}+\left\vert \frac{t}{T}-u\right\vert \right). \label{eq:garchderivdiff}
\end{equation} 
In view of (\ref{eq:gderivdiff1}), (\ref{eq:gderivdiff2}) and (\ref{eq:garchderivdiff}), we have
\begin{equation}
		\norm{\boldsymbol{\lambda}^\intercal\frac{\partial h_{t,T}(\boldsymbol{\theta}_0)}{\partial \boldsymbol{\theta}} - \boldsymbol{\lambda}^\intercal\frac{\partial \widetilde{h}_t(u,\boldsymbol{\theta}_0)}{\partial \theta}}_2\leq C\left(\frac{1}{T} + \left| \frac{t}{T}-u\right|\right). \label{eq:derivS1}
\end{equation}
Next we turn to the second term on the right-hand side of (\ref{eq:score_decomposition}). Setting $h_t^{\prime} = \widetilde{h}_t (u,\boldsymbol{\theta}_0)$ yields
\begin{equation}
\norm{\frac{1}{\widetilde{h}_t (u,\boldsymbol{\theta}_0)}\boldsymbol{\lambda}^\intercal\frac{\partial \widetilde{h}_t (u,\boldsymbol{\theta}_0)}{\partial \boldsymbol{\theta}}\left(\frac{h_{t,T}(\boldsymbol{\theta}_0)-\widetilde{h}_t (u,\boldsymbol{\theta}_0)}{h_{t,T}(\boldsymbol{\theta}_0)}\right)}_2.
\end{equation}
Let $\widetilde{\nabla}_t(u,\boldsymbol{\theta}) := (1/\widetilde{h}_t (u,\boldsymbol{\theta}))(\boldsymbol{\lambda}^\intercal\partial \widetilde{h}_t (u,\boldsymbol{\theta})/\partial \boldsymbol{\theta})$.
Then, by Lemma 4, formula ({\ref{eq:boundedh1deriv}}),
$$\sup_{\boldsymbol{\theta}\in\Theta} \norm{\widetilde{\nabla}_t(u,\boldsymbol{\theta})}_n < \infty$$ for all $n$.
Since the arguments in (\ref{eq:score_decomposition}) remain valid if we instead of $(1/h_t)\boldsymbol{\lambda}^\intercal\partial h_t^{\prime}/\partial \boldsymbol{\theta}$ add and subtract $(1/h_t^{\prime})\boldsymbol{\lambda}^\intercal\partial h_t/\partial \boldsymbol{\theta}$, we may without loss of generality assume that  $h_{t,T}(\boldsymbol{\theta}_0)\geq \widetilde{h}_t(u, \boldsymbol{\theta}_0)$. Then, since
$(h_{t,T}(\boldsymbol{\theta}_0)-\widetilde{h}_t(u, \boldsymbol{\theta}_0))/h_{t,T}(\boldsymbol{\theta}_0) \in [0,1)$, we have that for all $s \in (0,1)$,
$$\norm{\widetilde{\nabla}_t(u,\boldsymbol{\theta}_0)\left(\ \frac{h_{t,T}(\boldsymbol{\theta}_0)-\widetilde{h}_t (u,\boldsymbol{\theta}_0)}{h_{t,T}(\boldsymbol{\theta}_0)}\right)}_2 \leq 
\norm{\widetilde{\nabla}_t(u,\boldsymbol{\theta}_0)\left(\ \frac{h_{t,T}(\boldsymbol{\theta}_0)-\widetilde{h}_t (u,\boldsymbol{\theta}_0)}{h_{t,T}(\boldsymbol{\theta}_0)}\right)^s}_2.
$$ 
Now choose $s=1/2$. Then the Cauchy-Schwarz inequality and Lemma 2, formula (\ref{eq:htapprox}), give
\begin{align}
\norm{\widetilde{\nabla}_t(u,\boldsymbol{\theta}_0)\left(\ \frac{h_{t,T}(\boldsymbol{\theta}_0)-\widetilde{h}_t (u,\boldsymbol{\theta}_0)}{h_{t,T}(\boldsymbol{\theta}_0)}\right)^{1/2}}_2
&\leq \norm{\widetilde{\nabla}_t(u,\boldsymbol{\theta}_0)}_4\norm{\left(\ \frac{h_{t,T}(\boldsymbol{\theta}_0)-\widetilde{h}_t (u,\boldsymbol{\theta}_0)}{h_{t,T}(\boldsymbol{\theta}_0)}\right)^{1/2}}_4 \nonumber \\ 
&\leq
C_1\norm{h_t(\boldsymbol{\theta}_0) -\widetilde{h}_t(u,\boldsymbol{\theta}_0)}_2 \nonumber \\ &\leq C_2\left(\frac{1}{T} + \left|\frac{t}{T}-u\right|\right).\label{eq:csscore}    
\end{align}
Put together, (\ref{eq:derivS1}) and (\ref{eq:csscore}) imply (\ref{eq:scoreS1}), which completes the verification of (S1).   

Finally, we need Assumption (M1) to be fulfilled for an arbitrary linear combination of the elements of the score with $n=2$. The assumption entails the following mixing condition on the stationary approximation: 
\begin{equation}
\sum_{t=0}^{\infty}\delta_{2}^{\boldsymbol{\lambda}^\intercal \widetilde{\bold{s}}(\boldsymbol{\theta}_0)}(t)= \sum_{t=0}^{\infty}\sup_{u\in[0,1]}\norm{\boldsymbol{\lambda}^\intercal\widetilde{\bold{s}}_{t}(u,\boldsymbol{\theta}_0) - \boldsymbol{\lambda}^\intercal \widetilde{\bold{s}}_t^{e}(u,\boldsymbol{\theta}_0)}_2 <\infty, \label{eq:scoreM1}
\end{equation} 
where $\widetilde{\bold{s}}_t^{e}(\boldsymbol{\theta})$ is a coupled version of the stationary approximation $\widetilde{\bold{s}}_t(u,\boldsymbol{\theta})$ of the score with the error term replaced at index $t=0$ in the information set. The condition ensures that the process "forgets the history geometrically quickly"; see \textcite{wu2011asymptotic}.
The first term in the sum (\ref{eq:scoreM1}) is a constant, and for $t>0$ we obtain
\begin{align}
\delta_{2}^{\boldsymbol{\lambda}^\intercal \widetilde{\bold{s}}(u, \boldsymbol{\theta}_0)}(t)
 &\leq C_1\sup_{u\in[0,1]}\norm{\boldsymbol{\lambda}^\intercal \frac{\partial \widetilde{h}_{t}(u,\boldsymbol{\theta}_0)}{\partial \boldsymbol{\theta}} - \boldsymbol{\lambda}^\intercal\frac{\partial \widetilde{h}_{t}^{e}(u,\boldsymbol{\theta}_0)}{\partial \boldsymbol{\theta}}}_2  \nonumber \\
&+ C_2\sup_{u\in[0,1]}\norm{ \widetilde{h}_{t}(u,\boldsymbol{\theta}_0) - \widetilde{h}_{t}^{e}(u,\boldsymbol{\theta}_0)}_2 \nonumber \\
	&\leq C_1\sup_{u\in[0,1]}\sum_{i=1}^{\infty}i\rho_0^{i/q}\mathbb{E}\norm{\widetilde{X}^2_{t-i}(u) - \widetilde{X}^{e2}_{t-i}(u)}_2 \nonumber \\
	&+C_2\sup_{u\in[0,1]}\sum_{i=1}^{\infty}\rho_0^{i/q}\norm{\widetilde{X}^2_{t-i}(u) - \widetilde{X}^{e2}_{t-i}(u)}_2, \label{eq:initM1}
	\end{align}
where we have used the decomposition (\ref{eq:score_decomposition}) and then reused arguments from the proof of (S1). The inequalities are true because the derivatives with respect to the components in the time-varying intercept are deterministic, and thus the components of the score that correspond to them do not have memory, so swapping the error term at some index only affects them at $t=0$. The two terms in the last inequality of (\ref{eq:initM1}) are treated similarly. It suffices to consider the first one and the conclusion for the second one will follow using the same techniques.  
Since for $x,y\geq0$ we have $(x^2-y^2)=(x+y)(x-y)\leq C\left|x+y\right|\left|x-y\right|$ for some $C\geq1$, the function $f(x) = x^2$ satisfies the invariance property of DRW, Proposition 2.5, with $M = 1$. It states that if an assumption, among them (M1), is fulfilled for a process with $\widetilde{n} = n(M+1)$, then the assumption is fulfilled for a function of the process that satisfies the invariance property. In our case, since $M=1$ and we need $n=2$, we choose $\widetilde{n}=4$, which translates into us needing
\begin{equation}
C\sup_{u\in[0,1]}\sum_{i=1}^{\infty}i\rho_0^{i/q}\mathbb{E}\left(\widetilde{X}_{t-i}(u) - \widetilde{X}^e_{t-i}(u)\right)^4<\infty. \label{eq:M1}
\end{equation}
Since the expectation in (\ref{eq:M1}) involves a stationary GARCH process, we can invoke \textcite[Proposition 3]{wu2005linear}. For any $u$,

\begin{equation}
    	C\sum_{i=1}^{\infty}i\rho_0^{i/q}\mathbb{E}\left(\widetilde{X}_{t-i}(u) - \widetilde{X}^{e}_{t-i}(u)\right)^4\leq     	C\sum_{i=1}^{\infty}i\rho_0^{i/q}r^{t-i},
\end{equation}
where $r\in(0,1)$ for $t-i\geq0,$ $r=0$ otherwise. Now choose $a\in(0,1)$ such that $a>r,$ $a>\rho_0$ for $t-i\geq0,$ $a=0$ otherwise. Then, $$\delta_{2}(i,t):=C\sum_{i=1}^{\infty}ia^{i/q + t-i} \geq C\sum_{i=1}^{\infty}i\rho_0^{i/q}r^{t-i}.$$
To illustrate, we have
\begin{equation*}
 \delta_{2}(i,t) =\\\left\{
  \begin{array}{@{}ll@{}}
       C\left(a^{1/q+t-1}+\ldots+ (t-1)a^{(t-1)/q+t-(t-1)}\right), & \text{if}\ t\geq i \\
      0, & \text{otherwise.}
     \end{array}\right.
\end{equation*}
Since the terms in the sum $\delta_{2}(i,t)$ are increasing in $i$ until $i=t$, their sum is less than the term at index $i=t$ multiplied by the number of terms $t$, which yields the inequality \begin{equation*}
    \sum^{\infty}_{t=0}\delta_{2}(i,t)<\sum^{\infty}_{t=0}t^2a^{t/q} <\infty,
\end{equation*}
since $a<1$. Thus, $\sum_{t=0}^\infty\delta_{2}^{\boldsymbol{\lambda}^{\intercal}\widetilde{\bold{s}}(\boldsymbol{\theta}_0)}(t) < \infty$, and (M1) is satisfied.

Having verified (S1), (S2) and (M1), we can apply Theorem A.2 to the linear combination of the elements of the score. The theorem yields \begin{equation}
	\frac{1}{\sqrt{T}} \sum_{t=1}^{T} \boldsymbol{\lambda}^{\intercal}\frac{\partial l_{t,T}\left(\boldsymbol{\theta}_{0}\right)}{\partial \boldsymbol{\theta}}\overset{D}{\to}\left\{\int_{0}^{1}\sigma_\lambda(v,\boldsymbol{\theta}_{0})dW(v)\right\}, \label{eq:lambdascoreCLT}
\end{equation}
where $$\sigma_\lambda^2(v,\boldsymbol{\theta}_{0}) = \text{Var}\left(\boldsymbol{\lambda}^{\intercal}\frac{\partial \widetilde{l}_0\left(v,\boldsymbol{\theta}_{0}\right)}{\partial \boldsymbol{\theta}}\right) = \boldsymbol{\lambda}^\intercal\boldsymbol{A}(v)\boldsymbol{\lambda}$$ and $W$ is a Brownian motion. By the It\^{o} Isometry we can conclude that $$\mathbb{E}\left\{\int_{0}^{1}\sigma_\lambda(v,\boldsymbol{\theta}_{0})dW(v)\right\}=0$$ and $$\text{Var}\left\{\int_{0}^{1}\sigma_\lambda(v,\boldsymbol{\theta}_{0})dW(v)\right\}=\int_{0}^{1}\mathbb{E}\boldsymbol{\lambda}^{\intercal}\boldsymbol{A}(v)\boldsymbol{\lambda}\space \text{d}v:=\boldsymbol{\lambda}^{\intercal}\boldsymbol{A}\boldsymbol{\lambda}.$$
It is well known that integrals of the form (\ref{eq:lambdascoreCLT}) follow a normal distribution. Finally, since $\boldsymbol{\lambda}$ was arbitrary, the Cramér-Wold device yields \begin{equation}
	\frac{1}{\sqrt{T}} \sum_{t=1}^{T}\frac{\partial l_{t,T}\left(\boldsymbol{\theta}_{0}\right)}{\partial \boldsymbol{\theta}}\overset{D}{\to}N\left(\boldsymbol{0}, \boldsymbol{A}\right), \label{eq:scoreCLT}
\end{equation}
where $\bold{A}$ is nonrandom.

It remains to show that $\bold{A}$ is positive definite. It is clear that $\bold{A}$ is positive semi-definite. We now show that it is non-singular. By (\ref{eq:scoretheta0}), we have
\begin{align}
    \bold{A} =
    \left(\mathbb{E} \varepsilon_{0}^{4}-1 \right)
        \int_0^1\mathbb{E}\left(\frac{1}{\widetilde{h}_t(u,\boldsymbol{\theta}_{0})^2}\frac{\partial \widetilde{h}_t(u,\boldsymbol{\theta}_{0})}{\partial\boldsymbol{\theta}}\frac{\partial \widetilde{h}_t(u,\boldsymbol{\theta}_{0})}{\partial\boldsymbol{\theta}^{\intercal}}\right)\text{d}u. \label{eq:intA}
\end{align} 
Suppose
\begin{align*}
\boldsymbol{\lambda}^{\intercal} \bold{A}(u) \boldsymbol{\lambda}
= \boldsymbol{\lambda}^{\intercal} \mathbb{E}\left(\frac{1}{\widetilde{h}_t(u,\boldsymbol{\theta}_{0})^2}\frac{\partial \widetilde{h}_t(u,\boldsymbol{\theta}_{0})}{\partial\boldsymbol{\theta}}\frac{\partial \widetilde{h}_t(u,\boldsymbol{\theta}_{0})}{\partial\boldsymbol{\theta}^{\intercal}}\right) \boldsymbol{\lambda} \\
= \mathbb{E}\left(\frac{1}{\widetilde{h}_t(u,\boldsymbol{\theta}_{0})^2}\left(\boldsymbol{\lambda}^{\intercal}\frac{\partial \widetilde{h}_t(u,\boldsymbol{\theta}_{0})}{\partial\boldsymbol{\theta}}\right)^2\right) = 0 
\end{align*}
for some vector $\boldsymbol{\lambda}\in\mathbb{R}^{\text{dim}(\boldsymbol{\theta})}$. Then, almost surely $\boldsymbol{\lambda}^{T}(\partial \widetilde{h}_t(u,\boldsymbol{\theta}_0)/\partial \boldsymbol{\theta}) =0$. 
Consider first for notational simplicity a model with one logistic transition function with one $c$ parameter. Define the vector of partial derivatives of the conditional variance $\boldsymbol{\phi}_t = \partial h_{t,T}(\boldsymbol{\theta})/\partial \boldsymbol{\theta}$ and its stationary approximation $\widetilde{\boldsymbol{\phi}}_t(u) = \partial \widetilde{h}_t(u,\boldsymbol{\theta})/\partial \boldsymbol{\theta}.$ The stationary approximation takes the form \begin{equation}
\widetilde{\boldsymbol{\phi}}_t(u) = \widetilde{\boldsymbol{\psi}}_t + \beta_1\widetilde{\boldsymbol{\phi}}_{t-1}(u) + \ldots + \beta_q\widetilde{\boldsymbol{\phi}}_{t-1}(u),
\end{equation}
where $\widetilde{\boldsymbol{\psi}}_t = \left(1, g_\gamma(u), g_c(u), g_{\alpha 01}(u), \widetilde{X}^2_{t-1}(u), \widetilde{X}^2_{t-2}(u),\ldots,\widetilde{X}^2_{t-p}(u), \widetilde{h}_{t-1}(u), \widetilde{h}_{t-2}(u),  \ldots, \widetilde{h}_{t-q}(u)\right)^{\intercal}.$ Define $\boldsymbol{\alpha}_0(u) = (1, g_\gamma(u), g_c(u), g_{\alpha 01}(u))^\intercal$.
By stationarity and the arguments in BHK (2004, Lemma 5.7),  singularity implies that  $$\boldsymbol{\lambda}^{\intercal}\widetilde{\boldsymbol{\psi}}_t = \boldsymbol{0}$$ almost surely for some $\boldsymbol{\lambda} \neq \boldsymbol{0}$ and for all $t$. Consider the partition $\boldsymbol{\lambda} = (\boldsymbol{\lambda}_0^\intercal, \boldsymbol{\lambda}_1^\intercal, \boldsymbol{\lambda}_2^\intercal)^\intercal,$ with $\boldsymbol{\lambda}_0 \in \mathbb{R}^4, \boldsymbol{\lambda}_1 \in \mathbb{R}^p, \boldsymbol{\lambda}_2 \in \mathbb{R}^q.$ We begin by showing that $\boldsymbol{\lambda}_1$ and $\boldsymbol{\lambda}_2$ are zero. This is done as in \textcite{FrancqZakoian2004} and BHK (2004) by showing that it implies a reduced order of the GARCH process.
We have that  \begin{align}
\boldsymbol{\lambda^{\intercal}\widetilde{\psi}}_t &= \sum^4_{k=1} \lambda_{0,k} \alpha_{0,k} + \sum^p_{i=1}\lambda_{1,i}\widetilde{X}^2_{t-i} +  \sum^q_{j = 1}\lambda_{2,j}\widetilde{h}_{t-j}  = 0 \nonumber \\
\Leftrightarrow \lambda_{1,1}\widetilde{X}^2_{t-1} &= -\sum^p_{i=2}\lambda_{1,i}\widetilde{X}^2_{t-i} -\sum^q_{j = 1}\lambda_{2,j}\widetilde{h}_{t-j} - \sum^4_{k=1} \lambda_{0,k} \alpha_{0,k}
\label{eq:derivarg1}
\end{align}
Now we have that $\lambda_{1,1}=0$, otherwise $X^2_{t-1}$ is measurable with respect to the $\sigma$-field  $\mathcal{F}_{t-2}$ generated by $\{\varepsilon_{t-2}, \varepsilon_{t-3}, \ldots \}$ as everything on the RHS of (\ref{eq:derivarg1}) is contained in $\mathcal{F}_{t-2}.$ We can then further rearrange (\ref{eq:derivarg1}) to obtain \begin{align}
\lambda_{2,1}\widetilde{h}_{t-1} &= -\sum^p_{i=2}\lambda_{1,i}\widetilde{X}^2_{t-i} -\sum^q_{j = 2}\lambda_{2,j}\widetilde{h}_{t-j} - \sum^4_{k=1} \lambda_{0,k} \alpha_{0,k}. \label{eq:garchpqminus}
\end{align}
Thus, if $\lambda_{2,1}\neq 0$, $\widetilde{h}_{t-1}$ can be written in terms of a GARCH($p-1, q-1$) representation as the sums on the RHS of (\ref{eq:garchpqminus}) contain $p-1$ and $q-1$ terms, respectively. But then  $\lambda_{1,2} = 0$ as (\ref{eq:derivarg1}) can be rearranged so that the LHS is $\lambda_{1,2}\widetilde{X}^2_{t-2}$ and the RHS contains at most $\mathcal{F}_{t-3}$ information.
Continuing this argument iteratively shows that $\boldsymbol{\lambda}_1  \neq \boldsymbol{0} $ or $\boldsymbol{\lambda}_2 \neq \boldsymbol{0}$ contradicts the minimality of the model. By BHK, Theorem 2.5, non-degenerate errors (A1), $\mathbb{E}(\ln\sigma_0^2)<\infty$ and (A4)  imply minimality.  
A similar result can also be found in BHK, Lemma 5.7.

Next consider $\boldsymbol{\lambda}_0$. In this case with one transition function and one $c$ parameter, Lemma \ref{lemma:gpd} shows that $\boldsymbol{\lambda}_0 = \boldsymbol{0}.$ In the general case, it follows by (A3). This implies that $\bold{A}$ is positive definite and thus completes the proof of (N1).

Now we turn to verifying (N2). To be able to apply Theorem A.1, we have to show that the Hessian evaluated at $\boldsymbol{\theta}_0$ fulfils (S1) with $n=1$. The Hessian for observation $t$ evaluated at $\boldsymbol{\theta}_0$ is 
\begin{align}
\bold{H}_{t,T}(\boldsymbol{\theta}_0)&= \left(1-\varepsilon^2_t\right)\frac{1}{h_{t,T}(\boldsymbol{\theta}_0)}\frac{\partial^2 h_{t,T}(\boldsymbol{\theta}_0)}{\partial\boldsymbol{\theta}\partial\boldsymbol{\theta}^{\intercal}} \label{eq:H_0hessian} \nonumber \\
		&+\left(2\varepsilon_t^2-1\right)\frac{1}{h_{t,T}^{2}(\boldsymbol{\theta}_0)}\frac{\partial h_{t,T}(\boldsymbol{\theta}_0)}{\partial\boldsymbol{\theta}}\frac{\partial h_{t,T}(\boldsymbol{\theta}_0)}{\partial\boldsymbol{\theta}^{\intercal}}. 
	\end{align}
We need to show that
\begin{equation*}	
\sup_{u\in[0,1]} \norm{\widetilde{\mathbf{H}}_t(u,\boldsymbol{\theta}_0)}_1
<
\infty
\end{equation*}
and
\begin{equation*}	
\norm{\mathbf{H}_{t,T}(\boldsymbol{\theta}_0) - \widetilde{\mathbf{H}}_t(u,\boldsymbol{\theta}_0)}_1 
\leq
C\left(\frac{1}{T} + \left |\frac{t}{T}-u \right| \right),
\end{equation*}
where $\widetilde{\mathbf{H}}_t(u,\boldsymbol{\theta}_0)$
is the stationary approximation of the Hessian at $u$. By Lemma 5.6 of BHK, we have
\begin{equation*}
\mathbb{E}\left|\sup_{u\in[0,1]} \widetilde{\mathbf{H}}_t(u,\boldsymbol{\theta}_0)\right|<\infty,
\end{equation*}
so that
\begin{equation}	\sup_{u\in[0,1]}\mathbb{E}\left|\widetilde{\mathbf{H}}_t(u,\boldsymbol{\theta}_0)\right|\leq\mathbb{E}\left|\sup_{u\in[0,1]} \widetilde{\mathbf{H}}_t(u,\boldsymbol{\theta}_0)\right|<\infty. \label{eq:Hess2}
\end{equation}

The first term on the right-hand side of (\ref{eq:H_0hessian}) has expectation zero. Consider the second term. By (A1), $\mathbb{E}\left(2\varepsilon_0^2-1\right)=1$. Let $\boldsymbol{\nabla}_t = (1/h_{t})\partial h_{t}/\partial \boldsymbol{\theta}$ and $\boldsymbol{\nabla}_t^{\prime} = (1/h_{t}^{\prime})\partial h_{t}^{\prime}/\partial \boldsymbol{\theta}$ and consider differences of the type
$|\boldsymbol{\nabla}_t\boldsymbol{\nabla}_{t}^{\intercal}-\boldsymbol{\nabla}_t^{\prime}\boldsymbol{\nabla}_{t}^{\prime \intercal}|$.
By adding and subtracting $\boldsymbol{\nabla}_t\boldsymbol{\nabla}_{t}^{\prime \intercal}$, rearranging terms and using the triangle inequality, we can write
\begin{equation*}
|\boldsymbol{\nabla}_t\boldsymbol{\nabla}_{t}^{\intercal}-\boldsymbol{\nabla}_t^{\prime}\boldsymbol{\nabla}_{t}^{\prime \intercal}|
\leq
|\boldsymbol{\nabla}_t(\boldsymbol{\nabla}_{t}^{\intercal}-\boldsymbol{\nabla}_{t}^{\prime \intercal})| 
+
|(\boldsymbol{\nabla}_t -\boldsymbol{\nabla}_{t}^{\prime})\boldsymbol{\nabla}_t^{\prime \intercal}|.
\end{equation*}
Setting $\boldsymbol{\nabla}_{t} = \boldsymbol{\nabla}_{t,T}(\boldsymbol{\theta}_0) = (1/h_{t,T}(\boldsymbol{\theta}_0)) \partial h_{t,T}(\boldsymbol{\theta}_0)/\partial \boldsymbol{\theta}$ and $\boldsymbol{\nabla}_t^{\prime} = 
\widetilde{\boldsymbol{\nabla}}_{t}(u,\boldsymbol{\theta}_0) =
(1/\widetilde{h}_{t}(u,\boldsymbol{\theta}_0)) \partial \widetilde{h}_{t}(u,\boldsymbol{\theta}_0)/\partial \boldsymbol{\theta}$
yields
\begin{align*}
	&\norm{\boldsymbol{\nabla}_{t,T}(\boldsymbol{\theta}_0)\boldsymbol{\nabla}_{t,T}^{\intercal}(\boldsymbol{\theta}_0)-\widetilde{\boldsymbol{\nabla}}_{t}(u,\boldsymbol{\theta}_0)\widetilde{\boldsymbol{\nabla}}_{t}^{\intercal}(u,\boldsymbol{\theta}_0)}_1 \\
 &\leq
  \norm{\boldsymbol{\nabla}_{t,T}(\boldsymbol{\theta}_0)(\boldsymbol{\nabla}_{t,T}^{\intercal}(u,\boldsymbol{\theta}_0)-\widetilde{\boldsymbol{\nabla}}_{t}^{\intercal}(u,\boldsymbol{\theta}_0))}_1
 +
  \norm{(\boldsymbol{\nabla}_{t,T}(\boldsymbol{\theta}_0) -\widetilde{\boldsymbol{\nabla}}_{t}(u,\boldsymbol{\theta}_0))\widetilde{\boldsymbol{\nabla}}_t^{\intercal}(u,\boldsymbol{\theta}_0)}_1. \nonumber
\end{align*}
By the Cauchy-Schwarz inequality and Lemma 4, formula (\ref{eq:boundedh1deriv}),
\begin{align*}
 \norm{\boldsymbol{\nabla}_{t,T}(\boldsymbol{\theta}_0)(\boldsymbol{\nabla}_{t,T}^{\intercal}(\boldsymbol{\theta}_0)-\widetilde{\boldsymbol{\nabla}}_{t}^{\intercal}(u,\boldsymbol{\theta}_0))}_1
&\leq
\norm{\boldsymbol{\nabla}_{t,T}(\boldsymbol{\theta}_0)}_2 
\norm{\boldsymbol{\nabla}_{t,T}^{\intercal}(\boldsymbol{\theta}_0)-\widetilde{\boldsymbol{\nabla}}_{t}^{\intercal}(u,\boldsymbol{\theta}_0)}_2
 \\ 
&\leq
  C_{1} \norm{\boldsymbol{\nabla}_{t,T}^{\intercal}(\boldsymbol{\theta}_0)-\widetilde{\boldsymbol{\nabla}}_t^{\intercal}(u,\boldsymbol{\theta}_0)}_2. 
\end{align*}
The functions $\boldsymbol{\nabla}_{t,T}(\boldsymbol{\theta}_0)$ and $\widetilde{\boldsymbol{\nabla}}_t(u,\boldsymbol{\theta}_0)$ are of the type appearing in (\ref{eq:score_decomposition}), so we can conclude that
\begin{equation*}	
\norm{\mathbf{H}_{t,T}(\boldsymbol{\theta}_0) - \widetilde{\mathbf{H}}_t(u,\boldsymbol{\theta}_0)}_1 
\leq
C\left(\frac{1}{T} + |t/T-u| \right).
\end{equation*}
Applying Theorem A.1 component-wise to the Hessian yields
$$-\frac{1}{T} \sum_{t=1}^{\intercal} \frac{\partial^{2} l_{t,T}\left(\boldsymbol{\theta}_{0}\right)}{\partial \boldsymbol{\theta} \partial \boldsymbol{\theta}^{\intercal}} \overset{P}{\to}\int^1_0\mathbb{E} \left(\widetilde{\mathbf{H}}_{0}(u,\boldsymbol{\theta}_0) \right) \text{d}u:= 
\bold{B},$$ 
where $\bold{B}$ is nonrandom. Finally, using (\ref{eq:TaylorScore}) and anticipating the proof of (N3), we obtain
\begin{equation}
\frac{1}{T} \sum_{t=1}^{\intercal} \frac{\partial l_{t,T}\left(\boldsymbol{\theta}_{0}\right)}{\partial \boldsymbol{\theta}}
= -\left(\bold{B} + \bold{o}_{P}(1) \right)(\widehat{\boldsymbol{\theta}}_T - \boldsymbol{\theta}_0). \label{eq:A_nonsing} 
\end{equation}
By (N1), $\lim_{T \to \infty}T^{-1}\text{Cov}\left( \sum_{t=1}^{T} \partial l_{t,T}\left(\boldsymbol{\theta}_{0}\right)/\partial \boldsymbol{\theta} \right)=\bold{A}$, and $\bold{A}$ is nonsingular. Therefore (\ref{eq:A_nonsing}) implies nonsingularity of $\bold{B}$. This completes the verification of (N2).

To establish (N3) we show that the third derivatives of the log-likelihood function are uniformly integrable. If they are uniformly integrable, they are also integrable in a neighbourhood of $\boldsymbol{\theta}_0$, so that (N3) is satisfied. The third derivatives of the log-likelihood function with respect to the parameters $\theta_{i}$, $\theta_{j}$ and $\theta_{k}$ are given by 
\begin{align}
\frac{\partial^{3} l_{t,T}(\boldsymbol{\theta})}{\partial \theta_{i} \partial \theta_{j} \partial \theta_{k}}&=\left\{1-\frac{X^2_{t,T}}{h_{t,T}(\boldsymbol{\theta})}\right\}\left\{\frac{1}{h_{t,T}(\boldsymbol{\theta})} \frac{\partial^{3} h_{t,T}(\boldsymbol{\theta})}{\partial \theta_{i} \partial \theta_{j} \partial \theta_{k}}\right\} \label{eq:3deriv} \nonumber \\
&+\left\{2 \frac{X^2_{t,T}}{h_{t,T}(\boldsymbol{\theta})}-1\right\}\left\{\frac{1}{h_{t,T}(\boldsymbol{\theta})} \frac{\partial h_{t,T}(\boldsymbol{\theta})}{\partial \theta_{i}}\right\}\left\{\frac{1}{h_{t,T}(\boldsymbol{\theta})} \frac{\partial^{2} h_{t,T}(\boldsymbol{\theta})}{\partial \theta_{j} \partial \theta_{k}}\right\} \nonumber \\
&+\left\{2 \frac{X^2_{t,T}}{h_{t,T}(\boldsymbol{\theta})}-1\right\}\left\{\frac{1}{h_{t,T}(\boldsymbol{\theta})} \frac{\partial h_{t,T}(\boldsymbol{\theta})}{\partial \theta_{j}}\right\}\left\{\frac{1}{h_{t,T}(\boldsymbol{\theta})} \frac{\partial^{2} h_{t,T}(\boldsymbol{\theta})}{\partial \theta_{i} \partial \theta_{k}}\right\} \nonumber \\
&+\left\{2 \frac{X^2_{t,T}}{h_{t,T}(\boldsymbol{\theta})}-1\right\}\left\{\frac{1}{h_{t,T}(\boldsymbol{\theta})} \frac{\partial h_{t,T}(\boldsymbol{\theta})}{\partial \theta_{k}}\right\}\left\{\frac{1}{h_{t,T}(\boldsymbol{\theta})} \frac{\partial^{2} h_{t,T}(\boldsymbol{\theta})}{\partial \theta_{i} \partial \theta_{j}}\right\} \nonumber \\
&+\left\{2-6 \frac{X^2_{t,T}}{h_{t,T}(\boldsymbol{\theta})}\right\}\left\{\frac{1}{h_{t,T}(\boldsymbol{\theta})} \frac{\partial h_{t,T}(\boldsymbol{\theta})}{\partial \theta_{i}}\right\}\left\{\frac{1}{h_{t,T}(\boldsymbol{\theta})} \frac{\partial h_{t,T}(\boldsymbol{\theta})}{\partial \theta_{j}}\right\}\left\{\frac{1}{h_{t,T}(\boldsymbol{\theta})} \frac{\partial h_{t,T}(\boldsymbol{\theta})}{\partial \theta_{k}}\right\} .
\end{align}
By repeated use of the Cauchy-Schwarz inequality, \textcite{FrancqZakoian2004} proved that in the strictly stationary case, (N3) is true in some neighbourhood of $\boldsymbol{\theta}_0$. They used weaker assumptions than ours. The restriction to \emph{some} neighbourhood, rather than all of $\Theta$, is necessary in their case because the fraction $X^2_{t,T}/h_{t,T}(\boldsymbol{\theta})$ is not uniformly integrable under their assumptions but \emph{is} integrable in some neighbourhood of $\boldsymbol{\theta}_0$.  
The fraction $X^2_{t,T}/h_{t,T}(\boldsymbol{\theta})$ in (\ref{eq:3deriv}) is uniformly integrable, and the factors involving derivatives of $h_{t,T}(\boldsymbol{\theta})$ with respect to the parameters admit moments of any order by Lemma 4. By independence of $\varepsilon_{t}$ and $h_{t,T}(\boldsymbol{\theta})$, (A6) and (\ref{eq:boundedlsratio}), $$\norm{\sup_{\boldsymbol{\theta}\in\Theta}\frac{X^2_{t,T}}{h_{t,T}(\boldsymbol{\theta})}}_2=\norm{\varepsilon_0^2}_2\norm{\sup_{\boldsymbol{\theta}\in\Theta}\frac{h_{t,T}(\boldsymbol{\theta}_0)}{h_{t,T}(\boldsymbol{\theta})}}_2<\infty.$$ The proof now proceeds similarly to \textcite{FrancqZakoian2004}. Applying the Cauchy-Schwarz and Hölder inequalities to the terms involving the derivatives of $h_{t,T}(\boldsymbol{\theta})$ with respect to the parameters gives
\begin{align}&\norm{\sup_{\boldsymbol{\theta}\in\Theta}\left(\left\{1-\frac{X^2_{t,T}}{h_{t,T}(\boldsymbol{\theta})}\right\}\left\{\frac{1}{h_{t,T}(\boldsymbol{\theta})} \frac{\partial^{3} h_{t,T}(\boldsymbol{\theta})}{\partial \theta_{i} \partial \theta_{j} \partial \theta_{k}}\right\}\right)}_1 \nonumber \\ \leq  &\norm{\sup_{\boldsymbol{\theta}\in\Theta}\left\{1-\frac{X^2_{t,T}}{h_{t,T}(\boldsymbol{\theta})}\right\}}_2\norm{\sup_{\boldsymbol{\theta}\in\Theta}\left\{\frac{1}{h_{t,T}(\boldsymbol{\theta})} \frac{\partial^{3} h_{t,T}(\boldsymbol{\theta})}{\partial \theta_{i} \partial \theta_{j} \partial \theta_{k}}\right\}}_2
<\infty \label{eq:fin3deriv1}\end{align}
and
\begin{align}&\norm{\sup_{\boldsymbol{\theta}\in\Theta}\left(\left\{2 \frac{X^2_{t,T}}{h_{t,T}(\boldsymbol{\theta})}-1\right\}\left\{\frac{1}{h_{t,T}(\boldsymbol{\theta})} \frac{\partial h_{t,T}(\boldsymbol{\theta})}{\partial \theta_{i}}\right\}\left\{\frac{1}{h_{t,T}(\boldsymbol{\theta})} \frac{\partial^{2} h_{t,T}(\boldsymbol{\theta})}{\partial \theta_{j} \partial \theta_{k}}\right\}\right)}_1 \nonumber \\ 
\leq&\norm{\sup_{\boldsymbol{\theta}\in\Theta}\left\{2 \frac{X^2_{t,T}}{h_{t,T}(\boldsymbol{\theta})}-1\right\}}_2 \norm{\sup_{\boldsymbol{\theta}\in\Theta}\left\{\frac{1}{h_{t,T}(\boldsymbol{\theta})} \frac{\partial h_{t,T}(\boldsymbol{\theta})}{\partial \theta_{i}}\right\}}_p\norm{\sup_{\boldsymbol{\theta}\in\Theta}\left\{\frac{1}{h_{t,T}(\boldsymbol{\theta})} \frac{\partial^{2} h_{t,T}(\boldsymbol{\theta})}{\partial \theta_{j} \partial \theta_{k}}\right\}}_q<\infty \label{eq:fin3deriv2}
\end{align}
for some $1/p+1/q=1/2$, and similarly for the two subsequent terms. Following \textcite{FrancqZakoian2004}, we can deal with the last term on the right-hand side of (99) by writing \begin{align}&\norm{\sup_{\boldsymbol{\theta}\in\Theta}\left(\left\{2-6 \frac{X^2_{t,T}}{h_{t,T}(\boldsymbol{\theta})}\right\}\left\{\frac{1}{h_{t,T}(\boldsymbol{\theta})} \frac{\partial h_{t,T}(\boldsymbol{\theta})}{\partial \theta_{i}}\right\}\left\{\frac{1}{h_{t,T}(\boldsymbol{\theta})} \frac{\partial h_{t,T}(\boldsymbol{\theta})}{\partial \theta_{j}}\right\}\left\{\frac{1}{h_{t,T}(\boldsymbol{\theta})} \frac{\partial h_{t,T}(\boldsymbol{\theta})}{\partial \theta_{k}}\right\}\right)}_1 \nonumber \\
\leq&\norm{\sup_{\boldsymbol{\theta}\in\Theta}\left\{2-6 \frac{X^2_{t,T}}{h_{t,T}(\boldsymbol{\theta})}\right\}}_2\max_i\norm{\sup_{\boldsymbol{\theta}\in\Theta}\left\{\frac{1}{h_{t,T}(\boldsymbol{\theta})} \frac{\partial h_{t,T}(\boldsymbol{\theta})}{\partial \theta_{i}}\right\}}_6^3<\infty. \label{eq:fin3deriv3}
\end{align}
Combining (\ref{eq:fin3deriv1}), (\ref{eq:fin3deriv2}) and (\ref{eq:fin3deriv3}) verifies (N3), which in turn completes the proof of Theorem 2.  $\qedsymbol{}$

\subsection*{A.6 The truncated estimator}
We show that consistency and asymptotic normality of the QMLE 
$\widehat{\boldsymbol{\theta} }_{T}$ remain true for the truncated QMLE $\bar{\boldsymbol{\theta} }_{T}$. By Lemma 3, the truncated part of the recursion of $h_{t,T}(\boldsymbol{\theta})$ in 
(\ref{eq:decrecht}) is bounded by a stationary process. Because the rate of convergence of $\widehat{\boldsymbol{\theta} }_{T}$ in the ATV-GARCH model is the same as in stationary GARCH models, the situation is entirely analogous to that in the proofs of Theorems 4.3 and 4.4 in BHK, which we follow. We thus show that $\underset{\boldsymbol{\theta}\in\varTheta}{\sup} \left|L_T(\boldsymbol{\theta})-\bar{L}_T(\boldsymbol{\theta})\right| = o(1)$
almost surely. By the triangle inequality,
\begin{align}
		\underset{\boldsymbol{\theta}\in\varTheta}{\sup} \left|L_T(\boldsymbol{\theta})-\bar{L}_T(\boldsymbol{\theta})\right|
		\leq \underset{\boldsymbol{\theta}\in\varTheta}{\sup} \frac{1}{T}\sum^{T}_{t=1}\left|\ln h_{t,T}(\boldsymbol{\theta}) - \ln \bar{h}_{t,T}(\boldsymbol{\theta})\right|
		+ \underset{\boldsymbol{\theta}\in\varTheta}{\sup} \frac{1}{T}\sum^{T}_{t=1}\left|\frac{X_{t,T}^2}{h_{t,T}(\boldsymbol{\theta})}-\frac{X_{t,T}^2}{\bar{h}_{t,T}(\boldsymbol{\theta})}\right|.
		\label{eq:trianglell}
	\end{align}
Consider the first term on the right-hand side of (\ref{eq:trianglell}). In view of (\ref{eq:loglikh}),
\begin{align}
&\underset{\boldsymbol{\theta}\in\varTheta}{\sup}\frac{1}{T}\sum^{T}_{t=1}\left|\ln h_{t,T}(\boldsymbol{\theta}) - \ln \bar{h}_{t,T}(\boldsymbol{\theta})\right| \nonumber \\
&\leq \underset{\boldsymbol{\theta}\in\varTheta}{\sup}\frac{1}{T}\sum^{T}_{t=1}\frac{1}{C_1}\left|h_{t,T}(\boldsymbol{\theta}) - \bar{h}_{t,T}(\boldsymbol{\theta})\right| \nonumber \\
&\leq\underset{\boldsymbol{\theta}\in\varTheta}{\sup} \frac{1}{T}\sum^{T}_{t=1}\frac{1}{C_1}\left|\sum^{\infty}_{i=1}c_i(\boldsymbol{\theta})X^2_{t-i,T}-\sum^{t-1}_{i=1}c_i(\boldsymbol{\theta})X^2_{t-i,T}\right| \nonumber \\
&+
\underset{\boldsymbol{\theta}\in\varTheta}{\sup} \frac{1}{T}\sum^{T}_{t=1}\frac{1}{C_1}\left|\sum^{\infty}_{i=1}d_i(\boldsymbol{\theta})g_{t-i+1,T}(\boldsymbol{\theta}_1)-\sum^{t-1}_{i=1}d_i(\boldsymbol{\theta})g_{t-i+1,T}(\boldsymbol{\theta}_1)\right| \nonumber \\
&\leq \frac{1}{T}\sum^{T}_{t=1}\frac{C_2}{C_1}\sum^{\infty}_{i={t}}\rho_0^{i/q}X^2_{t-i,T} 
+
\frac{1}{T}\sum^{T}_{t=1}\frac{C_2}{C_1}\sum^{\infty}_{i={t}}\rho_0^{i/q}\underset{\boldsymbol{\theta}\in\varTheta}{\sup}g_{t-i+1,T}(\boldsymbol{\theta}_1) \nonumber \\
&= \frac{1}{T}\sum^{T}_{t=1}\frac{C_2}{C_1}\rho_0^{t/q}\sum^{\infty}_{j=0}\rho_0^{j/q}X^2_{-j,T} 
+
\frac{1}{T}\sum^{T}_{t=1}\frac{C_2}{C_1}\rho_0^{t/q}\sum^{\infty}_{j=0}\rho_0^{j/q}\underset{\boldsymbol{\theta}\in\varTheta}{\sup}g_{-j + 1,T}(\boldsymbol{\theta}_1).
\label{eq:ltruncdiff2}
\end{align}
Note that the stochastic sum $\sum^{\infty}_{j=0}\rho_0^{i/q}X^2_{-i,T}$ and the deterministic sum $\sum^{\infty}_{j=0}\rho_0^{i/q}g_{-i + 1,T}(\boldsymbol{\theta}_1)$ both start at $j = 0$. The stochastic one is smaller than the sum of the stationary process in Lemma 3 obtained by fixing the intercept of $X^2_{-j,T}$ at $\sup_{u<0} g(u)$. By Lemmas 2.2 and 2.3 of BHK, it is convergent with probability one. Hence it is $O(1)$ almost surely. The deterministic sum is also $O(1)$ for any $\boldsymbol{\theta}\in\varTheta$. It follows that the terms in (\ref{eq:ltruncdiff2}) are $o(1)$. Setting $C=C_2/C_1$ and applying the same results (\ref{eq:ltruncdiff2}) to the second term on the right-hand side of (\ref{eq:trianglell}), we obtain 
\begin{align}
		&\underset{\boldsymbol{\theta}\in\varTheta}{\sup}\frac{1}{T}\sum^{T}_{t=1}\left|\frac{X_{t,T}^2}{h_{t,T}(\boldsymbol{\theta})}-\frac{X_{t,T}^2}{\bar{h}_{t,T}(\boldsymbol{\theta})}\right| \nonumber \\ &= \underset{\boldsymbol{\theta}\in\varTheta}{\sup}\frac{1}{T}\sum^{T}_{t=1}\frac{X_{t,T}^2}{h_{t,T}(\boldsymbol{\theta})}\left|\frac{h_{t,T}(\boldsymbol{\theta})-\bar{h}_{t,T}(\boldsymbol{\theta})}{\bar{h}_{t,T}(\boldsymbol{\theta})}\right| \nonumber \\ &\leq \frac{C}{T}\left(\sum_{j=0}^{\infty}\rho_0^{j/q}X^2_{-j,T}
  +
  \sum_{j=0}^{\infty}\rho_0^{j/q}\underset{\boldsymbol{\theta}\in\varTheta}{\sup}g_{-j + 1,T}(\boldsymbol{\theta}_1)\right) \underset{\boldsymbol{\theta}\in\varTheta}
    {\sup}\sum_{t=1}^T\frac{X_{t,T}^2}{h_{t,T}(\boldsymbol{\theta})}\rho_0^{t/q} \nonumber \\
		&\leq
  \frac{C}{T}\left(\sum_{j=0}^{\infty}\rho_0^{j/q}X^2_{-j,T}
  +
  \sum_{j=0}^{\infty}\rho_0^{j/q}\underset{\boldsymbol{\theta}\in\varTheta}{\sup}g_{-j + 1,T}(\boldsymbol{\theta}_1)\right) 
  \underset{\boldsymbol{\theta}\in\varTheta}{\sup}\sum_{t=1}^T\frac{X_{t,T}^{*2}}{\widetilde{h}_{t,T}(\boldsymbol{\theta})}\rho_0^{t/q}.
		\label{eq:ltruncdiff3}
	\end{align}
The sequence $\{X_{t,T}^2/h_{t,T}(\boldsymbol{\theta})\}$ is bounded by the stationary process $\{X_{t,T}^{*2}/\widetilde{h}_{t,T}(\boldsymbol{\theta})\}$ in Lemma 3. By Lemmas 2.2 and 5.1 of BHK,
\begin{align*}
\underset{\boldsymbol{\theta}\in\varTheta}{\sup}\sum_{t=1}^T\frac{X_t^{*2}}{\widetilde{h}_t(\boldsymbol{\theta})}\rho_0^{t/q} < \infty 
\end{align*}
almost surely. Thus the second term of on the right-hand side of (\ref{eq:trianglell}) is $o(1)$ almost surely. Consistency of $\bar{\boldsymbol{\theta}}_T$ now follows from consistency of $\widehat{\boldsymbol{\theta}}_T$ in Theorem 1.

Now we turn to the limiting distribution of $\sqrt{T} (\bar{\boldsymbol{\theta}}_{T} - \boldsymbol{\theta}_0)$. To this end, by modifying the proofs of equations (5.34) and (5.37) in BHK, we readily find that 
\begin{equation}
    \underset{\boldsymbol{\theta}\in\boldsymbol{\varTheta}}{\sup} \left|\mathbf{s}_T(\boldsymbol{\theta})-\bar{\mathbf{s}}_T(\boldsymbol{\theta})\right|
    = O\left(\frac{1}{T}\right)
\end{equation}
almost surely. Since both $\mathbf{s}_T(\widehat{\boldsymbol{\theta}}_T)=\mathbf{0}$ and $\bar{\mathbf{s}}_T(\bar{\boldsymbol{\theta}}_T)=\mathbf{0}$, we have $\mathbf{0} = s_T(\widehat{\boldsymbol{\theta}}) - \bar{s}_T(\bar{\boldsymbol{\theta}})$. We can replace $\bar{\mathbf{s}}_T(\bar{\boldsymbol{\theta}}_T)$ by $\mathbf{s}_T(\bar{\boldsymbol{\theta}}_T)$. The error is
\begin{equation*}
\boldsymbol{s}_T(\widehat{\boldsymbol{\theta}}) - \mathbf{s}_T(\bar{\boldsymbol{\theta}}_T)
= \boldsymbol{s}_T(\widehat{\boldsymbol{\theta}}) -
\bar{\mathbf{s}}_T(\bar{\boldsymbol{\theta}}_T)
+ \boldsymbol{O}\left(\frac{1}{T} \right) = \boldsymbol{O}\left(\frac{1}{T} \right)
\end{equation*}
almost surely. As in the proof of Theorem 4.4 of BHK, we linearize the difference $\widehat{\boldsymbol{\theta} }_{T}-\bar{\boldsymbol{\theta} }_{T}$ by an
application of the mean value theorem. This gives
\begin{equation}
\mathbf{s}_T(\widehat{\boldsymbol{\theta} }_{T})-\mathbf{s}_T(\bar{\boldsymbol{\theta} }_{T})=
\mathbf{H}_{T}(\boldsymbol{\theta} ^{\ast}) 
(\widehat{\boldsymbol{\theta} }_{T}-
\bar{\boldsymbol{\theta} }_{T})+\mathbf{O}\left( \frac{1}{T}\right)\label{eq:truncscorediff}
\end{equation}
almost surely, where $\boldsymbol{\theta}^{\ast}$ lies on the line segment connecting $\widehat{\boldsymbol{\theta} }_{T}$ and $\bar{\boldsymbol{\theta} }_{T}$. (N2) and (N3) together imply $\mathbf{H}_{T}(\boldsymbol{\theta} ^{\ast})\overset{P}{\rightarrow }\mathbf{H}(\boldsymbol{\theta}_{0})$.
We can therefore write
\begin{equation*}
\mathbf{s}_T(\widehat{\boldsymbol{\theta}}_{T})-\mathbf{s}_T(\bar{\boldsymbol{\theta}}_{T})=\mathbf{H}(\boldsymbol{\theta}_{0})(1+\mathbf{o}_{P}(1))(\widehat{\boldsymbol{\theta} }_{T}
- \bar{\boldsymbol{\theta}}_{T}),
\end{equation*}
which implies
\begin{equation*}
|\widehat{\boldsymbol{\theta}}_{T}-\bar{\boldsymbol{\theta}}_{T}|=\mathbf{O}_{P}\left(\frac{1}{T}%
\right).
\end{equation*}
The limiting distribution of $\sqrt{T} (\bar{\boldsymbol{\theta}}_{T} - \boldsymbol{\theta}_0)$ is therefore identical to the limiting distribution of $\sqrt{T} (\widehat{\boldsymbol{\theta}}_{T} - \boldsymbol{\theta}_0)$ in Theorem 2. Hence it is proved that asymptotic normality of $\bar{\boldsymbol{\theta} }_{T}$ follows from asymptotic normality of 
$\widehat{\boldsymbol{\theta} }_{T}$ in Theorem 2.

\clearpage
\printbibliography
\end{document}